\documentclass[12pt]{article}

\usepackage[margin=1.25in]{geometry}
\usepackage{amsmath,amssymb,amsthm,mathtools,mathrsfs}
\usepackage{enumitem}
\usepackage{natbib}
\usepackage[colorlinks=true,linkcolor=blue,citecolor=blue,urlcolor=blue]{hyperref}
\usepackage{setspace}
\usepackage{graphicx}
\usepackage{booktabs}
\usepackage[T1]{fontenc}
\usepackage{comment}

\numberwithin{equation}{section}
\newtheorem{theorem}{Theorem}[section]
\newtheorem{proposition}[theorem]{Proposition}
\newtheorem{lemma}[theorem]{Lemma}
\newtheorem{corollary}[theorem]{Corollary}
\newtheorem{definition}[theorem]{Definition}
\newtheorem{assumption}[theorem]{Assumption}

\newcommand{\E}{\mathbb{E}}
\newcommand{\R}{\mathbb{R}}
\newcommand{\Var}{\operatorname{Var}}
\newcommand{\Cov}{\operatorname{Cov}}

\newcommand{\Int}{\operatorname{int}}

\title{\vspace{-1.5cm}\textbf{Private Languages}}
\author{Jeremy Bertomeu\thanks{Olin Business School, Washington University in St.\ Louis, One Brookings Drive, St.\ Louis, MO 63130. Email: \texttt{bjeremy@wustl.edu}.}}
\date{}

\begin{document}

\maketitle

\begin{abstract}
\noindent
Strategic communication often relies on anchors observed by the sender but not by the receiver. An analyst may report against a proprietary valuation model, an auditor against an internal score, a manager against an accounting estimate, or an institution against its own standard. I study a sender-receiver game in which reports are costly to move away from such privately observed anchors. Anchor heterogeneity changes the geometry of communication. Rather than relying on partitions, privately anchored reporting generates continuous variation in messages because different senders find different reports costly to make. This mechanism can improve information transmission, but it can also pull reports toward noisy private anchors.
I show that (i) small positive reporting costs can make communication approach full revelation, even though zero costs return the model to cheap talk, (ii) uninformative anchors can transmit information through strategic distortions. Anchored reports and cheap-talk messages can coexist as endogenous hard and soft information, but cheap-talk alone is preferred by all parties under sufficiently low misalignment, explaining why organizations may rely exclusively on informal channels. \end{abstract}

\noindent \textbf{Keywords:} private languages, strategic communication, lying costs, language inflation, compliance and informativeness.

\noindent \textbf{JEL Codes:} D82, D83, C72.

\newpage

\section{Introduction}

A financial analyst may be asked to justify a price target with a proprietary valuation model.  An auditor may assign a risk rating using an internal checklist.  A manager may report a number that is disciplined by an accounting estimate. A student receives a grade in a course. In each case, the receiver observes a public report, while the anchor used to judge whether the report is faithful is private to the sender.  Verification is therefore relative to the anchor: the report is not simply compared with the objective state, but with a model, score, estimate, or internal standard that the sender privately observes and may understand better than the receiver.

This distinction has implications for strategic communication. A sender may disagree with the receiver about the relevant representation of the state, or may face an ethical or institutional cost of departing from an internal standard rather than from the objective state itself. The same objective state can then generate different locally cheap reports across senders. In some settings, the sender may also have influenced the anchor before reporting, as when a forecaster chooses a model specification,  or an auditor selects inputs that determine a score. The paper takes the anchor as privately observed at the reporting stage and examines how the private language affects equilibrium communication.

Various practical design question result from these considerations. Should institutions tie reports to a standardized anchor, such as accounting standards, audit checklists, bank stress tests, or structured grading rubrics, or leave the channel to qualitative discretion without a fixed reference? Schools sometimes adopt grade non-disclosure to push communication toward interviews and recommendations rather than a public quantitative grade \citep{FloydTomarLee2024}, and standard setters periodically debate whether to tighten or relax mandated reporting categories. The model below yields a simple answer: both the cost to impose to enforce the anchor and which arrangement is preferred depends on how aligned the sender and receiver are.

The model is a sender-receiver game in which the sender observes both the payoff-relevant state and a private anchor. The receiver observes a report and chooses an action. The sender cares about the receiver's action, but also bears a direct cost when the report departs from the anchor. The anchor is the report the sender would make in the absence of strategic incentives. It may be informative about the state, uninformative, or only partially connected to the state. This formulation differs both from cheap talk, where all messages are costless, and from deterministic lying-cost models, where the anchor for truthful reporting is the state or another commonly known object.

The first analytical step is geometric. When the receiver's action varies continuously with the report, the sender's optimal report varies continuously with the private anchor, so a fixed state generates a range of reports rather than a single one or a finite set of messages. Cross-sectional variation in anchors fills the report space in a way that deterministic lying-cost models do not: pooling can still arise at the boundary of a compact message space, but the interior communication region has no unused gaps. Specializing to additive anchors and translation-invariant reporting costs turns this geometric statement into a tractable condition. The sender's first-order condition links each report to the anchor realization that makes it optimal, and Bayes' rule then determines the receiver's action from the distribution of states and anchors consistent with that report. With this characterization in hand, we prove three economic insights.

First, there is a discontinuity at zero reporting costs. At exactly zero cost, the model collapses to cheap talk and informative communication takes the form of partitions, while at any positive cost an equilibrium exists along which the receiver's posterior approaches full revelation. Such convergence is not free however. The sender uses increasingly inflated reports, which,
by increasing the span of the effective messages as a function of the state, filters the noise from the anchor.
The expected reporting cost converges to a cumulative incentive term equal to the sender's marginal gain from moving the receiver's action. When the state space is unbounded, the resulting cost to the sender of message inflation diverges, even though the report is asymptotically fully informative.

Second,  even uninformative anchors can transmit information through strategic distortions alone. Different states have different marginal gains from moving the receiver and therefore distort by different amounts from the same anchor. Further, an explicit design of the anchor cost and noise drives the receiver's posterior variance toward zero.
The sender's 
expected welfare may fall below babbling because the communication is achieved from message inflation.

Third, quantitative and qualitative reporting channels can coexist in the same equilibrium. Treating a qualitative label as cheap talk and an anchored number as the within-label refinement, the labeled report drives residual uncertainty to zero as alignment becomes perfect, while pure anchored communication retains a fixed loss from anchor noise. At moderate misalignment, this hybrid communication strictly dominates either pure cheap talk or pure anchors. At low misalignment, by contrast, pure cheap talk is strictly preferred to the hybrid, because the cheap-talk partition is already fine and the anchor-induced cutoff shifts coarsen it more than the within-label refinement helps.

The analysis contributes to the literature on strategic communication with cheap talk and lying costs.  In \citet{CrawfordSobel1982}, costless messages transmit information through partitions.  \citet{Kartik2009} shows that lying costs can support more informative equilibria, including equilibria with bounded messages featuring low-type separation and high-type pooling. \citet{KartikOttavianiSquintani2007} study costly talk when receivers may be credulous.  \citet{GuttmanKadanKandel06} develop an accounting-reporting model in which a translation-invariant cost of departing from a fixed reference generates kinks at salient reporting thresholds. \citet{cas26} analyzes bounded penalties and shows that this environment can imply pooling over low-types and, under full support, requires binary messages.
In the standard lying-cost setup, the report's cost reference is the state itself, so reporting costs are deterministic given $\theta$. The present setting adds an idiosyncratic anchor, distinct from $\theta$, that shifts the local cost of reports. This generates cross-sectional variation in reporting costs conditional on a fixed state and produces interior continuity rather than the boundary kinks typical of fixed-reference models.

The contribution also links to equilibrium selection in cheap talk. \citet{ChenKartikSobel2008} and \citet{Chen2011} refine cheap-talk communication through perturbations of preferences or information. \citet{LipnowskiRavid2020} characterize cheap-talk equilibrium payoffs under transparent state-independent sender motives. In this model, selection comes from a privately observed anchor that changes the sender's reporting costs.  Multidimensional communication models, including \citet{Battaglini2002}, \citet{LevyRazin2007}, and \citet{ChakrabortyHarbaugh2010}, show that additional dimensions can alter credible communication.  \citet{FrankelKartik2019} study muddled information generated by a sender's privately observed ability to game the signal.  

Finally, the accounting and disclosure literature provides a natural applied motivation.  \citet{FischerVerrecchia2000} and \citet{DyeSridhar2004,DyeSridhar2007} study reporting systems in which noisy accounting signals and managerial reports interact.  However, the channel studied here is different.  Internal anchors determine which public reports are easy to defend, so stronger compliance with the anchor can transmit anchor noise rather than objective information.  This mechanism also applies to institutional templates, scorecards, ratings, and forecasts, where formal verification and informal persuasion coexist \citep{Dessein2002,MorganStocken2003,OttavianiSorensen2006}.

\section{The Model}\label{sec:model}

\subsection{Environment and preferences}

There are two players, a sender and a receiver. The payoff-relevant state is $\theta\in\Theta=[\underline\theta,\overline\theta]\subset\R$, where $\Theta$ is compact and the prior density $f$ is continuous and strictly positive. The sender observes $\theta$ and a private anchor $b\in\R$. The anchor is an internal reporting reference: an accounting estimate, a verbal label, an internal model output, a checklist score, or any private signal that affects which reports are directly costly. The relation between states and anchors is the sender's private language. Because this language need not coincide with the objective state space, the anchor may differ from the state even when the sender is not attempting to affect the action. Conditional on $\theta$, the anchor has density $q(\cdot\mid\theta)$ supported on an interval $B(\theta)\subseteq\R$.

The sender chooses a public report $r\in M$, where the message space $M$ is either $\R$ or a compact interval. After observing $r$, the receiver chooses an action $a\in\R$. The receiver utility is $U^R(a,\theta)$, and the sender utility is $U^S(a,\theta)-D(r,b)$. The first term captures the sender's motive to influence the receiver's action. The second term is the direct cost of reporting away from the sender's private anchor. I refer to $D$ as a reporting cost. Depending on the application, this cost can represent lying costs, ethical costs, institutional discipline, legal exposure, or verification costs measured relative to a private anchor.

\begin{assumption}\label{ass:receiver-main}\label{ass:sender-main}
The receiver payoff $U^R$ is $C^4$, strictly concave in $a$, and satisfies $U^R_{11}(a,\theta)\le -\eta<0$ and $U^R_{12}(a,\theta)\ge \underline\gamma>0$ on the relevant domain.  For each $\theta$, the equation $U^R_1(a,\theta)=0$ has a unique solution $a^R(\theta)$.  The sender action payoff $U^S$ is $C^4$ on $A\times\Theta$, where $A=[a^R(\underline\theta),a^R(\overline\theta)]$, and satisfies $U^S_{11}(a,\theta)\le-\eta_S<0$ and $U^S_{12}(a,\theta)\ge\underline\gamma_S>0$ on $A\times\Theta$.
\end{assumption}

Strict concavity yields a unique receiver best reply for every posterior.  The cross-partial conditions imply that, for both players, higher states make higher actions more desirable.  If $U^S(\cdot,\theta)$ has a unique maximizer, denote $a^S(\theta)$ as the sender's ideal action and $\delta(\theta)=a^S(\theta)-a^R(\theta)$ as the bias.  
\begin{assumption}\label{ass:cost-main}
The cost $D$ is $C^4$ and satisfies (i) $D_{11}(r,b)>0$ for all $(r,b)$, (ii) $D_{12}(r,b)<0$ for all $(r,b)$, (iii) for each $b$, the function $r\mapsto D(r,b)$ is uniquely minimized at $r=b$, (iv) for every $L<\infty$ there exists $K(L)<\infty$ such that $|r-b|\ge K(L)$ implies $D(r,b)-D(b,b)\ge L$ uniformly in $b$, and (v) for every $K<\infty$, $\inf_{|r-b|\le K}D_{11}(r,b)>0$.
\end{assumption}

Condition (i) implies a unique cheapest report for each anchor.  Condition (ii) is the Spence-Mirrlees single-crossing property, in the sense that a higher anchor makes higher reports relatively cheaper.\footnote{The additive-anchor formulation and the translation-invariant cost $c\phi(r-b)$ used below are convenient but not essential. The no-holes result of Section~\ref{sec:noholes} extends to any private shifter $z$ whose marginal effect on the sender's report problem satisfies single-crossing in $(r,z)$, whether $z$ enters through the cost or through the strategic payoff. Appendix~\ref{app:supp-shifters} formulates the general statement and notes that it covers, among other cases, the \citet{FischerVerrecchia2000} model, in which the private variable is a manager's taste for price and the cost depends on the state rather than on a separate anchor.}  Condition (iii) states that absent strategic motives, the sender reports the anchor.  Condition (iv) prevents unbounded distortions from being optimal when the strategic gain is bounded.  Condition (v) imposes a uniform curvature near the anchor.  Among other examples, the translation-invariant specification $D(r,b)=c\phi(r-b)$ with $\phi''>0$ (\citealt{GuttmanKadanKandel06}) satisfies these conditions and the bounded-claim environment in \citet{Kartik2009} can be recovered when the anchor is deterministic and $M$ is compact.

\subsection{Strategy and regular equilibrium}

A pure-strategy perfect Bayesian equilibrium consists of a measurable report rule $R:\Theta\times\R\to M$, a belief system $\mu(\cdot\mid r)\in\Delta(\Theta)$ for each report, and an action rule $a:M\to\R$ such that the sender optimizes given $a(\cdot)$, the receiver optimizes given beliefs, and beliefs are updated by Bayes' rule wherever possible.

The analysis is restricted to equilibria in which the receiver's response varies continuously and monotonically with the report, so that nearby reports have nearby consequences and the sender's first-order condition is informative.

\begin{definition}\label{def:regular-main}
A pure-strategy equilibrium $(R,a)$ is regular if (i) $a:M\to A$ is continuous and strictly increasing, with finitely many kink points and $a\in C^3$ away from those points, (ii) for every $(\theta,b)$, the sender's report problem has a unique maximizer $R(\theta,b)$, and (iii) whenever $R(\theta,b)$ lies in the interior of a smooth region of $a$, the sender's second-order condition is strict.
\end{definition}

A smooth region of $a$ is a connected component of $M$ after removing the finite kink set. Continuity is assumed with respect to all messages, including potentially off-equilibrium ones, so the definition does not explicitly rule out gaps in messages on the equilibrium path. Regular equilibria are a natural selection because they imply languages in which small changes in messages do not produce extreme changes in inference, and so are less susceptible to small errors in the receiver's processing of the message. In what follows, I assume that regular equilibria exist and defer conditions sufficient for existence to a supplementary appendix.

\section{The Geometry of Communication}\label{sec:noholes}

\subsection{Continuity and interval images}

Private anchors affect which reports can arise from a fixed state.  Since the state space is compact, any receiver best reply to a posterior supported on $\Theta$ lies in $A=[a^R(\underline\theta),a^R(\overline\theta)]$.  This bounded action range caps the sender's strategic gain from moving the receiver.  When the anchor is feasible, the sender can always report the anchor at minimum direct cost; on an unbounded message space, this fallback implies a uniform bound on equilibrium distortions away from the anchor.  

\begin{proposition}\label{prop:bounded-main}
Under Assumptions~\ref{ass:receiver-main}--\ref{ass:cost-main}, in any regular equilibrium and for each fixed $\theta$, the map $b\mapsto R(\theta,b)$ is nondecreasing and continuous on $B(\theta)$, and is strictly increasing on the set of anchors whose optimal reports lie in the interior of a smooth region of $a$.
\end{proposition}

Higher anchors cannot lead the sender to choose lower reports because, by single-crossing in the cost, an increase in the anchor makes high reports relatively cheaper than low ones. The sender's report also cannot jump as the anchor varies. On a compact message space this is immediate from continuity of the objective and uniqueness of the sender's optimum. On an unbounded message space the same conclusion holds once one notes that equilibrium reports remain within a bounded distance of the anchor on bounded sets of anchors: reporting the anchor itself costs nothing directly, while the strategic gain is capped because the receiver's action stays in the bounded interval $A$.

This continuity is the step that separates private-anchor communication from deterministic lying-cost models.  With a deterministic anchor, a fixed state is associated with a single cost-minimizing report.  With private anchors, the same state is associated with a continuum of cost-minimizing reports, indexed by the sender's privately observed anchor.  Since nearby anchors generate nearby optimal reports, the reports used at a fixed state move continuously as the private anchor varies.  
\begin{theorem}\label{thm:noholes-main}
For every $\theta\in\Theta$, the used-report set $R(\theta,B(\theta))=\{R(\theta,b):b\in B(\theta)\}$
is a connected interval.
\end{theorem}

The theorem implies a ``no-holes'' property.  Since $B(\theta)$ is an interval and $b\mapsto R(\theta,b)$ is continuous, the image of $B(\theta)$ must also be an interval.  The result thus identifies the source of continuous communication in the model.  Different senders at the same state find different reports locally cheap, and these local reporting incentives fill the report space between the extremes used at that state.  Boundary pooling may still arise on compact message spaces, but interior holes are not a feature of regular equilibria.

\begin{corollary}\label{cor:fulluse-unbounded}
Suppose $M=\R$ and $B(\theta)=\R$ for every $\theta$.  Then, in any regular equilibrium, $R(\theta,\R)=\R$ for every $\theta$.
\end{corollary}

When anchors have full real support, the no-holes result becomes a ``full-use'' result.  The distortion bound implies that reports move to $-\infty$ as anchors move to $-\infty$ and to $\infty$ as anchors move to $\infty$.  Since the used-report set is an interval, it cannot omit any report.  Every state therefore uses the entire message line, with different reports generated by different realizations of the private anchor.

\subsection{Compact message spaces}

The compact-message case offers a relevant comparison with lying-cost models restricted to bounded message spaces, such as  \citet{KartikOttavianiSquintani2007} and \citet{Kartik2009}. In these models, the upper bound on messages can generate equilibria with separation at low reports and pooling at the top. Private anchors do not remove the pooling at the boundary, because extreme anchors may still make the sender choose an endpoint report. However, in the interior, conditional on a fixed state, a continuum of privately observed anchors generates a continuum of locally cheap reports. Thus boundary pooling can remain, but interior holes disappear.

\begin{proposition}\label{prop:boundedmessages}
Suppose $M=[\underline m,\overline m]$ and $B(\theta)=\R$ for every $\theta$. In any regular equilibrium, for every $\theta$, the state-specific used-report set $I_\theta=R(\theta,\R)$ is a connected subinterval of $M$. The global on-path report set $\mathcal U=R(\Theta\times\R)$
is also a connected interval in $M$. 
\end{proposition}

The first part of the proposition is the state-by-state no-holes result. Since the anchor support is connected and the sender's report varies continuously with the anchor, a fixed state cannot use two separated regions of reports without also using the reports between them. The second part uses the same continuity in both arguments. With a compact message space and a unique sender optimum, the report rule is continuous in the state and the anchor. Since $\Theta\times\R$ is connected, its image under the equilibrium report rule is connected. The on-path report set may therefore fail to cover all of $M$, but any missing reports must lie at the bottom or the top of the message space, not in the interior of the on-path region. 
\begin{assumption}\label{ass:invert-main}
For every report $r$, the map $b\mapsto D_1(r,b)$ is continuous, strictly decreasing, and onto $\R$.
\end{assumption}
The message boundary may still force pooling at $\underline m$ or $\overline m$; the conclusion is stronger under the invertibility condition in Assumption~\ref{ass:invert-main}.

\begin{corollary}\label{cor:boundedmessages-surj}
Under the hypotheses of Proposition~\ref{prop:boundedmessages} and Assumption~\ref{ass:invert-main}, for any $\theta$, $(\underline m,\overline m)\subseteq I_\theta$. If, in addition, the equilibrium payoff $r\mapsto U^S(a(r),\theta)$ is Lipschitz on $M$ for each $\theta$, then $I_\theta=M$ for every $\theta$.
\end{corollary}

Assumption~\ref{ass:invert-main} ensures that extreme anchors have arbitrarily strong marginal effects on reporting costs. As $b$ becomes very low, reports near the lower boundary become cheaper than interior reports; as $b$ becomes very high, reports near the upper boundary become cheaper than interior reports. The report rule for each state therefore approaches both endpoints of the message space. Since the state-specific used-report set is an interval, every interior report must be used by some anchor realization.
Under a Lipschitz condition, once the anchor is sufficiently extreme, the direct reporting-cost advantage of an endpoint dominates any bounded marginal gain from moving the receiver's action, so the endpoint report is actually chosen rather than only approached as a limit.

When $M=\R$ and anchors have unbounded support, the no-holes result implies that every report can be generated at every state. Since the state space is compact, receiver actions remain in the bounded interval $A$. Hence the strategic payoff gain from changing the receiver's action is bounded. By contrast, moving far from an extreme anchor becomes increasingly costly. A sender with an extreme anchor therefore has little reason to choose a report far from that anchor.

\begin{corollary}\label{cor:reversion}
Under Assumptions~\ref{ass:receiver-main}-\ref{ass:cost-main}, suppose $M=\R$ and $a'$ is uniformly continuous on each tail. For any fixed $\theta\in\Theta$, if $(b_n)$ satisfies $|R(\theta,b_n)|\to\infty$, then $|R(\theta,b_n)-b_n|\to0$ .
\end{corollary}

The corollary states that extreme reports are generated by extreme anchors, not by unbounded strategic exaggeration away from moderate anchors. The compact state space bounds rationalizable receiver actions, so the sender's marginal benefit from additional distortion vanishes along the tails. The first-order condition then forces the marginal reporting cost to vanish. Since the cost is strictly convex and uniquely minimized at the anchor, the report must converge back to the anchor.

\subsection{Receiver posterior belief}

The no-holes result has a useful implication for equilibrium characterization. On an unbounded message space with full anchor support, every report in a smooth region is used at every state. Hence the receiver's posterior after a report can be computed from the density of anchor realizations that make that report optimal. The sender's first-order condition identifies this anchor realization state by state, and Bayes' rule then determines an exact equation for the receiver's action.

\begin{theorem}\label{thm:continuous-main}
Suppose Assumptions~\ref{ass:receiver-main}-\ref{ass:invert-main} hold, $M=\R$, and $B(\theta)=\R$ for every $\theta$. Let $(R,a)$ be a regular equilibrium, and let $I\subseteq\R$ be an open smooth region of $a$. For each $(r,\theta)\in I\times\Theta$, there is a unique anchor $b^\ast(r,\theta)$ satisfying $D_1(r,b^\ast(r,\theta))=U^S_1(a(r),\theta)a'(r)$.
Moreover, $b^\ast(\cdot,\theta)\in C^2(I)$ and $\partial_r b^\ast(r,\theta)>0$. At every report $r\in I$ with positive report density, the posterior density is
\begin{equation}\label{eq:posterior-main}
\pi(\theta\mid r)=
\frac{
f(\theta)q(b^\ast(r,\theta)\mid\theta)\partial_r b^\ast(r,\theta)
}{
\int_\Theta f(\vartheta)q(b^\ast(r,\vartheta)\mid\vartheta)
\partial_r b^\ast(r,\vartheta)d\vartheta
}.
\end{equation}
The receiver's action therefore satisfies
\begin{equation}\label{eq:continuous-main}
\int_\Theta U^R_1(a(r),\theta)
f(\theta)q(b^\ast(r,\theta)\mid\theta)
\partial_r b^\ast(r,\theta)d\theta=0.
\end{equation}
\end{theorem}

The theorem turns equilibrium updating into a change-of-variables formula. For a given report and state, $b^\ast(r,\theta)$ is the anchor realization for which report $r$ satisfies the sender's first-order condition. The term $\partial_r b^\ast(r,\theta)$ is the Jacobian that converts the density of anchors into the density of reports. Thus the posterior weight placed on state $\theta$ is proportional to the prior density of the state, the likelihood of the anchor realization that rationalizes the report, and the local rate at which anchor types are mapped into nearby reports.
Equation~\eqref{eq:continuous-main} is the receiver's first-order condition after substituting this posterior density. In the next section, additive anchors and translation-invariant reporting costs turn the inverse-anchor map into an explicit function of the receiver's action and its derivative. This converts \eqref{eq:continuous-main} into a differential equation for the equilibrium action rule.

\section{Linear Anchors}\label{sec:packageB}

\subsection{A linear class}

Assume now that the anchor is additive, $b=b_0(\theta)+\sigma x$, and the reporting cost is translation invariant, $D(r,b)=c\phi(r-b)$. Under this structure, the exact Bayes equation from the previous section becomes a differential equation for the receiver's action rule. The parameter \(c\) measures the strength of the reporting cost, \(\sigma\) measures cross-sectional dispersion in private anchors, and \(b_0(\theta)\) is the systematic component of the sender's private language.

\begin{assumption}\label{ass:PB-main}
In addition to Assumptions~\ref{ass:receiver-main}-\ref{ass:invert-main}, the anchor satisfies $b=b_0(\theta)+\sigma x$, where $\sigma>0$, $b_0\in C^3(\Theta)$ is strictly increasing with $b_0'(\theta)\ge \underline b_0'>0$, and $x$ has density $h\in C^4(\R)$ that is strictly positive, strictly log concave, and has bounded and integrable derivatives through order four. The reporting cost is $D(r,b)=c\phi(r-b)$, where $c>0$, $\phi\in C^5(\R)$, $\phi'$ is onto with $\phi(0)=\phi'(0)=0$, and there are constants $0<\underline\kappa\le\overline\kappa<\infty$ and $M_\phi<\infty$ such that $\underline\kappa\le\phi''(u)\le\overline\kappa$ and $|\phi^{(j)}(u)|\le M_\phi$ for $j=3,4,5$.
\end{assumption}

The conditional density of the anchor is
\[
q_\sigma(b\mid\theta)=\sigma^{-1}h\!\left(\frac{b-b_0(\theta)}{\sigma}\right).
\]
The noise term \(\sigma x\) captures differences in methodology, judgment, data vintage, institutional position, or any other private variation that shifts which reports are locally cheap. A nonconstant function \(b_0(\theta)\) can allow the private anchor to be systematically related to the payoff-relevant state without requiring the anchor to equal the state.

Let $\psi\equiv (\phi')^{-1}$. Given a smooth, strictly increasing action rule $a$, the sender's first-order condition at an interior report is $c\phi'(r-b)=U^S_1(a(r),\theta)a'(r)$.
Hence the anchor realization that rationalizes report \(r\) at state \(\theta\) is
\[
b_a(r,\theta)
=
r-\psi\!\left(\frac{U^S_1(a(r),\theta)a'(r)}{c}\right).
\]
The conditional report density induced by state \(\theta\) is obtained by evaluating the anchor density at this rationalizing anchor and multiplying by the Jacobian, $q_\sigma(b_a(r,\theta)\mid\theta)\,\partial_r b_a(r,\theta)$. Multiplying by the prior $f(\theta)$ defines the unnormalized posterior weight
\begin{equation}\label{eq:weight-main}
w_a^\sigma(r,\theta)
=
f(\theta)q_\sigma(b_a(r,\theta)\mid\theta)\partial_r b_a(r,\theta).
\end{equation}

\begin{theorem}\label{thm:ode-main}
Under Assumption~\ref{ass:PB-main}, let \(a\) be the \(C^2\), strictly increasing action rule on an open interval \(I\), and suppose the induced inverse-anchor map satisfies \(\partial_r b_a(r,\theta)>0\) on \(I\times\Theta\). Then the Bayes equation (\ref{eq:continuous-main}) is equivalent on \(I\) to
\begin{equation}\label{eq:ode-main2}
\mathcal A(r,a(r),a'(r))a''(r)=\mathcal B(r,a(r),a'(r)),
\end{equation}
where the rationalizing anchor and the auxiliary score are evaluated with the dummy slope $p$ in place of $a'(r)$,
\[
b_a(r,\theta;p)=r-\psi\!\left(\tfrac{U^S_1(a,\theta)\,p}{c}\right),
\qquad
s(a,\theta;p)=\frac{U^S_1(a,\theta)\,p}{c},
\]
and
\begin{align*}
\mathcal A(r,a,p)
&=
\int_\Theta
U^R_1(a,\theta)f(\theta)q_\sigma(b_a(r,\theta;p)\mid\theta)
\frac{\psi'(s)}{c}U^S_1(a,\theta)\,d\theta,\\[4pt]
\mathcal B(r,a,p)
&=
\int_\Theta
U^R_1(a,\theta)f(\theta)q_\sigma(b_a(r,\theta;p)\mid\theta)
\left[
1-\frac{\psi'(s)}{c}U^S_{11}(a,\theta)\,p^{\,2}
\right]d\theta .
\end{align*}
\end{theorem}

Curvature of the action rule enters the Bayes equation because it controls the rate at which nearby anchors are mapped into nearby reports and so shifts the posterior density induced by a report. The coefficient on curvature is $\mathcal A$, and the residual effect of moving the report while holding curvature fixed is $\mathcal B$. When $\mathcal A\neq0$, receiver optimality pins down the curvature required for Bayes' rule to hold, and the equation can be written as\footnote{On a compact message space, an endpoint pooling replaces the interior density in \eqref{eq:weight-main} with tail probabilities; the corresponding boundary equations are recorded in the supplementary appendix.}
\[
a''(r)=
\frac{\mathcal B(r,a(r),a'(r))}
{\mathcal A(r,a(r),a'(r))}.
\]

For the comparative-static and limit results that follow, I restrict attention
to the smooth global branch of the characterization. A regular equilibrium
\((R_c,a_c)\) is called a smooth regular equilibrium if \(R(\theta,b)\)  is the unique solution to the sender first-order condition and
the local
characterization in Theorem~\ref{thm:ode-main} is global, that is,
\(M=I=\mathbb R\). Equivalently, \(a_c\in C^3(\mathbb R)\), \(a_c'(r)>0\) for all
\(r\), the inverse-anchor map is well defined on \(\mathbb R\times\Theta\), and
\(a_c\) satisfies the ODE on all of \(\mathbb R\).

\subsection{Low cost limit}\label{sec:cost-floor}

The baseline model imposes more structure than standard cheap-talk because it assumes both an informative anchor and a disutility from deviating from that anchor. The next two subsections show that this structure has bite even near pure cheap talk. Surprisingly, equilibrium informativeness can remain bounded away from babbling when reporting costs vanish, and anchored reporting can discipline communication even when the anchor itself is uninformative noise. Small anchored perturbations of a model without anchors therefore generate discontinuous changes in equilibrium communication.

When the reporting cost is exactly zero, the private anchor no longer disciplines reports. One might expect regular equilibria to converge to babbling as the cost becomes small. The opposite occurs. Even a small positive cost ranks deviations from the private anchor, and that ranking keeps reports informative as the cost level vanishes.
Let  $W^R_{\mathrm{bab}}$ ($W^R_c$ ) be the receiver payoff from ignoring the report (in a regular equilibrium) given cost $c$.

\begin{proposition}\label{prop:no-babbling-smallcost}
Suppose that Assumptions~\ref{ass:receiver-main}-\ref{ass:invert-main} hold, $b=b_0(\theta)+x$, where $b_0$ is continuous and strictly increasing and $x$ is independent of $\theta$ with nondegenerate continuous distribution, and $U^S_{12}(a,\theta)\ge0$ on $A\times\Theta$. There is $\varepsilon>0$ such that every smooth regular equilibrium with cost $cD$ satisfies $W^R_c\ge W^R_{\mathrm{bab}}+\varepsilon$.
\end{proposition}

A regular equilibrium therefore cannot converge to babbling. At low costs, the sender can inflate messages, but inflation is still ranked by the distance from the private anchor. The ordered private anchors generate an ordered cross-section of locally cheap reports, and weak increasing differences preserve enough of that ordering to leave the receiver with a uniformly nontrivial coarsening of the state.

How much communication can be sustained when the cost is small? The answer depends on two forces. The systematic part of the anchor must move with the state, and the sender's marginal value of raising the receiver's action must also move with the state. If the second force is weak, low costs make reports track the noisy private anchor rather than the objective state. The remainder of this subsection imposes the following sharper payoff structure.

\begin{assumption}\label{ass:smallcost-main}
In addition to Assumption~\ref{ass:PB-main}, $M=\R$, $U^S_1(a^R(\theta),\theta)>0$, $U^S_{11}(a,\theta)\le-\eta_S<0$, and $U^S_{12}(a,\theta)\ge\eta_{S\theta}>0$ on the relevant action-state domain.
\end{assumption}

Together, the conditions of Assumption~\ref{ass:smallcost-main} say that the sender wants to push the receiver upward at the full-information action, that the marginal value of further persuasion falls as the action rises, and that higher states have a stronger incentive to raise the action. They dampen the idiosyncratic part of the anchor and make the likelihood of a fixed report sharply state dependent.

The small-cost analysis compares the equilibrium action rule with the deterministic least-cost separating language. Let $\rho_c$ denote the strictly increasing solution of $c\phi'(\rho_c(\theta)-b_0(\theta))\rho_c'(\theta)=U^S_1(a^R(\theta),\theta)(a^R)'(\theta)$ with $\rho_c(\underline\theta)=b_0(\underline\theta)$, the separating language that would arise without anchor noise.\footnote{Existence and uniqueness of $\rho_c$ are established in Corollary~\ref{cor:det-least} of the supplementary appendix.} The next assumption is a stability restriction: the receiver's response to the deterministic-language report tracks the full-information action as the cost vanishes. It strengthens the maintained existence convention, which by itself does not pin down the limiting behavior of the selected equilibria.

\begin{assumption}\label{ass:branch-tracking}
As $c\downarrow0$, the regular equilibria $(R_c,a_c)$ satisfy $a_c(\rho_c(t))=a^R(t)+o(1)$ and $\partial_t\,a_c(\rho_c(t))=(a^R)'(t)+o(1)$, uniformly on compact subsets of the interior of $\Theta$.
\end{assumption}

\begin{proposition}\label{prop:small-cost-revelation}
Under Assumptions~\ref{ass:smallcost-main} and~\ref{ass:branch-tracking}, for all sufficiently small $c$, there is a regular equilibrium $(R_c,a_c)$ such that $a_c(R_c(\theta,b))\to a^R(\theta)$ in probability as $c\downarrow 0$.
\end{proposition}

One might expect lower anchor costs to make reports less informative. After
all, as \(c\) moves toward zero, the limiting ODE at \(c=0\) corresponds to
a completely uninformative report. But lowering \(c\) also reduces the extent
to which the sender passes anchor noise into the message. The result shows
that this second force dominates at low cost: as \(c\downarrow0\), the
receiver's action converges to the full-information action. The reason is
that a small positive anchor cost does not eliminate the anchor's role.
Instead, it induces the sender to expand the report scale without bound. On
that expanded scale, the anchor noise becomes small relative to the movement
in reports generated by changes in the state. Reports therefore become
asymptotically fully revealing, even though the zero-cost limit itself is
uninformative.

To examine the cost borne by the sender, define the marginal incentive to separate state $\theta$ at the full-information action, $T(\theta)=U^S_1(a^R(\theta),\theta)(a^R)'(\theta)$, and its cumulative integral $\Gamma(\theta)=\int_{\underline\theta}^\theta T(u)\,du$.

\begin{theorem}\label{thm:cost-floor}
Along the equilibrium of Proposition~\ref{prop:small-cost-revelation},
\[
\E[c\phi(R_c(\theta,b)-b)]\to\E[\Gamma(\theta)]=\int_\Theta(1-F(u))T(u)\,du
\]
as $c\downarrow 0$. On a compact state space the limit is finite. If the model is extended to $\underline\theta=-\infty$ with $T$ bounded away from zero in the lower tail, the limit is infinite.
\end{theorem}

 On compact state spaces, the cost to the sender is finite under bounded marginal conflict. In lower-unbounded environments with non-vanishing conflict of interest, the cumulative incentive $\E[\Gamma(\theta)]$ is infinite, so near-full revelation requires unbounded expected reporting costs. Full revelation with small anchor costs can therefore be valuable to the receiver while imposing large expected reporting costs on the sender.

\subsection{Uninformative anchors}\label{sec:uninformative}

Consider the extreme case in which the anchor is independent of the payoff-relevant state. Formally, $b=\mu+\sigma x$, where $x$ is independent of $\theta$ and $\sigma>0$, so that a report of $b$ carries no information. Even here, the anchor still pins down the sender's cost-minimizing report: if different states have different gains from moving the receiver's action, they distort by different amounts from the same uninformative anchor, and the receiver can learn from the distortion rather than from the anchor itself. Let $a^0=\arg\max_a\E[U^R(a,\theta)]$ denote the receiver's prior-optimal action.

\begin{assumption}\label{ass:uninf-align}
The function $\theta\mapsto U^R_1(a^0,\theta)$ is strictly increasing, $\theta\mapsto U^S_1(a^0,\theta)$ is weakly increasing and not almost surely constant, $\E[U^R_{11}(a^0,\theta)]<0$, and $\phi''(0)>0$.
\end{assumption}

The monotonicity requirements are the standard single-crossing conditions at the babbling action. Since $a^0$ is prior-optimal, Assumption~\ref{ass:uninf-align} implies that the receiver's marginal value of a higher action and the sender's marginal incentive to induce one have positive covariance across states. If the sender's marginal incentive were independent of the state, every state would distort from the same anchor in the same way and an uninformative anchor could not create information.

The next assumption restricts only the noise distribution. The compact-message case is stated for completeness; the existence result below is established in the real-line case, where the linearized operator has a canonical self-adjoint realization.

\begin{assumption}\label{ass:uninf-noise}
Either $M$ is a compact interval and the density of $b$ is $C^2$ and bounded above and away from zero on $M$, or $M=\R$ and $x$ has density proportional to $\exp\{-V(x)\}$, where $V\in C^3$, $V''$ is bounded below by a positive constant, and $V(x)\to\infty$ as $|x|\to\infty$.
\end{assumption}

Substituting the uninformative anchor into the regular-equilibrium ODE of Theorem~\ref{thm:ode-main} reduces equilibrium to a calibration exercise. Babbling $a\equiv a^0$ is always a solution; the remaining question is whether a nonconstant solution can sustain more information than babbling. Linearizing the Bayes equation at $a^0$ produces a Sturm-Liouville problem whose first nontrivial eigenvalue pins down a positive cost scale.

\begin{theorem}\label{thm:uninf-existence}
Under Assumptions~\ref{ass:uninf-align} and~\ref{ass:uninf-noise} in the real-line case, there is a positive cost level $c^\ast$ at which the uninformative-anchor model admits regular equilibria with nonconstant action rules. 
\end{theorem}

Shrinking the noise along the separating solution drives the receiver's posterior toward full information.
\begin{assumption}\label{ass:uninf-design}
The state space is a compact interval, $M=\R$, and $a^R$ is $C^1$ and strictly increasing. For every $\theta$, $U^S_1(a^R(\theta),\theta)>0$, and for all $t,\theta$, $U^S_{12}(a^R(t),\theta)>0$. The cost function $\phi$ is strictly convex, strictly increasing on $[0,\infty)$, and unbounded.
\end{assumption}

Under Assumption \ref{ass:uninf-design}, every type wants to raise the receiver's action at the full-information action, and higher states have a stronger marginal gain from doing so. 
\begin{theorem}\label{thm:uninf-fullinfo}
Under Assumptions~\ref{ass:branch-tracking}, ~\ref{ass:uninf-noise}, and~\ref{ass:uninf-design} in the real-line case, there exist regular equilibria with anchors $b=\mu+\sigma x$ such that, as $\sigma\downarrow0$, the receiver action satisfies $a_\sigma(R_\sigma(\theta,b))\to a^R(\theta)$ in probability and the expected reporting cost converges to $\E[\Gamma(\theta)]$.
\end{theorem}

Theorem~\ref{thm:uninf-fullinfo} strengthens the earlier result by showing that an uninformative anchor can not only
be used to support informative communication in a regular equilibrium but, also, in the limit such that the anchor becomes deterministic, communication becomes fully informative. However, as in the case of full-communication with a vanishing cost, the sender incurs a possibly large bias cost.

?

\subsection{Gaussian-Quadratic model}\label{sec:benchmark}

The Gaussian-Quadratic special case of additive anchors admits closed-form expressions for the main forces in the model, with additional explicit welfare implications. Throughout this section, I assume that $\theta\sim N(0,\sigma_\theta^2)$ and $b=\beta\theta+\sigma x$, where \(x\sim N(0,1)\) is independent of \(\theta\), \(\sigma>0\), and \(\beta\ge0\). Payoffs are
\[
U^R(a,\theta)=-\frac12(a-\theta)^2,
\qquad
U^S(a,\theta)=-\frac12(a-\theta-d)^2,
\qquad
D(r,b)=\frac c2(r-b)^2,
\]
with \(c>0\) and \(d\ge0\). The sender's ideal action is therefore \(\theta+d\), the constant \(d\) is the bias, and the parameter \(\beta\) measures the informativeness of the anchor.

For an informative anchor (\(\beta>0\)), the first-order condition and Bayes' rule together force the regular equilibrium to be linear.
\begin{theorem}\label{thm:rello-main}
Suppose \(\beta>0\) and \(M=\R\). In a regular equilibrium, the receiver action is
\begin{equation}\label{eq:PA-action-main}
a(r)=\alpha r-\frac{\alpha^2 d}{c},
\end{equation}
where \(\alpha\) is the unique positive root of
\begin{equation}\label{eq:alpha-beta}
\beta\sigma_\theta^2\alpha^2+
\bigl(c\sigma^2+c\beta^2\sigma_\theta^2-\sigma_\theta^2\bigr)\alpha
-c\beta\sigma_\theta^2=0.
\end{equation}
The sender report rule is
\begin{equation}
R(\theta,b)=
\frac{c\beta+\alpha}{c+\alpha^2}\theta
+
\frac{c}{c+\alpha^2}(b-\beta\theta)
+
\frac{\alpha d}{c}.
\end{equation}
\end{theorem}

Because both the sender's payoff and the reporting cost are quadratic, the sender's first-order condition is linear in the state; for each report $r$, it pins down the break-even anchor $b_r(\theta)$ that makes type $\theta$ just willing to send $r$, and this break-even anchor is a linear function of $\theta$ with slope $-a'(r)/c$. If two reports had different sensitivities $a'(r_1)\neq a'(r_2)$, their break-even-anchor lines would have different slopes and cross at some state, contradicting the unique best response in  a regular equilibrium. The slope $a'(r)$ must therefore be the same for every report, and the regular equilibrium is affine. A direct calculation then produces the posterior variance
\begin{equation}\label{eq:posterior-var-beta}
\Var(\theta\mid r)=\frac{c(1-\beta\alpha)\sigma_\theta^2}{c+\alpha^2}.
\end{equation}

\begin{corollary}\label{cor1r}
The expected report-anchor distortion
\begin{equation}\label{eq:gauss-cost}
\E[D(r,b)]=\frac{\alpha^2}{2c}\bigl(\Var(\theta\mid r)+d^2\bigr),
\end{equation}
and the receiver and sender welfare are
\begin{equation}\label{eq:gauss-WS}
W^R=-\frac{1}{2}\Var(\theta\mid r),\qquad
W^S=-\frac{\Var(\theta\mid r)+d^2}{2}\Bigl(1+\frac{\alpha^2}{c}\Bigr).
\end{equation}
\end{corollary}
The receiver's payoff loss is half of the residual posterior variance. The sender's payoff has two pieces: an action-mismatch loss $\Var(\theta\mid r)(1+\alpha^2/c)/2$ that combines the residual posterior variance with the inflation needed to convey information, and a bias-induced loss $d^2(1+\alpha^2/c)/2$ paid whenever the sender's bliss point differs from the state. The equilibrium sensitivity $\alpha$, the posterior variance, and the receiver welfare $W^R$ depend on $(c,\sigma,\beta,\sigma_\theta)$ but not on the bias $d$ since it is undone in equilibrium.

\begin{corollary}\label{cor:PA-comp}
For $\beta>0$:
\begin{enumerate}
\item[(i)] The sensitivity $\alpha$ lies in $(0,1/\beta)$ and is strictly decreasing in the reporting cost $c$ and the anchor noise $\sigma$, and independent of the bias $d$. The product $\beta\alpha\in(0,1)$ is strictly increasing in $\beta$. The share of prior variance transmitted by the report is $1-\Var(\theta\mid r)/\sigma_\theta^2=(\alpha^2+c\beta\alpha)/(c+\alpha^2)$; this share is strictly decreasing in $c$ and $\sigma$, but it need not be monotone in $\beta$, because $\beta$ enters both through $\beta\alpha$ (rising in $\beta$) and through $\alpha$ itself (falling in $\beta$).
\item[(ii)] The residual posterior variance $\Var(\theta\mid r)$ is strictly increasing in $c$ and $\sigma$ and independent of $d$. \item[(iii)] The average message inflation $\E[r-b]=\alpha d/c$ is strictly decreasing in $c$ and  $\sigma$ and increasing in $d$.
 \end{enumerate}
\end{corollary}

The reporting cost and the anchor noise act through the same channel for the receiver: stronger compliance ($c$ large) and noisier anchors ($\sigma$ large) both shrink the equilibrium $\alpha$, which lowers $\beta\alpha$ and the transmitted share $(\alpha^2+c\beta\alpha)/(c+\alpha^2)$ and leaves the receiver with more residual posterior variance and lower welfare. The effect of the anchor informativeness $\beta$ on the transmitted share is non-monotone in general: a higher $\beta$ raises the composite $\beta\alpha$ but also lowers $\alpha$ itself, and in low-cost or low-$\sigma$ regions the second effect can dominate. The bias $d$ does not enter the informativeness channel: it shifts the constant $\alpha d/c$ in the report rule, the matching intercept $-\alpha^2 d/c$ in the receiver's action rule, and the sender's equilibrium welfare loss.

\begin{corollary}\label{cor:PA-comp_s}
Sender welfare can be written as
\begin{equation}\label{eq:WS-decomposition}
W^S = -\tfrac12(1-\beta\alpha)\sigma_\theta^2 - \tfrac12 d^2 - \tfrac{d^2\alpha^2}{2c},
\end{equation}
There exist two misalignment thresholds $|\underline{d}|<|\overline{d}|$ such that:
\begin{itemize}
\item[(i)] if $|d|<|\underline{d}|$, $W^S$ is strictly decreasing in $c$ and in $\sigma$. The sender prefers a low reporting cost and a low anchor noise, i.e., the comparative statics are aligned with the receiver. 
\item[(ii)] if $|\underline{d}|<|d|<|\overline{d}|$, $W^S$ is strictly increasing in $c$ and strictly decreasing in $\sigma$, i.e., receiver and sender benefit from low anchor noise but their preference are different with respect to the cost $c$;
\item[(iii)] if $|d|>|\overline{d}|$, $W^S$ is strictly increasing in $c$ and in $\sigma$, implying that sender and receiver have opposite preferences in these parameters.
\end{itemize}
\end{corollary}

Two effects of stronger compliance or noisier anchors pull in opposite directions through $\alpha$. A lower $\alpha$ raises the term $(1-\beta\alpha)\sigma_\theta^2/2$ (the residual variance grows) but reduces the bias-induced inflation cost $d^2\alpha^2/(2c)$. The first dominates at low bias and the second at high bias. The thresholds in $c$ and $\sigma$ are not the same because $c$ also enters the inflation cost directly through the $1/c$ factor; the bias level needed to flip the sign of $\partial W^S/\partial c$ is therefore strictly lower than the level needed to flip the sign of $\partial W^S/\partial \sigma$. At intermediate misalignment the sender prefers stricter reporting but less volatile anchors, so her interest in $c$ is opposite to the receiver's while her interest in $\sigma$ remains aligned with it.

Finally, consider the case of uninformative anchors, i.e., \(\beta=0\), so that \(b=\sigma x\) is independent of the state. In the affine regular class, the equilibrium condition becomes
\begin{equation}\label{eq:uninf-degen}
(c\sigma^2-\sigma_\theta^2)\alpha=0.
\end{equation}
Away from the condition \(c\sigma^2=\sigma_\theta^2\), the only affine regular equilibrium has \(\alpha=0\) and is uninformative. When this condition is met, however, every \(\alpha>0\) is an affine regular equilibrium, so that there is a family of equilibria indexed by $\alpha$:\footnote{The calibration $c\sigma^2=\sigma_\theta^2$ is non-generic in $(c,\sigma)$ but reachable by design: $\sigma=\sigma_\theta/\sqrt c$ places the model on the knife edge for any $c>0$, and symmetrically for $c$ given $\sigma$.}
\[
a(r)=\alpha r-\frac{\alpha^2 d}{c},
\qquad
R(\theta,b)=
\frac{\alpha}{c+\alpha^2}\theta
+
\frac{c}{c+\alpha^2}b
+
\frac{\alpha d}{c}.
\]
The posterior variance and welfare are
\begin{align}
\Var(\theta\mid r)
=
\frac{c\sigma_\theta^2}{c+\alpha^2},\qquad
W^R
=
-\frac{c\sigma_\theta^2}{2(c+\alpha^2)}
,\qquad
W^S
=
-\frac{\sigma_\theta^2}{2}
-\frac{d^2}{2}
-\frac{d^2\alpha^2}{2c}.
\end{align}

 The sender welfare is bounded above by $-\sigma_\theta^2/2-d^2/2$, the welfare under babbling, so for any positive bias the sender always prefers babbling (and thus also prefers any informative cheap-talk equilibrium). The receiver prefers higher $\alpha^2/c$, which drives the residual posterior variance toward zero, since $\Var(\theta\mid r)=\sigma_\theta^2/(1+\alpha^2/c)$. The sender's ex-ante payoff is strictly decreasing in $\alpha^2/c$. When the calibration $c\sigma^2=\sigma_\theta^2$ is maintained, the equilibrium $\alpha$ is a free selection along the family; sending $\alpha\uparrow\infty$ (for any fixed $c$, or with $c\downarrow 0$ along the calibration $\sigma=\sigma_\theta/\sqrt c$) drives $\alpha^2/c\to\infty$ and the receiver toward the full-information welfare $W^R=0$. Hence a designer who can select within the family of informative regular equilibria can approximate full revelation, even though the anchor itself carries no information about the state. The sender's welfare diverges along this path.

Figure~\ref{fig:pareto} illustrates the frontier at $d=0.25$, $\sigma_\theta=1$, $c=1$ (so the condition for informative message fixes $\sigma=1$). Babbling corresponds to $\alpha=0$, the level that maximizes the joint surplus is $\alpha^\star=\sqrt{3}\approx1.73$, and the receiver welfare improves along the frontier as $\alpha$ grows. The most informative Crawford-Sobel partition with bias $d$  lies strictly above the frontier in sender welfare but is not always preferred by the receiver.

\IfFileExists{figure_pareto_frontier.pdf}{%
\begin{figure}[t]
\centering
\includegraphics[width=0.7\textwidth]{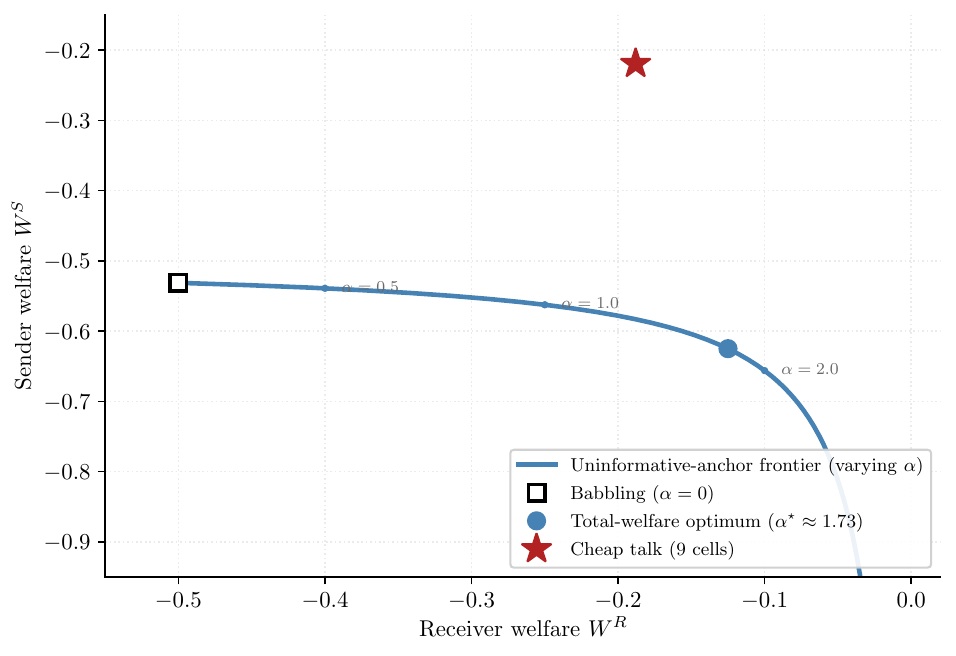}
\caption{Pareto frontier of uninformative anchors ($\beta=0$, $c\sigma^2=\sigma_\theta^2$) at $d=0.25$, $\sigma_\theta=1$, $c=1$. Babbling is at $\alpha=0$; the total-welfare-maximizing equilibrium is at $\alpha^\star=\sqrt{3}\approx1.73$. The most informative cheap-talk partition with bias $d$ is shown for comparison.}\label{fig:pareto}
\end{figure}
}{}

\section{Hard and Soft Communication}\label{sec:hybrid}

\subsection{Hybrid equilibrium}

Many reporting systems combine a continuous, anchor-disciplined report with a coarse qualitative assessment. Analysts pair price targets with buy/hold/sell recommendations, auditors pair numerical risk scores with verbal evaluations, and review platforms display star ratings alongside numerical scores. The qualitative message is cheap talk; the number is costly to move away from the sender's private anchor. The two channels are complementary: the label sorts states into coarse categories, and the anchored number refines information within the selected category.

The hybrid class studied here uses one-dimensional labels monotonic in the state, with the report inside label $k$ interpreted under the regular anchored continuation of the prior truncated to $[\theta_{k-1},\theta_k]$. The qualitative channel is \textit{pre-anchor}: the sender chooses the label after observing $\theta$ but before observing $b$, or, equivalently, the label is institutionally constrained not to condition on the private anchor.

Pre-anchor timing has two analytical consequences. A label chosen after observing $(\theta,b)$ would have to satisfy label incentive constraints pointwise in $b$, ruling out one-dimensional state-cell or anchor-cell partitions and forcing non-monotonic or two-dimensional labels.\footnote{Proofs of these impossibility statements are given in the supplementary appendix.} Pre-anchor timing is also what allows regularity within each label without forcing it across labels, since continuous within-label action rules can differ across labels and that difference is what makes the qualitative channel informative.\footnote{With message space $\{1,\ldots,N\}\times\R$, informative labels generally imply $a_k(r)\ne a_\ell(r)$ for $k\ne\ell$, so the global action rule is discontinuous across labels at any fixed numerical report. Forcing $a_k(r)=a_\ell(r)$ would collapse the qualitative channel.} The timing also matches the institutional ordering in many applications: an analyst issues a recommendation before the justifying details of the underlying valuation are pinned down, and a manager commits to a categorical assessment before observing the precise estimate that supports the public report.

\begin{definition}\label{def:cell-hybrid-main}
A \emph{regular full-label state-cell hybrid equilibrium with $N$ labels} consists of ordered cutoffs $\underline\theta=\theta_0<\theta_1<\cdots<\theta_N=\overline\theta$, cells $I_k=[\theta_{k-1},\theta_k)$ for $k<N$ and $I_N=[\theta_{N-1},\theta_N]$,
regular anchored continuations $(R_k,a_k)$ for the priors $f$ truncated to $I_k$, and a state-cell label rule $\chi(\theta)=k$ for $\theta\in I_k$. For a candidate label $k$, define the sender's ex ante label value by
\[
\bar V_k(\theta) = \int_{\R}\sup_{r\in\R}\Big\{U^S(a_k(r),\theta)-D(r,b)\Big\}q(db\mid\theta).
\]
The state-cell label rule is incentive compatible if $
\bar V_k(\theta)\ge \bar V_\ell(\theta)$ for every $\theta\in I_k$ and every $\ell\in\{1,\ldots,N\}$.
After label $k$ and anchor realization $b$, the sender uses the within-cell optimal report $R_k(\theta,b)$. The receiver's action after $(k,r)$ is $a_k(r)$, and receiver beliefs after $(k,r)$ are the Bayes posterior generated by the regular within-cell continuation on $I_k$.
\end{definition}

The next result shows how the qualitative label changes the interpretation of a numerical report. Throughout this subsection, payoffs satisfy Assumptions~\ref{ass:receiver-main}--\ref{ass:cost-main} and the additive-anchor environment of Assumption~\ref{ass:PB-main} is in force; the message space and the per-state anchor support are taken to be $\R$, as in Section~\ref{sec:packageB}.

\begin{proposition}\label{prop:hybrid-local-language}
In any regular full-label state-cell hybrid equilibrium, for every report $r$ and labels $k<\ell$, $a_k(r)<a_\ell(r)$. Moreover, for every label $k$, every state $\theta\in I_k$, and every numerical report $r\in\R$, there exists an anchor realization $b$ such that $R_k(\theta,b)=r$.\footnote{Under the pre-anchor label timing, the anchor cannot induce a type to switch labels; the strict ranking follows from strict first-order stochastic dominance of the posterior conditional on the higher label, applied together with the full-use property within each cell. When the message space is a compact interval or the anchor support is bounded, the strict inequality holds on the interior of the report set, with equality possible at endpoints where full report support fails.}
\end{proposition}

The same number carries different meanings under different labels. A price target attached to ``buy'' is not interpreted like the same target attached to ``hold'' because the label shifts the support of the receiver's posterior and the anchored number refines beliefs within that support.

\subsection{Hybrid versus anchor-only}

A pre-anchor partitional message is uninformative if the misalignment is high,
since only babbling can be sustained; then, a hybrid message can add no
information. The comparison below therefore considers the interesting case of small misalignment using a perturbation of the sender's action payoff. For a given value of the
perturbation parameter, let 
$a^S(\theta;d)=\arg\max_{a\in A} U^S(a,\theta;d),$ satisfy the following.

\begin{assumption}\label{ass:primitive-smallbias}
The sender-payoff restrictions in
Assumption~\ref{ass:receiver-main} hold for $U^S(\cdot,\cdot;d)$ uniformly in
$d$. Further: (i) $(a,\theta,d)\mapsto U^S(a,\theta;d)$ is $C^4$ on
$A\times\Theta\times[0,\bar d_0]$; (ii) At $d=0$, sender and receiver ideal actions coincide: $U^S_1(a^R(\theta),\theta;0)=0$
for all $\theta\in\Theta$; (iii) There exists $\hat\delta\in C^2(\Theta)$, with
$\hat\delta(\theta)>0$ for all $\theta$, such that
$\partial_d U^S_1(a^R(\theta),\theta;d)|_{d=0}
=
-U^S_{11}(a^R(\theta),\theta;0)\hat\delta(\theta).
$
\end{assumption}

Under Assumption~\ref{ass:primitive-smallbias}, $a^S(\theta;d)=a^R(\theta)+d\hat\delta(\theta)+O(d^2)$ is the first-order displacement of the sender's ideal action from the receiver's ideal action. When the ex ante label value depends on the cutoff vector $\mathbf t=(t_1,\ldots,t_{N-1})$ and on $d$, it is written $\bar V_k(\theta;\mathbf t,d)$. Let $W_H^i(d)$ denote the surplus of the hybrid equilibrium and $W_A^i(d)$ the surplus of the anchor-only equilibrium.

\begin{theorem}\label{thm:hybrid-anchor}
Suppose Assumptions~\ref{ass:primitive-smallbias} and~\ref{ass:lowbias-symmetric} hold, the latter stated in Section~\ref{sec:hybrid-vs-cheaptalk}. There exists $\bar d>0$ such that $W_H^i(d)>W_A^i(d)$ for all $d\in(0,\bar d)$ and $i\in\{R,S\}$.
\end{theorem}

The pre-anchor partitional message breaks the state space into labels of vanishing width and the anchored numerical report then refines information within each label. The resulting hybrid loss vanishes as $d\downarrow0$. In a regular anchor-only equilibrium at exact alignment, by contrast, the receiver interprets a single numerical report using the full prior support. With full-support anchor noise, this leaves a non-degenerate posterior and hence a strictly positive exact-alignment loss, implying that the added pre-anchor label is strictly valuable to both players when misalignment is sufficiently low.

\IfFileExists{figure_welfare_vs_d.pdf}{%

\begin{figure}[t]
\centering
\includegraphics[width=\textwidth]{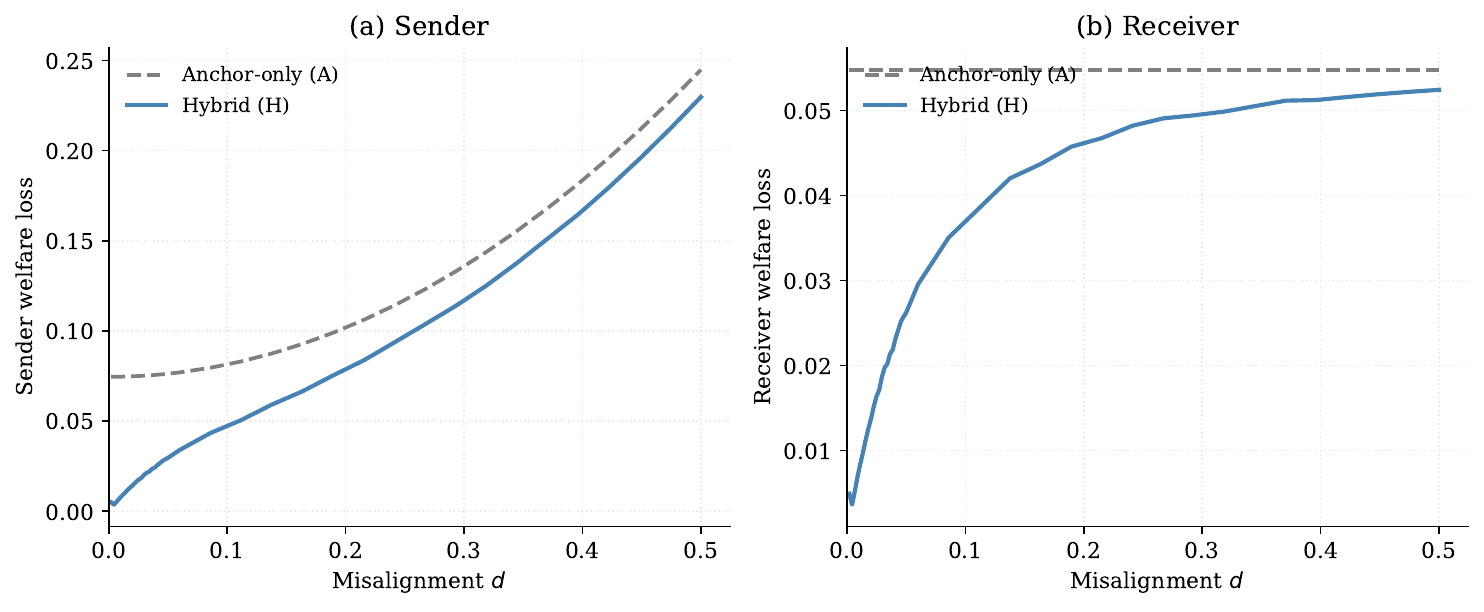}
\caption{Sender (panel~a) and receiver (panel~b) welfare in the Gaussian-quadratic benchmark of Section~\ref{sec:benchmark} for the hybrid (H) and a plotted regular anchor-only selection (A). The hybrid welfare loss vanishes as $d\downarrow 0$, while the plotted anchor-only loss has a strictly positive limit driven by the residual posterior variance generated by anchor noise. Parameters: $\theta\sim N(0,1)$, $b=\theta+\sigma\varepsilon$, $\varepsilon\sim N(0,1)$, $\sigma=0.5$, $c=2$, $\beta=1$.}\label{fig:hybrid-anchor}
\end{figure}
}{}

\subsection{Hybrid versus pure cheap talk}\label{sec:hybrid-vs-cheaptalk}

While the previous analysis demonstrates the value of cheap talk for low misalignment, it does not yet rule out that anchors should be used in complement with cheap talk. However, the asymptotic comparison with pure cheap talk is more subtle. As $d\downarrow 0$, the most informative \citet{CrawfordSobel1982} partition uses an unboundedly fine grid of labels, so cheap-talk welfare also approaches the full-information benchmark for both players. The two modes of communication therefore converge toward first-best, and the comparison hinges on which one approaches faster. Numerical anchors refine actions within each label, but, at the same time, by changing the relative attractiveness of adjacent labels they also force a coarsening of the partition needed to restore label incentives, and the cost of that coarsening can outweigh the within-label refinement gain and create direct welfare losses for the sender. Let $W_C^i(d)$ denote the surplus of the most informative cheap-talk partition, defined from Definition~\ref{def:cell-hybrid-main} after dropping the anchored numerical message.

A tractable environment for the closed-form cheap-talk comparison retains the small-bias perturbation $d\hat\delta(\theta)$ of Assumption~\ref{ass:primitive-smallbias} and specializes payoffs to symmetric error losses, the prior to a log-concave density, and the anchor to a linear function of the state with symmetric noise. The regular full-label state-cell hybrid setup of Definition~\ref{def:cell-hybrid-main} and the additive-anchor reporting-cost environment of Assumption~\ref{ass:PB-main} remain in force.

\begin{assumption}\label{ass:lowbias-symmetric}
Suppose that Assumption~\ref{ass:primitive-smallbias} holds. In addition: 
\begin{enumerate}[label=(\roman*)]
\item Payoffs are symmetric error losses, i.e., for each $i\in\{R,S\}$, $U^R(a,\theta)=v_R(\theta)-\ell_R(a-\theta)$ and $U^S(a,\theta;d)=v_S(\theta)-\ell_S(a-\theta-d\hat\delta(\theta))$, with $\hat\delta\in C^2(\Theta)$, $\hat\delta>0$, $\hat\delta'\ge 0$, $\ell_i\in C^5$, $\ell_i(-e)=\ell_i(e)$, $\ell_i$ minimized at zero, and $\ell_i''(0)>0$.
\item The state density $g$ on the compact interval $\Theta$ is $C^4$, strictly positive, and log-concave.
\item The anchor is linear: $b=\beta_0+\beta\theta+\sigma\varepsilon$ with $\sigma>0$ and $\beta\ne 0$, where the noise $\varepsilon$ has a strictly positive, symmetric, $C^4$ density $\omega$. Letting $s_\omega=\omega'/\omega$ be the score, there is $K<\infty$ such that $|s_\omega(u)|\le K(1+|u|)$ and $|s_\omega'(u)|+|s_\omega''(u)|+|s_\omega'''(u)|\le K$ for every $u\in\R$, and $\E[e^{\eta|\varepsilon|}]<\infty$ for some $\eta>0$.
\end{enumerate}
\end{assumption}

\begin{theorem}\label{thm:hybrid-ranking}
Under Assumption~\ref{ass:lowbias-symmetric}, for each player $i\in\{R,S\}$ there exists $\bar d>0$ such that $W_C^i(d)>W_H^i(d)$ for all $d\in(0,\bar d)$.
\end{theorem}

The intuition for the Theorem is that anchored refinement adds little inside an already narrow label, and what little it adds is more than offset by the widening of the labels needed to preserve incentive compatibility. Each cheap-talk label has width of order $\sqrt d$, so the residual posterior uncertainty inside a label is already small. Anchors reduce it further, but the proportional reduction scales with $h^2=O(d)$ relative to the within-label variance, leaving only a second-order absolute gain of order $h^4=O(d^2)$. At the same time, the within-label anchor changes how attractive each label looks at the boundary type; the upper label becomes more attractive once the anchor is added, so restoring incentive compatibility requires widening the labels. The within-label information gain and the partition-widening loss are both of order $d^2$, but the widening loss has a strictly larger coefficient, so the net effect on welfare is a second-order loss rather than a second-order gain. Reporting costs themselves do not enter this comparison: the sender's optimal report deviates from the anchor by $O(d^{3/2})$ and is paid quadratically, so the reporting cost is $O(d^3)$, an order of magnitude below the $d^2$ ranking. Figure~\ref{fig:lowd} illustrates both points in the canonical uniform-quadratic setup.

\begin{figure}[t]
\centering
\includegraphics[width=\textwidth]{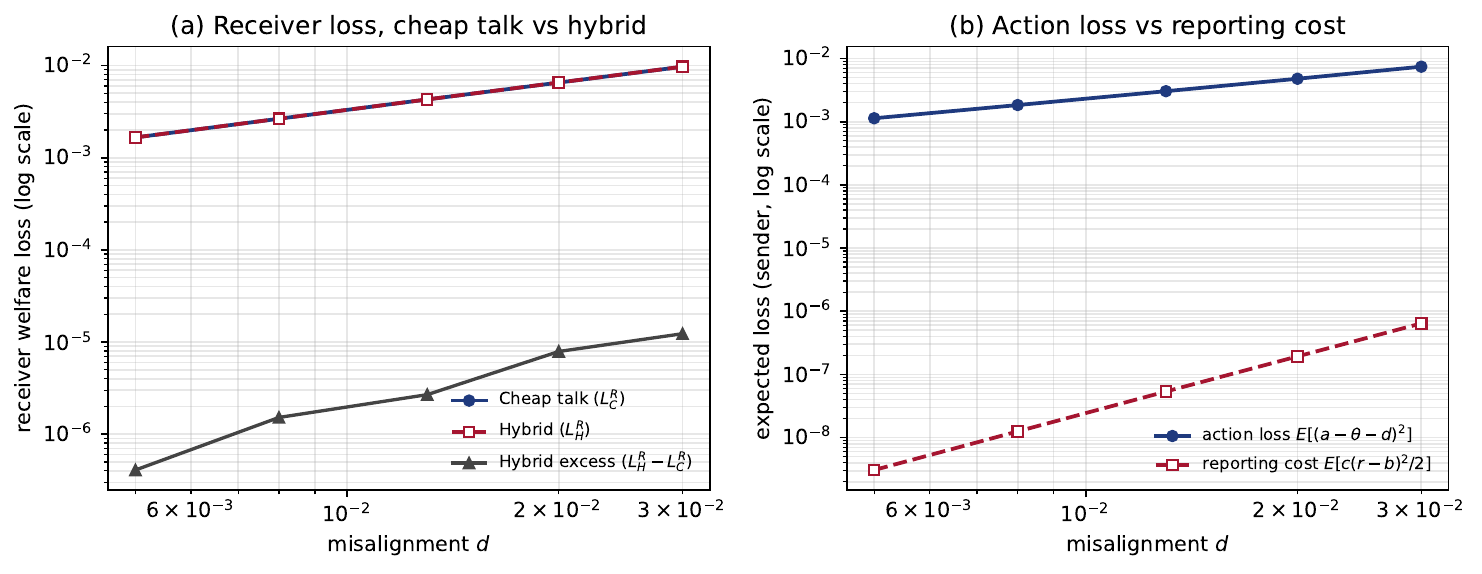}
\caption{Numerical illustration of Theorem~\ref{thm:hybrid-ranking} in the canonical setup with $\theta$ uniform on $[0,1]$, constant misalignment $\hat\delta\equiv1$, quadratic losses with $\kappa_R=\kappa_S=1$, Gaussian anchor $b=\theta+\sigma\varepsilon$ with $\sigma=1$, and quadratic reporting cost with $c=1$. Panel~(a), log-log: the absolute receiver losses $L_C^R$ and $L_H^R$ coincide visually because they share the same leading term, and the hybrid excess $L_H^R-L_C^R$, shown as a third curve two orders of magnitude below, is positive at every $d$. Panel~(b), log-log: the two components of the sender's expected loss at the hybrid equilibrium; the expected reporting cost falls several orders of magnitude faster than the loss from a wrong receiver action and is irrelevant to the ranking between cheap talk and the hybrid.}\label{fig:lowd}
\end{figure}

Combining the two theorems, $W_C^i(d)>W_H^i(d)>W_A^i(d)$ for both players at sufficiently small $d$: anchor-only communication carries a fixed information loss while cheap talk and the hybrid both approach the full-information benchmark with the same leading $O(d)$ welfare loss, and the cheap-talk partition has a strictly smaller second-order term. At high misalignment $d$, when labels are wide enough that anchored refinement carries substantial within-label information, the hybrid can in turn dominate cheap talk. \IfFileExists{figure_hybrid_illustration.pdf}{%
Figure~\ref{fig:hybrid} illustrates a moderate-bias regime.

\begin{figure}[t]
\centering
\includegraphics[width=\textwidth]{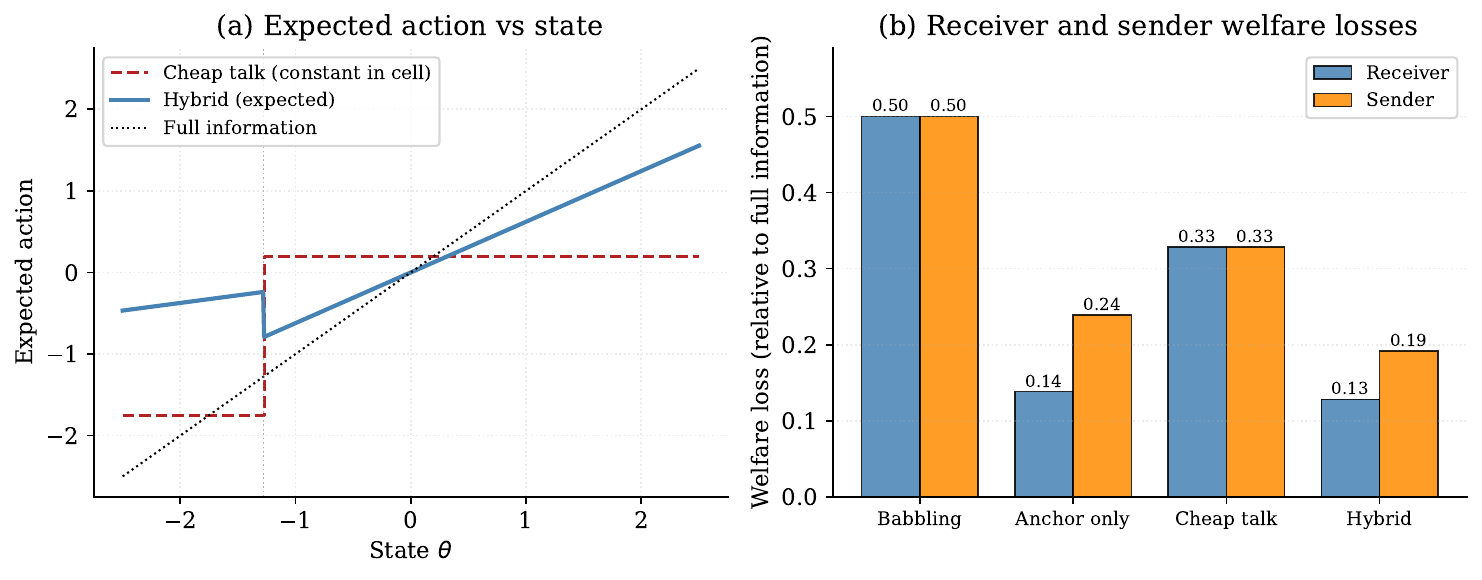}
\caption{Gaussian benchmark: $\theta\sim N(0,1)$, $d=1/2$, $c=1$, two-label cheap-talk  partition (threshold $t^\star\approx-1.28$) with within-label informative anchor $b=\theta+\sigma_a\varepsilon$, $\sigma_a=1$, $\beta=1$. Panel~(a): expected receiver action under cheap talk alone (red dashed step), the hybrid (blue), and the full-information benchmark; panel~(b): receiver and sender welfare losses for babbling, anchor-only, cheap talk, and the hybrid. At this moderate bias the hybrid strictly dominates both pure formats for both players, the opposite of the small-$d$ regime in Theorem~\ref{thm:hybrid-ranking}.}\label{fig:hybrid}
\end{figure}
}{}

\section{Conclusion}\label{sec:conclusion}

Strategic communication is usually modeled as if speakers share a common  language mapping states and messages. Treating the language as a primitive that varies across speakers changes the equilibrium picture in two manners relevant to many applications of language.

The first is institutional. Compliance regimes in accounting, auditing, and peer review are usually defended as devices that align reports with an objective truth. Once the cost structure of a report runs through a private reference rather than the state, verification becomes a choice about \emph{what} the sender is held to, not only how strictly. A reporting standard then coordinates senders on a common internal language as much as it disciplines them against falsehood. The relevant design questions are which private references to recognize as legitimate, how much heterogeneity in those references to tolerate, and when to pair quantitative reports with a coarser qualitative channel. These are questions that fall outside the pure cheap-talk and lying-cost frames.

The second is about the form of communication. Cheap talk predicts partition-valued reports, deterministic lying-cost models predict one report per state, and the mechanism here predicts continuous cross-sectional variation in reports tied to differences in private references. The same state is reported differently by senders who use different internal benchmarks, and the same number can carry different meanings under different qualitative labels. Both patterns should be visible in reporting systems in which a common underlying object is reported by multiple senders, or in which a numerical disclosure is paired with a categorical one.

Anchors offer interesting directions for future research. Cheap-talk theory typically holds the language of communication fixed and varies the strategic structure.  The current analysis treats the language itself as private, and shows that small private differences in language can produce large changes in equilibrium informativeness, including a discontinuity in the limit of vanishing reporting cost. Communication models in which speakers differ not in what they want but in how they speak, and in which the cost structure of speech is itself private, appear to be a relatively unexplored direction. 

\section*{Appendix}

\begin{proof}[\textbf{Proof of Proposition~\ref{prop:bounded-main}}]
\emph{Monotonicity.} The sender's objective $r\mapsto U^S(a(r),\theta)-D(r,b)$ has strictly increasing differences in $(r,b)$ because $-D_{12}(r,b)>0$ by Assumption~\ref{ass:cost-main}(ii). Topkis's theorem yields that $b\mapsto R(\theta,b)$ is nondecreasing on $B(\theta)$. On any smooth region of $a$, strict monotonicity then follows readily from the implicit-function theorem. Let 
\begin{equation}\label{Frbt}
F(r,b;\theta)=U^S_1(a(r),\theta)a'(r)-D_1(r,b).
\end{equation}
 At an interior optimum, $F(R(\theta,b),b;\theta)=0$ and $F_r<0$ by the strict second-order condition. Hence $\partial_b R(\theta,b)=D_{12}/F_r>0$.

\emph{Continuity.} When $M$ is compact, Berge's maximum theorem applies directly, since the objective is continuous and the feasible-report set is compact. When $M=\R$, consider $b_0\in B(\theta)$. Since $a(M)\subseteq A$, the sender's strategic gain from the receiver action is bounded by 
\begin{equation}\label{L0}
L_0=\sup_{(a,\theta)\in A\times\Theta}|U^S(a,\theta)|<\infty.
\end{equation}
 Reporting $b$ is feasible and minimizes direct cost, so Assumption~\ref{ass:cost-main}(iv) delivers $K<\infty$ such that $|r-b|>K$ implies $D(r,b)-D(b,b)>2L_0$; such an $r$ cannot be optimal. Thus $|R(\theta,b)-b|\le K$. For $b$ near $b_0$, all optimal reports therefore lie in a compact interval, and Berge's theorem implies continuity at $b_0$.
\end{proof}

\begin{proof}[\textbf{Proofs of Theorem~\ref{thm:noholes-main} and Corollary~\ref{cor:fulluse-unbounded}}]
By Proposition~\ref{prop:bounded-main}, for any $\theta$, $b\mapsto R(\theta,b)$ is continuous on $B(\theta)$. The image of a connected set under a continuous map is connected, i.e., it is an interval.
To show the Corollary, the proof of Proposition~\ref{prop:bounded-main} provides a uniform bound $|R(\theta,b)-b|\le K$, so $R(\theta,b)\to\infty$ as $b\to\infty$ and $R(\theta,b)\to-\infty$ as $b\to-\infty$. The interval $R(\theta,\R)$ is therefore unbounded in both directions and must equal $\R$.
\end{proof}

\begin{proof}[\textbf{Proof of Proposition~\ref{prop:boundedmessages}}]
(i) follows immediately from Theorem~\ref{thm:noholes-main}.
(ii) Because $M$ is compact, the sender's feasible-report set is the constant compact set $M$. The objective $U^S(a(r),\theta)-D(r,b)$
is jointly continuous in $(r,\theta,b)$, and regularity yields a unique maximizer for each $(\theta,b)$. Berge's theorem therefore makes $R(\theta,b)$ jointly continuous on the connected domain $\Theta \times \R$. A continuous image of a connected set in $\R$ is connected, so $\mathcal{U} = R(\Theta \times \R)$ is an interval in $M$.
\end{proof}

\begin{proof}[\textbf{Proof of Corollary~\ref{cor:boundedmessages-surj}}]
Let $\theta\in\Theta$. By Proposition~\ref{prop:boundedmessages}, $I_\theta=R(\theta,\R)$ is an interval in $M$. 

Consider first the claim that every interior report is used. For any $\varepsilon>0$, Assumption~\ref{ass:invert-main} yields $D_1(\overline m,b)\to-\infty$ as $b\to\infty$. Take $b$ sufficiently large to satisfy $D_1(\overline m,b)<-2L_0/\varepsilon$, where the bound $L_0$ is defined in (\ref{L0}). Since $D_{11}>0$, every $r\le \overline m-\varepsilon$ satisfies
$D(r,b)-D(\overline m,b)=-\int_r^{\overline m}D_1(s,b)\,ds>2L_0$.
Such an $r$ cannot be optimal, because replacing it with $\overline m$ changes the sender's action payoff by at most $2L_0$. Hence $R(\theta,b)>\overline m-\varepsilon$ for all sufficiently large $b$. The symmetric argument, using $D_1(\underline m,b)\to\infty$ as $b\to-\infty$, implies $R(\theta,b)<\underline m+\varepsilon$ for all sufficiently negative $b$. Thus $I_\theta$ contains reports arbitrarily close to both endpoints. Since $I_\theta$ is an interval, $(\underline m,\overline m)\subseteq I_\theta$.

Next, suppose $r\mapsto U^S(a(r),\theta)$ is Lipschitz on $M$, with constant $L_\theta$. Choose $b^+$ so large that $D_1(\overline m,b)<-L_\theta$ for all $b\ge b^+$. Since $D_{11}>0$, every $r<\overline m$ and $b\ge b^+$ satisfy
$D(r,b)-D(\overline m,b)>L_\theta(\overline m-r)$.
The Lipschitz bound yields
$U^S(a(\overline m),\theta)-U^S(a(r),\theta)\ge -L_\theta(\overline m-r)$,
so $\overline m$ strictly dominates every $r<\overline m$. Hence $R(\theta,b)=\overline m$ for all $b\ge b^+$. A symmetric argument shows that $R(\theta,b)=\underline m$ for all sufficiently negative $b$. Therefore both endpoints are attained, and $I_\theta=M$.
\end{proof}

\begin{proof}[\textbf{Proof of Corollary~\ref{cor:reversion}}]
Because $a$ is bounded and monotone, it has finite limits at $\pm\infty$. Since $a'$ is integrable on each tail and uniformly continuous there, $a'(r)\to 0$ as $|r|\to\infty$.

Let $r_n=R(\theta,b_n)$. Finitely many kink points exist, so after discarding finitely many terms, each $r_n$ lies in a smooth tail region. The sender first-order condition pins down $D_1(r_n,b_n)=U^S_1(a(r_n),\theta)\,a'(r_n)$. Since $U^S_1$ is bounded on $A\times\Theta$, the right-hand side tends to zero, so $D_1(r_n,b_n)\to 0$.
The proof of Proposition~\ref{prop:bounded-main} establishes a bound $|r_n-b_n|\le K$. Since $D_1(b_n,b_n)=0$ and Assumption~\ref{ass:cost-main}(v) yields $\underline\kappa_K=\inf_{|r-b|\le K}D_{11}(r,b)>0$, the mean value theorem produces $|D_1(r_n,b_n)|\ge \underline\kappa_K\,|r_n-b_n|$. Because the left-hand side tends to zero, $|r_n-b_n|\to 0$.
\end{proof}

\begin{proof}[\textbf{Proof of Theorem~\ref{thm:continuous-main}}]
Let $(r,\theta)\in I\times\Theta$. Since $M=\R$ and $B(\theta)=\R$, the full-use property delivers an anchor $b_\ast$ with $R(\theta,b_\ast)=r$. Because $r\in I$ lies in a smooth region of $a$ and the sender optimum is interior, $b_\ast$ satisfies $D_1(r,b_\ast)=U^S_1(a(r),\theta)a'(r)$. Assumption~\ref{ass:invert-main} states that, for fixed $r$, the map $b\mapsto D_1(r,b)$ is continuous, strictly decreasing, and onto $\R$, so there is a unique solution $b^\ast(r,\theta)$ and necessarily $b_\ast=b^\ast(r,\theta)$.

Then $F(r,b^\ast(r,\theta);\theta)=0$, where $F$ is defined in (\ref{Frbt}). Differentiating $F(r,b;\theta)=U^S_1(a(r),\theta)a'(r)-D_1(r,b)$ in $b$, $F_b=-D_{12}>0$ since $D_{12}<0$. Since $a$ is $C^3$ on $I$ and $D$ is $C^4$, the implicit-function theorem yields $b^\ast(\cdot,\theta)\in C^2(I)$. Differentiating in $r$,
\[
\partial_r b^\ast(r,\theta)
=
\frac{D_{11}(r,b^\ast(r,\theta))-U^S_{11}(a(r),\theta)(a'(r))^2-U^S_1(a(r),\theta)a''(r)}{-D_{12}(r,b^\ast(r,\theta))}.
\]
The numerator is strictly positive by the strict second-order condition and the denominator is positive because $D_{12}<0$. Hence $\partial_r b^\ast(r,\theta)>0$.

A change of variables produces the conditional report density $g(r|\theta)=q(b^\ast(r,\theta)|\theta)\,\partial_r b^\ast(r,\theta)$, which readily proves \eqref{eq:posterior-main}. Substituting the posterior into the receiver first-order condition delivers \eqref{eq:continuous-main}.
\end{proof}

\begin{proof}[\textbf{Proof of Proposition~\ref{prop:no-babbling-smallcost}}]
Let $s(\theta)=U^R_1(a^0,\theta)$. Then $\int_\Theta s(\theta)f(\theta)d\theta=0$, and $s$ is strictly increasing because $U^R_{12}\ge\underline\gamma>0$. In a smooth regular equilibrium, $a_c$ has no kinks, so by Proposition~\ref{prop:bounded-main} the map $b\mapsto R_c(\theta,b)$ is strictly increasing on $\R$. It is also nondecreasing in $\theta$: since $a_c$ is increasing and $U^S_{12}\ge0$, $r\mapsto U^S(a_c(r),\theta)$ has increasing differences in $(r,\theta)$, so Topkis theorem applies.

Consider interior states $\underline\theta<\theta^-<\theta^+<\overline\theta$ with $\Delta=b_0(\theta^+)-b_0(\theta^-)>0$. Let $F_x$ be the distribution function of $x$, and choose $t$ such that $F_x(t)-F_x(t-\Delta)>0$. With $\bar r_c=R_c(\theta^-,b_0(\theta^-)+t)$ and $E_c=\{(\theta,b):R_c(\theta,b)\le\bar r_c\}$, monotonicity implies $\Pr(E_c|\theta)\ge F_x(t)$ for $\theta\le\theta^-$ and $\Pr(E_c|\theta)\le F_x(t-\Delta)$ for $\theta\ge\theta^+$. Thus $\theta\mapsto\Pr(E_c|\theta)$ is nonincreasing.

By the FKG-Chebyshev covariance identity for monotone functions,
\[
\Cov(s(\theta),1_{E_c})
=
-\frac12
\int_\Theta\int_\Theta
(s(\theta')-s(\theta))
(\Pr(E_c|\theta)-\Pr(E_c|\theta'))
f(\theta)f(\theta')d\theta d\theta'
\]
with nonnegative integrand because $s$ is increasing and $\theta\mapsto\Pr(E_c|\theta)$ is nonincreasing. Define
\[
\delta_s=s(\theta^+)-s(\theta^-)>0,
\qquad
\delta_x=F_x(t)-F_x(t-\Delta)>0,
\]
and
\[
m_-=\int_{\underline\theta}^{\theta^-}f(\theta)d\theta>0,
\qquad
m_+=\int_{\theta^+}^{\overline\theta}f(\theta)d\theta>0.
\]
For every $\theta\le\theta^-$ and $\theta'\ge\theta^+$, $s(\theta')-s(\theta)\ge\delta_s$ and $\Pr(E_c|\theta)-\Pr(E_c|\theta')\ge\delta_x$, so the integrand is at least $\delta_s\delta_x$ on the rectangle $[\underline\theta,\theta^-]\times[\theta^+,\overline\theta]$; the same bound holds on the reversed rectangle. The double integral is therefore at least $2\delta_s\delta_x m_-m_+$, and
\[
\Cov(s(\theta),1_{E_c})\le-\delta_s\delta_x m_-m_+\equiv-\kappa.
\]
The constant $\kappa>0$ depends only on $s$, the prior $f$, the chosen $\theta^-,\theta^+,t$, and the noise distribution of $x$, not on $c$ or on the particular smooth regular equilibrium. Since $\int_\Theta s(\theta)f(\theta)d\theta=0$, $\Cov(s,1_{E_c})=\E[s(\theta)1_{E_c}]$, and the derivative of the receiver's expected payoff at $a^0$ is at most $-\kappa/\Pr(E_c)$ after $E_c$ and at least $\kappa/\Pr(E_c^c)$ after $E_c^c$.

Let $L=\sup_{(a,\theta)\in A\times\Theta}|U^R_{11}(a,\theta)|<\infty$. For any posterior $\mu$ on $\Theta$, the receiver's expected payoff $V_\mu(a)=\int_\Theta U^R(a,\theta)\,\mu(d\theta)$ has $V_\mu'(a^0)=\int_\Theta U^R_1(a^0,\theta)\,\mu(d\theta)\equiv g_\mu$ and $V_\mu''(a)=\int_\Theta U^R_{11}(a,\theta)\,\mu(d\theta)\ge -L$ on $A=[a^R(\underline\theta),a^R(\overline\theta)]$. Since $a^R$ is the unique maximizer of $U^R(\cdot,\theta)$ and $U^R$ is concave in $a$, $U^R_1(a^R(\underline\theta),\theta)\ge0$ and $U^R_1(a^R(\overline\theta),\theta)\le0$ for every $\theta\in\Theta$, hence $V_\mu'(a^R(\underline\theta))\ge0$ and $V_\mu'(a^R(\overline\theta))\le0$ for every posterior $\mu$. If $g_\mu>0$, the bound $V_\mu''\ge -L$ returns $V_\mu'(a^R(\overline\theta))\ge g_\mu-L(a^R(\overline\theta)-a^0)$, and combined with $V_\mu'(a^R(\overline\theta))\le0$, this yields $a^0+g_\mu/L\le a^R(\overline\theta)$. The case $g_\mu<0$ is symmetric using $V_\mu'(a^R(\underline\theta))\ge0$. Hence the action $a^0+g_\mu/L$ is in $A$. Taylor's theorem at $a^0$ then yields, for some $\xi$ between $a^0$ and $a^0+g_\mu/L$,
\[
V_\mu(a^0+g_\mu/L)-V_\mu(a^0)
=\frac{g_\mu^2}{L}+\tfrac12 V_\mu''(\xi)\!\left(\frac{g_\mu}{L}\right)^2
\ge\frac{g_\mu^2}{L}-\frac{g_\mu^2}{2L}=\frac{g_\mu^2}{2L},
\]
so the receiver's best response under $\mu$ improves on $a^0$ by at least $g_\mu^2/(2L)$. Applied to the posterior conditional on $E_c$, for which $|g_\mu|\ge\kappa/\Pr(E_c)$, the $E_c$ information set contributes
\[
\Pr(E_c)\cdot\frac{1}{2L}\!\left(\frac{\kappa}{\Pr(E_c)}\right)^2
=\frac{\kappa^2}{2L\,\Pr(E_c)}\ge\frac{\kappa^2}{2L}
\]
to the ex-ante receiver payoff over using $a^0$. \end{proof}

\begin{proof}[\textbf{Proof of Theorem~\ref{thm:ode-main}}]
Differentiating $b_a(r,\theta)=r-\psi(U^S_1(a(r),\theta)a'(r)/c)$ in $r$ at fixed $\theta$,
\[
\partial_r b_a(r,\theta)
=
1-\frac{1}{c}\,\psi'\!\left(\frac{U^S_1(a(r),\theta)a'(r)}{c}\right)\!
\left[
U^S_{11}(a(r),\theta)\,(a'(r))^2+U^S_1(a(r),\theta)\,a''(r)
\right].
\]
Substituting this into Bayes integral and collecting terms,
\[
\int_\Theta U^R_1(a(r),\theta)\,w_a^\sigma(r,\theta)\,d\theta
=
\mathcal B(r,a(r),a'(r))
-
\mathcal A(r,a(r),a'(r))\,a''(r).
\]
Hence the Bayes condition holds if and only if \eqref{eq:ode-main2} holds. 
\end{proof}

\begin{proof}[\textbf{Proof of Proposition~\ref{prop:small-cost-revelation}}]
Let $\rho_c$ denote the deterministic full-information separating language, the unique strictly increasing solution on $\Theta$ of
\begin{equation}\label{eq:det-separating-ode}
c\phi'(\rho_c(\theta)-b_0(\theta))\,\rho_c'(\theta)=T(\theta),
\qquad
\rho_c(\underline\theta)=b_0(\underline\theta),
\end{equation}
the separating equation that would hold without anchor noise, with the boundary condition selecting the least-cost separation.\footnote{Existence and uniqueness of this boundary solution are shown in Corollary~\ref{cor:det-least} of the supplementary appendix.}

Let $d_c(\theta)=\rho_c(\theta)-b_0(\theta)$ and $Y_c(\theta)=c\phi(d_c(\theta))$. Differentiating \(Y_c\) and using
\eqref{eq:det-separating-ode},
\[
\begin{aligned}
Y_c'(\theta)
&=
c\phi'(d_c(\theta))d_c'(\theta)=
T(\theta)-c\phi'(d_c(\theta))b_0'(\theta).
\end{aligned}
\]
Since $\Gamma'(\theta)=T(\theta)$, $(\Gamma-Y_c)'(\theta)
=
c\phi'(d_c(\theta))b_0'(\theta)\geq 0$. Given that $\Gamma(\underline{\theta})=Y_c(\underline{\theta})=0$, $0\le Y_c(\theta)\le \Gamma(\theta)$ for all $\theta\in\Theta$. The curvature bounds in Assumption \ref{ass:PB-main} readily imply, by integration and given that $\phi(0)=\phi'(0)=0$, that $\phi'(x)\le \overline\kappa x$ and $\phi(x)\ge \frac{\underline\kappa}{2}x^2$, therefore, for $x\geq 0$, $\phi'(x)\le C_0\sqrt{\phi(x)}$  with $C_0=\overline\kappa\sqrt{\frac{2}{\underline\kappa}}$.
Applying this to \(x=d_c(\theta)\),
\[
c\phi'(d_c(\theta))
\le
C_0c\sqrt{\phi(d_c(\theta))}
=
C_0\sqrt{cY_c(\theta)}
\le
C_0\sqrt{c\Gamma(\theta)}.
\]
Since \(\Gamma\) is bounded on the compact set \(\Theta\), this implies
the existence of $C_1>0$ such that $
c\phi'(d_c(\theta))\le C_1\sqrt c.$ Integrating the identity $(\Gamma-Y_c)'(\theta)
=
c\phi'(d_c(\theta))b_0'(\theta)$
and using boundedness of \(b_0'\) yields $0\le \Gamma(\theta)-Y_c(\theta)\le C_1\sqrt c $
uniformly on \(\Theta\). Hence $Y_c\to \Gamma$ uniformly.

Now fix any compact \(K\subset \operatorname{Int}\Theta\) and let $\gamma_K=\inf_{\theta\in K}\Gamma(\theta)>0$.
Uniform convergence delivers \(Y_c(\theta)\ge \gamma_K/2\) on \(K\) for small \(c\). Thus
\[
\phi(d_c(\theta))
=
\frac{Y_c(\theta)}{c}
\ge
\frac{\gamma_K}{2c}
\to\infty
\]
uniformly on \(K\). This implies $d_c(\theta)\to\infty$ uniformly on \(K\); hence also \(\phi'(d_c(\theta))\to\infty\) uniformly on \(K\).
Equation~\eqref{eq:det-separating-ode} also implies $c\rho_c'(\theta)=T(\theta)/\phi'(d_c(\theta))\to0$ and $(a^R)'(\theta)/(c\rho_c'(\theta))\to\infty$, uniformly on compact subsets of the interior.

Let \(\tau_c\) be defined by $R_c(\theta,b)=\rho_c(\tau_c)$.
The sender's first-order condition evaluated at \(r=\rho_c(\tau)\) and
\(b=b_0(\theta)+\sigma x\) can be written
\begin{equation}\label{jb1}
G_c(\tau;\theta,x)
=c\phi'(\rho_c(\tau)-b_0(\theta)-\sigma x)
-U^S_1(a_c(\rho_c(\tau)),\theta)a_c'(\rho_c(\tau))=0.
\end{equation}
For any \(\tau\), by the mean-value theorem and Assumption~\ref{ass:PB-main},
\[
\begin{aligned}
|G_c(\tau;\theta,x)-G_c(\tau;\theta,0)|
&=
c\left|
\phi'(\rho_c(\tau)-b_0(\theta)-\sigma x)
-
\phi'(\rho_c(\tau)-b_0(\theta))
\right| \\
&\le
c|\sigma|\overline\kappa |x|
=
O(c|x|).
\end{aligned}
\]

Assumption~\ref{ass:branch-tracking} provides the \(C^1\) convergence
$a_c(\rho_c(t))=a^R(t)+o(1)$, and
$\partial_t a_c(\rho_c(t))=(a^R)'(t)+o(1)$,
uniformly on compact subsets. Hence, by the chain rule,
\begin{equation}\label{jb2}
a_c'(\rho_c(t))
=
\frac{(a^R)'(t)+o(1)}{\rho_c'(t)}.
\end{equation}
Therefore,
\[
G_c(\theta;\theta,0)
=
o\left(\frac{1}{\rho_c'(\theta)}\right).
\]

Combining \eqref{jb1} and \eqref{jb2},
\[
G_c(\tau;\theta,0)
=
\underbrace{
c\phi'(\rho_c(\tau)-b_0(\theta))
-
\frac{
U^S_1(a^R(\tau),\theta)(a^R)'(\tau)
}{
\rho_c'(\tau)
}
}_{\overline{G}_c(\tau,\theta)}
+
o\left(\frac{1}{\rho_c'(\tau)}\right).
\]

Differentiating the leading term \(\overline G_c\),
\begin{equation}\label{tru1}
\begin{aligned}
\partial_\tau \overline G_c(\theta,\theta)
&=
c\phi''(\rho_c(\theta)-b_0(\theta))\rho_c'(\theta) \\
&\quad
+
\frac{
U^S_1(a^R(\theta),\theta)(a^R)'(\theta)\rho_c''(\theta)
}{
(\rho_c'(\theta))^2
} \\
&\quad
-
\frac{
U^S_{11}(a^R(\theta),\theta)((a^R)'(\theta))^2
+
U^S_1(a^R(\theta),\theta)(a^R)''(\theta)
}{
\rho_c'(\theta)
}.
\end{aligned}
\end{equation}

The defining ODE for \(\rho_c\) is
\begin{equation}\label{ode1}
c\phi'(\rho_c(t)-b_0(t))\rho_c'(t)
=
U^S_1(a^R(t),t)(a^R)'(t).
\end{equation}
Differentiating \eqref{ode1} and evaluating at \(t=\theta\) yields\begin{equation}\label{odediff}
\begin{aligned}
&c\phi''(\rho_c(\theta)-b_0(\theta))
(\rho_c'(\theta)-b_0'(\theta))\rho_c'(\theta)
+
c\phi'(\rho_c(\theta)-b_0(\theta))\rho_c''(\theta)
\\
&\qquad =
U^S_{11}(a^R(\theta),\theta)((a^R)'(\theta))^2
+
U^S_{12}(a^R(\theta),\theta)(a^R)'(\theta)
+
U^S_1(a^R(\theta),\theta)(a^R)''(\theta).
\end{aligned}
\end{equation}
Using \eqref{ode1}, this becomes
\begin{equation}\label{odediff-div}
\begin{aligned}
&c\phi''(\rho_c(\theta)-b_0(\theta))\rho_c'(\theta)
-
c\phi''(\rho_c(\theta)-b_0(\theta))b_0'(\theta)
+
\frac{
U^S_1(a^R(\theta),\theta)(a^R)'(\theta)\rho_c''(\theta)
}{
(\rho_c'(\theta))^2
}
\\
&\qquad =
\frac{
U^S_{11}(a^R(\theta),\theta)((a^R)'(\theta))^2
+
U^S_{12}(a^R(\theta),\theta)(a^R)'(\theta)
+
U^S_1(a^R(\theta),\theta)(a^R)''(\theta)
}{
\rho_c'(\theta)
}.
\end{aligned}
\end{equation}
Rearranging \eqref{odediff-div}, the left-hand side of \eqref{tru1} equals
\[
c\phi''(\rho_c(\theta)-b_0(\theta))b_0'(\theta)
+
\frac{
U^S_{12}(a^R(\theta),\theta)(a^R)'(\theta)
}{
\rho_c'(\theta)
}.
\]
Thus
\[
\partial_\tau \overline G_c(\theta,\theta)
=
c\phi''(\rho_c(\theta)-b_0(\theta))b_0'(\theta)
+
\frac{
U^S_{12}(a^R(\theta),\theta)(a^R)'(\theta)
}{
\rho_c'(\theta)
}.
\]

By Assumption~\ref{ass:smallcost-main}, \(U^S_{12}>0\), and since
\((a^R)'>0\), \(\rho_c'>0\), \(\phi''>0\), and \(b_0'>0\), the right-hand
side is positive. Moreover, for every compact
\(K\Subset\operatorname{Int}\Theta\), there is \(m_K>0\) such that
\[
\partial_\tau \overline G_c(\theta,\theta)
\ge
\frac{m_K}{\rho_c'(\theta)}
\qquad\text{for all }\theta\in K.
\]
By continuity, the same lower bound holds for
\(\tau\) in a sufficiently small neighborhood of \(\theta\), uniformly for
\(\theta\in K\). Since \(\overline G_c(\theta,\theta)=0\), the mean-value
theorem gives, for every sufficiently small \(h>0\),
\[
\overline G_c(\theta-h,\theta)
\le
-\frac{m_Kh}{\rho_c'(\theta)}
<
0
<
\frac{m_Kh}{\rho_c'(\theta)}
\le
\overline G_c(\theta+h,\theta).
\]
Using the expansion of \(G_c(\tau;\theta,0)\) around \(\overline G_c(\tau,\theta)\)
uniformly on this local neighborhood, the same signs hold for\[
G_c(\theta-h;\theta,0)<0<G_c(\theta+h;\theta,0)
\]
for all sufficiently small \(c\).

Finally, since $
|G_c(\tau;\theta,x)-G_c(\tau;\theta,0)|=O(c|x|)$
and \(c\rho_c'(\theta)\to0\) uniformly on compact subsets,
\[
c|x|
=
o_p\left(\frac{1}{\rho_c'(\theta)}\right).
\]
Hence the shock perturbation is smaller than the sign gap
\(m_Kh/\rho_c'(\theta)\) with probability tending to one. Therefore
$G_c(\theta-h;\theta,x)<0<G_c(\theta+h;\theta,x)$
with probability tending to one.
Hence, with
probability tending to one, $\theta-h<\tau_c<\theta+h$.
Choosing \(h<\varepsilon\), 
$\sup_{\theta\in K}
\Pr\left(|\tau_c-\theta|>\varepsilon\mid\theta\right)
\to0$. Thus \(\tau_c\to\theta\) in probability, uniformly for \(\theta\in K\) and $a_c(R_c(\theta,b))\to a^R(\theta)$ follows immediately.
\end{proof}

\begin{proof}[\textbf{Proof of Theorem~\ref{thm:cost-floor}}]
Recall that $d_c(\theta)=\rho_c(\theta)-b_0(\theta)$ and $Y_c(\theta)=c\phi(d_c(\theta))$. The first paragraph of the proof of Proposition~\ref{prop:small-cost-revelation} establishes $Y_c\to\Gamma$ uniformly, and its last paragraph supplies $c\phi(R_c(\theta,b)-b)-c\phi(d_c(\theta))\to0$ in probability with a uniform $L^1$ bound. Dominated convergence then yields $\E[c\phi(R_c(\theta,b)-b)]-\E[Y_c(\theta)]\to0$, and uniform convergence of $Y_c$ to $\Gamma$ pins down $\E[c\phi(R_c(\theta,b)-b)]\to\E[\Gamma]$.

Since $T\ge0$, exchanging the order of integration produces
\[
\E[\Gamma(\theta)]
=
\int_\Theta\int_{\underline\theta}^{\theta}T(u)\,du\,dF(\theta)
=
\int_\Theta(1-F(u))T(u)\,du,
\]
with both sides interpreted in $[0,\infty]$.
If $\underline\theta=-\infty$ and $T(\theta)\ge\eta>0$ for all $\theta\le\theta_0$, then $\int_{-\infty}^{\min\{\theta,\theta_0\}}T(u)\,du=\infty$ for every finite $\theta$, so $\Gamma(\theta)=\infty$ a.s.\ and $\E[\Gamma(\theta)]=\infty$.
\end{proof}

\begin{proof}[\textbf{Proof of Theorem~\ref{thm:uninf-existence}}]
Denoting $\psi=(\phi')^{-1}$, for an increasing action rule $a$, the sender's first-order condition is
\begin{equation}\label{eq:def-ba}
b_a(r,\theta)=r-\psi\!\left(\frac{U^S_1(a(r),\theta)\,a'(r)}{c}\right),
\end{equation}

and Bayes' rule joint with receiver optimality require
\begin{equation}\label{bayesjb}
\int_\Theta U^R_1(a(r),\theta)\,f(\theta)\,q(b_a(r,\theta))\,\partial_r b_a(r,\theta)\,d\theta=0.
\end{equation}
Consider perturbed babbling with $a(r)=a^0+\varepsilon v(r)+o(\varepsilon)$.
A Taylor expansion jointly with \(\psi'(0)=1/\phi''(0)\) yields
\begin{equation*}
\label{eq:psi-expanded}
\psi\!\left(
\frac{U^S_1(a(r),\theta)a'(r)}{c}
\right)
=
\varepsilon
\frac{U^S_1(a^0,\theta)v'(r)}
{c\phi''(0)}
+
o(\varepsilon).
\end{equation*}

Substituting the above into the definition of \(b_a\) in
\eqref{eq:def-ba} gives
\begin{eqnarray*}
\label{eq:b-expanded}
b_a(r,\theta)
&=&
r-
\varepsilon
\frac{U^S_1(a^0,\theta)v'(r)}
{c\phi''(0)}
+
o(\varepsilon)\\
\label{eq:br-expanded}
\partial_r b_a(r,\theta)
&=&
1-
\varepsilon
\frac{U^S_1(a^0,\theta)v''(r)}
{c\phi''(0)}
+
o(\varepsilon).
\end{eqnarray*}

So that using the previous equation on a Taylor expansion of $q$ at $r$,
\begin{equation}
\label{eq:q-expanded}
q(b_a(r,\theta))
=
q(r)
-
\varepsilon
\frac{U^S_1(a^0,\theta)v'(r)q'(r)}
{c\phi''(0)}
+
o(\varepsilon).
\end{equation}
Another Taylor expansion yields
\begin{equation}
\label{eq:UR-expanded}
U^R_1(a(r),\theta)
=
U^R_1(a^0,\theta)
+
\varepsilon U^R_{11}(a^0,\theta)v(r)
+
o(\varepsilon).
\end{equation}

Substituting these expressions into \eqref{bayesjb} and letting \(\varepsilon\to0\) simplifies to
\begin{equation}
\label{eq:first-order-condition}
q(r)v(r)\E[U^R_{11}(a^0,\theta)]
-
\frac{v'(r)q'(r)+v''(r)q(r)}
{c\phi''(0)}
\E[U^R_1(a^0,\theta)U^S_1(a^0,\theta)]
=
0.
\end{equation}

Using the product rule, $v'(r)q'(r)+v''(r)q(r)
=
(q(r)v'(r))'$ and rearranging reduces to the Sturm--Liouville equation
\begin{equation}\label{sturm}
-\underbrace{\E[U^R_1(a^0,\theta)U^S_1(a^0,\theta)]}_{A}\,(q\,v')'\;=\;c\,\phi''(0)\,\,q\,v \underbrace{(-\E[U^R_{11}(a^0,\theta)])}_{H}.
\end{equation}
By Assumption~\ref{ass:uninf-align}, $\theta\mapsto U^R_1(a^0,\theta)$ is strictly increasing and $\theta\mapsto U^S_1(a^0,\theta)$ is weakly increasing and not constant. Since $\E[U^R_1(a^0,\theta)]=0$ at the prior-optimal action, Chebyshev's covariance inequality yields $A=\Cov(U^R_1,U^S_1)>0$.

Let $L_qv:=-q^{-1}(qv')'$ on $L^2(q)$, where $q$ is the anchor density. In the real-line case of Assumption~\ref{ass:uninf-noise}, $q\propto e^{-V_b}$ with $V_b''$ bounded below by a positive constant, so $L_q$ is essentially self-adjoint and its spectrum is discrete: the ground-state transformation maps $L_q$ to a Schr\"odinger operator with potential $(V_b')^2/4-V_b''/2$, which diverges as $|x|\to\infty$ because the lower bound on $V_b''$ forces $|V_b'|$ to grow linearly. The Poincar\'e inequality for strongly log-concave measures yields a first nonzero eigenvalue \(\lambda_1>0\), which is simple. An associated eigenfunction \(v_1\) with
$L_qv_1=\lambda_1v_1$ may be chosen strictly increasing: $w=qv_1'$ satisfies $w'=-\lambda_1qv_1$ and vanishes at $\pm\infty$, and since $v_1$ is orthogonal to constants it changes sign exactly once, so $w$ increases while $v_1<0$ and decreases while $v_1>0$, hence $w>0$ everywhere and $v_1'>0$. Then, \eqref{sturm} can be written as \(\mathcal L_c v=0\), where $L_c v=A L_qv-c\phi''(0)Hv$. Setting
$
\bar c=\frac{A\lambda_1}{H\phi''(0)}
$
returns $$
\mathcal L_{\bar c}v_1
=
A L_qv_1-\bar c\,\phi''(0)Hv_1
=
A\lambda_1v_1-\bar c\,\phi''(0)Hv_1
=
0.
$$
Thus \(v_1\in\ker\mathcal L_{\bar c}\). Conversely,
\(\mathcal L_{\bar c}v=0\) is equivalent to \(L_qv=\lambda_1v\). Since
\(\lambda_1\) is simple, $\ker\mathcal L_{\bar c}=\operatorname{span}\{v_1\}$.

Let \(\mathcal F(c,a)=0\) denote the normalized nonlinear equilibrium equation
corresponding to \eqref{bayesjb}, viewed as a $C^1$ map on a weighted Sobolev space based on $L^2(q)$ on which $\mathcal L_c$, being self-adjoint with discrete spectrum, is Fredholm of index zero, so the Crandall--Rabinowitz theorem applies. The expansion above implies $D_a\mathcal F(c,a^0)=\mathcal L_c$. The Crandall--Rabinowitz transversality condition requires $\partial_c\mathcal L_{\bar c}v_1\notin
\operatorname{Range}(\mathcal L_{\bar c})$.
Since \(\mathcal L_{\bar c}\) is self-adjoint, its adjoint kernel is also
\(\operatorname{span}\{v_1\}\). Hence it suffices to check that
\[
\left\langle v_1,\partial_c\mathcal L_{\bar c}v_1\right\rangle_{L^2(q)}
=
-\phi''(0)H\int v_1(r)^2q(r)\,dr
\neq 0.
\]
This is nonzero because \(\phi''(0)>0\), \(H>0\), \(q>0\), and
\(v_1\not\equiv0\). Therefore the transversality
condition holds. The theorem yields a \(C^1\) local branch of nontrivial
solutions \((c_\varepsilon,a_\varepsilon)\) bifurcating from
\((\bar c,a^0)\), with
\[
a_\varepsilon(r)=a^0+\varepsilon v_1(r)+o(\varepsilon),
\qquad
c_\varepsilon\to\bar c.
\]

For small nonzero $\varepsilon$, the action rule is strictly increasing ($v_1'>0$ everywhere, and the branch remainder is $o(\varepsilon)$ in the $C^1$ norm of the construction), $\partial_r b_{a_\varepsilon}>0$, and the sender's second-order condition is strict since the cost curvature is bounded away from zero while the action-rule perturbation is of order $\varepsilon$. Each nearby nontrivial solution is therefore a regular equilibrium. Since $a_\varepsilon$ is nonconstant, the receiver's posterior is not the prior almost surely. Taking any sufficiently small nonzero $\varepsilon$ and setting $c^\ast=c_\varepsilon$ delivers the claimed existence of a positive cost level supporting informative communication.
\end{proof}

\begin{proof}[\textbf{Proof of Theorem~\ref{thm:uninf-fullinfo}}]
I use the maintained convention stated after Definition~\ref{def:regular-main} that regular equilibria exist, and in the additive real-line case the selected branch is the smooth global branch characterized by Theorem~\ref{thm:ode-main}. 

Recall that
\(T(\theta)=U^S_1(a^R(\theta),\theta)(a^R)'(\theta)\) and
\(\Gamma(\theta)=\int_{\underline\theta}^\theta T(u)\,du\). Fix
\(\kappa>0\), and define \(\rho_\kappa\) by
\[
c\,\phi(\rho_\kappa(\theta)-\mu)=\kappa+\Gamma(\theta),
\qquad \rho_\kappa(\theta)\ge\mu .
\]
By the maintained assumptions, \(\rho_\kappa\) is well defined and strictly
increasing. The role of \(\kappa\) is to avoid the endpoint degeneracy at
\(\theta=\underline\theta\), where \(\Gamma(\underline\theta)=0\) and
\(\phi'(0)=0\).
For the deterministic anchor \(b=\mu\), if type \(\theta\) reports as type
\(t\), her payoff is 
$U^S(a^R(t),\theta)-\kappa-\Gamma(t)$.
Differentiating with respect to \(t\):
\[
(a^R)'(t)
\bigl\{
U^S_1(a^R(t),\theta)-U^S_1(a^R(t),t)
\bigr\}.
\]
By \(U^S_{12}>0\), this derivative is positive for \(t<\theta\) and negative
for \(t>\theta\). Hence truthful separation is globally incentive compatible,
and after report \(\rho_\kappa(\theta)\), the receiver chooses \(a^R(\theta)\).

For the deterministic action rule \(a=a^R\circ\rho_\kappa^{-1}\), the defining
equation for \(\rho_\kappa\) is $b_a(\rho_\kappa(t),t)=\mu$. Moreover,
\[
\left.
\partial_\theta b_a(\rho_\kappa(t),\theta)
\right|_{\theta=t}
=
-\psi'\!\left(\phi'(\rho_\kappa(t)-\mu)\right)
\frac{U^S_{12}(a^R(t),t)(a^R)'(t)}
{c\rho_\kappa'(t)}
<0.
\]
Thus \(b_a(\rho_\kappa(t),\theta)\) crosses \(\mu\) with nonzero slope at
\(\theta=t\). Hence, under \(b=\mu+\sigma x\), the likelihood
term \(q_\sigma(b_a(\rho_\kappa(t),\theta))\) concentrates at \(\theta=t\) as
\(\sigma\downarrow0\). Write \((R_{\kappa,\sigma},a_{\kappa,\sigma})\) for the regular equilibria along the selected smooth global branch under \(b=\mu+\sigma x\). The \(\sigma\downarrow0\) analogue of Assumption~\ref{ass:branch-tracking} maintained in the statement gives
\[
a_{\kappa,\sigma}(\rho_\kappa(t))=a^R(t)+o(1),\qquad
\frac{d}{dt}a_{\kappa,\sigma}(\rho_\kappa(t))=(a^R)'(t)+o(1),
\]
uniformly on compact sets. Let \(\tau_{\kappa,\sigma}\) be defined by
\(R_{\kappa,\sigma}(\theta,\mu+\sigma x)=\rho_\kappa(\tau_{\kappa,\sigma})\). The sender first-order condition at \(r=\rho_\kappa(\tau)\) is
\[
G_{\kappa,\sigma}(\tau;\theta,x)
=
c\phi'(\rho_\kappa(\tau)-\mu-\sigma x)
-U^S_1(a_{\kappa,\sigma}(\rho_\kappa(\tau)),\theta)
 a_{\kappa,\sigma}'(\rho_\kappa(\tau))=0.
\]
Using the chain rule and the preceding convergence,
\[
G_{\kappa,\sigma}(\tau;\theta,0)
=
\bar G_\kappa(\tau,\theta)+o\left(\frac{1}{\rho_\kappa'(\tau)}\right),
\]
uniformly on compact sets, where
\[
\bar G_\kappa(\tau,\theta)
=
c\phi'(\rho_\kappa(\tau)-\mu)-
\frac{U^S_1(a^R(\tau),\theta)(a^R)'(\tau)}{\rho_\kappa'(\tau)} .
\]
At \(\tau=\theta\), \(\bar G_\kappa(\theta,\theta)=0\). Differentiating the defining equation for \(\rho_\kappa\) yields
\[
\partial_\tau \bar G_\kappa(\theta,\theta)
=
\frac{U^S_{12}(a^R(\theta),\theta)(a^R)'(\theta)}{\rho_\kappa'(\theta)}>0.
\]
For every compact subset of \(\operatorname{Int}\Theta\), this derivative is bounded away from zero. Moreover,
\[
|G_{\kappa,\sigma}(\tau;\theta,x)-G_{\kappa,\sigma}(\tau;\theta,0)|=O(\sigma |x|).
\]
The same bracketing argument used in the proof of Proposition~\ref{prop:small-cost-revelation} therefore implies that, for every compact set of noise realizations,
\[
a_{\kappa,\sigma}(R_{\kappa,\sigma}(\theta,\mu+\sigma x))\to a^R(\theta), \quad
R_{\kappa,\sigma}(\theta,\mu+\sigma x)-(\mu+\sigma x)\to \rho_\kappa(\theta)-\mu
\]
in probability, uniformly on compact subsets of \(\operatorname{Int}\Theta\); compactness and monotonicity extend the convergence in probability to \(\Theta\).
Consequently,
\[
c\phi\bigl(R_{\kappa,\sigma}(\theta,\mu+\sigma x)-(\mu+\sigma x)\bigr)
\to
c\phi(\rho_\kappa(\theta)-\mu)
=
\kappa+\Gamma(\theta)
\]
in probability. The sender can always report the anchor \(b\), so the equilibrium reporting cost is bounded above by \(2\sup_{(a,\theta)\in A\times\Theta}|U^S(a,\theta)|\). Hence, integrating both sides,
\[ \E\!\left[c\phi\bigl(R_{\kappa,\sigma}(\theta,\mu+\sigma x)-(\mu+\sigma x)\bigr)\right]
\to
\kappa+\E[\Gamma(\theta)] .
\]

Finally choose a sequence \(\kappa_n\downarrow0\). For each \(n\), choose
\(\sigma_n>0\) small enough that the fixed-\(\kappa_n\) convergence above
holds within \(1/n\) whenever \(\sigma<\sigma_n\). Taking
\(\kappa(\sigma)=\kappa_n\) along such a diagonal sequence yields 
$a_{\kappa(\sigma),\sigma}(R_{\kappa(\sigma),\sigma}(\theta,\mu+\sigma x))
\to a^R(\theta)$
in probability and sends the expected reporting cost to
\(\E[\Gamma(\theta)]\).\end{proof}

\begin{proof}[\textbf{Proof of Theorem~\ref{thm:rello-main}}]
First consider any regular equilibrium, not necessarily affine.  Since the anchor support is all of $\R$ and the sender's payoff from the receiver action is bounded above for fixed $\theta$, the optimal report satisfies $R(\theta,b)\to\pm\infty$ as $b\to\pm\infty$.  Monotonicity and continuity in the anchor imply that every report is used by every state.  If report $r$ is used by type $\theta$, the first-order condition $-(a(r)-\theta-d)\,a'(r)-c(r-b)=0$ returns the unique anchor that induces it:
\[
b_r(\theta)=r+\frac{a'(r)}{c}\bigl(a(r)-\theta-d\bigr)
=r+\frac{a'(r)}{c}\bigl(a(r)-d\bigr)-\frac{a'(r)}{c}\theta.
\]
Thus, for each report, the inducing-anchor locus is affine in $\theta$ with slope $-a'(r)/c$.

Take two reports $r_1$ and $r_2$.  If $a'(r_1)\neq a'(r_2)$, the two affine loci $b_{r_1}(\theta)$ and $b_{r_2}(\theta)$ cross at some $\theta\in\R$.  At the crossing, the same sender type $(\theta,b)$ would have both reports as optimal, contradicting uniqueness of the sender's best reply in a regular equilibrium.  Hence $a'(r)$ is constant on every smooth region.  Continuity and finitely many kink points imply a common constant slope on all of $\R$, so $a(r)=\alpha r+\delta_0$ with $\alpha>0$.

For an affine action rule, the sender's first-order condition simplifies to
\[
R(\theta,b)=\frac{\alpha\theta+cb+\alpha d-\alpha\delta_0}{c+\alpha^2}=\frac{c\beta+\alpha}{c+\alpha^2}\theta+\frac{c\sigma}{c+\alpha^2}x+\frac{\alpha(d-\delta_0)}{c+\alpha^2}.
\]
The receiver's best response is the posterior mean, so
$\alpha=\frac{\Cov(\theta,r)}{\Var(r)}$. Substitution produces equation~\eqref{eq:alpha-beta}. For $\beta>0$, this equation has value $-c\beta\sigma_\theta^2<0$ at $\alpha=0$ and positive leading coefficient, so it has a unique positive root. Since its value at $1/\beta$ is $c\sigma^2/\beta>0$, the positive root lies in $(0,1/\beta)$.

Finally, equilibrium requires $\delta_0=-\alpha\E[r]$. Since $\E[\theta]=\E[x]=0$, $\E[r]=\frac{\alpha(d-\delta_0)}{c+\alpha^2}$, hence $\delta_0=-\alpha^2 d/c$. Substituting this intercept recovers the equilibrium report.
\end{proof}

\begin{proof}[\textbf{Proof of Corollary~\ref{cor:PA-comp}}]
Define\[
F(\alpha;c,\sigma,\beta)=\beta\sigma_\theta^2\alpha^2+(c\sigma^2+c\beta^2\sigma_\theta^2-\sigma_\theta^2)\alpha-c\beta\sigma_\theta^2.
\]
Theorem~\ref{thm:rello-main} implies that $0<\alpha<1/\beta$, and $F$ is convex in $\alpha$ with $F(0)<0$ and $F(1/\beta)=c\sigma^2/\beta>0$, so $F_\alpha>0$ at the equilibrium root. Rearranging produces the identity
\begin{equation}\label{eq:F-identity}
\sigma_\theta^2\alpha(1-\beta\alpha)=c\bigl[(\sigma^2+\beta^2\sigma_\theta^2)\alpha-\beta\sigma_\theta^2\bigr],
\end{equation}
whose left-hand side is strictly positive since $\beta\alpha<1$. Implicit differentiation yields $\partial\alpha/\partial p=-F_p/F_\alpha$ for any primitive $p$. A direct calculation implies $F_\sigma=2c\sigma\alpha>0$, so $\partial\alpha/\partial\sigma<0$. From \eqref{eq:F-identity}, $F_c=(\sigma^2+\beta^2\sigma_\theta^2)\alpha-\beta\sigma_\theta^2>0$, so $\partial\alpha/\partial c<0$. The bias $d$ does not appear in $F$, so $\alpha$ is independent of $d$.

For the monotonicity of \(\beta\alpha\) in \(\beta\), differentiate \(F=0\)
with respect to \(\beta\). Since
\[
\frac{d(\beta\alpha)}{d\beta}
=
\alpha+\beta\frac{\partial\alpha}{\partial\beta}
=
\frac{\alpha F_\alpha-\beta F_\beta}{F_\alpha},
\]
a direct calculation using \(F=0\) returns $\alpha F_\alpha-\beta F_\beta
=
2c\beta\sigma_\theta^2(1-\beta\alpha)$. Since \(F_\alpha>0\), \(c>0\), \(\beta>0\), and \(\beta\alpha<1\), it follows that $\frac{d(\beta\alpha)}{d\beta}>0$.
This proves the claimed monotonicity of \(\beta\alpha\) in \(\beta\), completing part (i).

For (ii), the residual posterior variance $V=c(1-\beta\alpha)\sigma_\theta^2/(c+\alpha^2)$ has $V_c=(1-\beta\alpha)\sigma_\theta^2\alpha^2/(c+\alpha^2)^2>0$ and $V_\sigma=0$ at fixed $\alpha$, while $V_\alpha=-c\sigma_\theta^2[\beta c+\alpha(2-\beta\alpha)]/(c+\alpha^2)^2<0$ since $\beta\alpha<1$. Combining with $\partial\alpha/\partial c<0$ and $\partial\alpha/\partial\sigma<0$ delivers $\mathrm{d}V/\mathrm{d}c>0$ and $\mathrm{d}V/\mathrm{d}\sigma>0$. The bias $d$ does not appear in $V$, so $V$ and $W^R=-V/2$ are independent of $d$.

For (iii), the average inflation is $\E[r-b]=\E[R(\theta,b)]$, which from the report rule equals $\alpha d/c$ (the constant term, since $\E[\theta]=\E[b-\beta\theta]=0$). Both factors $\alpha$ and $1/c$ are decreasing in $c$, so $\partial(\alpha d/c)/\partial c<0$; in $\sigma$, only $\alpha$ moves and decreases, so $\partial(\alpha d/c)/\partial\sigma<0$; the dependence in $d$ is linear with slope $\alpha/c>0$.
\end{proof}

\begin{proof}[\textbf{Proof of Corollary~\ref{cor:PA-comp_s}}]
Substituting the equilibrium identity $V(c+\alpha^2)=c(1-\beta\alpha)\sigma_\theta^2$ into $W^S=-(V+d^2)(1+\alpha^2/c)/2$,
\[
W^S = -\tfrac12 V(1+\alpha^2/c) - \tfrac12 d^2(1+\alpha^2/c) = -\tfrac12(1-\beta\alpha)\sigma_\theta^2 - \tfrac12 d^2 - \tfrac{d^2\alpha^2}{2c}.
\]

For the comparative statics in $c$ and $\sigma$, substituting $F(\alpha)=0$ into $F_\alpha=2\beta\sigma_\theta^2\alpha+c\sigma^2+c\beta^2\sigma_\theta^2-\sigma_\theta^2$ simplifies to $F_\alpha=\beta\sigma_\theta^2(c+\alpha^2)/\alpha$, so implicit differentiation of $F=0$ produces
\[
\frac{d\alpha}{dc}=-\frac{\alpha^2(1-\beta\alpha)}{c\beta(c+\alpha^2)}<0,
\qquad
\frac{d\alpha}{d\sigma}=-\frac{2c\sigma\alpha^2}{\beta\sigma_\theta^2(c+\alpha^2)}<0.
\]
Total differentiation of $W^S$ in $c$ produces
\[
\frac{dW^S}{dc}=\frac{d\alpha}{dc}\left[\frac{\beta\sigma_\theta^2}{2}-\frac{d^2\alpha}{c}\right]+\frac{d^2\alpha^2}{2c^2}.
\]
At the zero of this derivative, substituting the closed form for $d\alpha/dc$ yields, after rearranging,
\begin{equation}\label{eq:dlow}
|\underline d|^2 = \frac{c\beta\sigma_\theta^2(1-\beta\alpha)}{2\alpha+\beta(c-\alpha^2)},
\end{equation}
the unique positive value of $d^2$ at which $dW^S/dc=0$, with $dW^S/dc>0$ iff $d^2>|\underline d|^2$. In $\sigma$, $\alpha^2/c$ depends on $\sigma$ only through $\alpha$, so
\[
\frac{dW^S}{d\sigma}=\frac{d\alpha}{d\sigma}\left[\frac{\beta\sigma_\theta^2}{2}-\frac{d^2\alpha}{c}\right],
\]
whose sign is opposite that of the bracket; hence $dW^S/d\sigma>0$ iff
$
d^2>|\overline d|^2=\frac{c\beta\sigma_\theta^2}{2\alpha}.
$
Finally,
\[
\frac{|\underline d|^2}{|\overline d|^2}=\frac{2\alpha(1-\beta\alpha)}{2\alpha+\beta(c-\alpha^2)},
\]
and the denominator exceeds the numerator by $2\alpha+\beta c-\beta\alpha^2-2\alpha+2\beta\alpha^2=\beta(c+\alpha^2)>0$, so $|\underline d|<|\overline d|$.
\end{proof}

\begin{proof}[\textbf{Proof of Corollary~\ref{cor1r}}]
The sender's first-order condition is
$
r-b=-\frac{\alpha}{c}\bigl(a(r)-\theta-d\bigr).
$
Since $a(r)=\E[\theta| r]$, $\E[a-\theta]=0$ and $\E[(a-\theta-d)^2]=\Var(\theta| r)+d^2$.  Therefore
\[
\E[D(r,b)]=\frac{c}{2}\E[(r-b)^2]
=\frac{\alpha^2}{2c}\bigl(\Var(\theta| r)+d^2\bigr).
\]
Sender welfare is
\[
W^S=-\tfrac{1}{2}\E[(a-\theta-d)^2]-\E[D(r,b)]
=-\tfrac{1}{2}\bigl(\Var(\theta| r)+d^2\bigr)-\frac{\alpha^2}{2c}\bigl(\Var(\theta| r)+d^2\bigr).
\]

\textit{Uninformative anchor derivations}.
Suppose $\beta=0$ in \eqref{eq:alpha-beta}. The regular-equilibrium condition becomes \eqref{eq:uninf-degen}. At  $c\sigma^2=\sigma_\theta^2$, the affine report rule is well defined for every $\alpha>0$:
\[
a(r)=\alpha r-\frac{\alpha^2 d}{c},
\qquad
R(\theta,b)=
\frac{\alpha}{c+\alpha^2}\theta
+
\frac{c}{c+\alpha^2}b
+
\frac{\alpha d}{c}.
\]
Since $(\theta,r)$ is jointly Gaussian, $\Var(\theta| r)=\sigma_\theta^2-\frac{\Cov(\theta,r)^2}{\Var(r)}=\frac{c\sigma_\theta^2}{c+\alpha^2}$,
which directly recovers receiver welfare. The sender-cost calculation in \eqref{eq:gauss-cost} and \eqref{eq:gauss-WS}, evaluated at $\beta=0$ and $c\sigma^2=\sigma_\theta^2$, produces $V(1+\alpha^2/c)=\sigma_\theta^2$ (so the action-mismatch piece reduces to $\sigma_\theta^2/2$) and a bias piece $d^2(c+\alpha^2)/(2c)=d^2/2+d^2\alpha^2/(2c)$, hence $W^S=-\sigma_\theta^2/2-d^2/2-d^2\alpha^2/(2c)$.
\end{proof}

\begin{proof}[\textbf{Proof of Proposition~\ref{prop:hybrid-local-language}}]
Consider labels $k<\ell$ and a numerical report $r$. By Definition~\ref{def:cell-hybrid-main}, label $k$ is used only by states in $I_k$, and label $\ell$ is used only by states in $I_\ell$. Conditional on a label, the continuation is the regular anchored equilibrium for the prior truncated to that cell. Since $M=\R$, $B(\theta)=\R$, and the anchor density is strictly positive, Corollary~\ref{cor:fulluse-unbounded} applies inside every cell: each state in the cell can generate every report by some anchor realization. Hence the posterior after $(k,r)$ has support $I_k$, while the posterior after $(\ell,r)$ has support $I_\ell$.

Because the cells are ordered, the posterior after $(\ell,r)$ strictly first-order stochastically dominates the posterior after $(k,r)$. Assumption~\ref{ass:receiver-main} implies strict concavity in $a$ and increasing differences in $(a,\theta)$, so the receiver's optimal action is strictly increasing under strict first-order stochastic dominance. Therefore $a_k(r)<a_\ell(r)$. The same within-cell full-use argument shows the final claim: for every $\theta\in I_k$ and every $r\in\R$, some anchor $b$ satisfies $R_k(\theta,b)=r$.
\end{proof}

\begin{proof}[\textbf{Proof of Theorem~\ref{thm:hybrid-anchor}}]
Define $L_X^i(d)=W^{i,FB}(d)-W_X^i(d)$ for $X\in\{A,H\}$, with $W^{R,FB}=\int_\Theta U^R(a^R(\theta),\theta)f(\theta)d\theta$ and $W^{S,FB}(d)=\int_\Theta U^S(a^S(\theta;d),\theta;d)f(\theta)d\theta$ the ideal-action benchmarks at zero reporting cost. The proof of Theorem~\ref{thm:hybrid-ranking} establishes $L_H^i(d)\to0$ as $d\downarrow 0$; it remains to bound $L_A^i(d)$ below by a primitive constant uniformly in $d$ small.

\smallskip
\noindent\emph{Reports stay within $K$ of the anchor.}  At any regular anchor-only equilibrium $(R,a)$, sender optimality rules out the deviation to $b$, so $c\phi(R(\theta,b)-b)\le U^S(a(R(\theta,b)),\theta;d)-U^S(a(b),\theta;d)$ for every $(\theta,b)$. Receiver actions stay in $A=[a^R(\underline\theta),a^R(\overline\theta)]$ and $U^S$ is bounded on $A\times\Theta\times[0,\bar d_0]$, so the right side is bounded by $\Delta_S<\infty$. Since Assumption~\ref{ass:PB-main} implies $\phi(u)\ge(\underline\kappa/2)u^2$, $|R(\theta,b)-b|\le K=\sqrt{2\Delta_S/(c\underline\kappa)}$.

\smallskip
\noindent\emph{Information gap.} For any bounded continuous $g:\Theta\to\mathbb R$ with
$g(\overline\theta)>g(\underline\theta)$, let $\Gamma_K(g)$ be the infimum, over
nondecreasing functions $\alpha:\R\to A$, of\[
\E_{\theta,b}\!\left[\max\{g(\theta)-\alpha(b+K),0\}^2+\max\{\alpha(b-K)-g(\theta),0\}^2\right],
\]
where $(\theta,b)$ has density $f(\theta)q_\sigma(b\mid\theta)$. The integrand is the squared distance of \(g(\theta)\) from the interval
\([\alpha(b-K),\alpha(b+K)]\). It is therefore zero exactly when
$
\alpha(b-K)\le g(\theta)\le \alpha(b+K)$. Suppose by contradiction that \(\Gamma_K(g)=0\). Taking a minimizing sequence and using the standard compactness of bounded monotone functions, one may pass to a monotone limit \(\alpha\):
$$\E_{\theta,b}\!\left[\max\{g(\theta)-\alpha(b+K),0\}^2+\max\{\alpha(b-K)-g(\theta),0\}^2\right]=0.$$
By continuity of \(g\), for almost every anchor value \(b\) this condition holds for every state \(\theta\). For any  \(b'\), applying the condition to the lowest and highest states yields
$\alpha(b'-K)\le g(\underline\theta)$
and $g(\overline\theta)\le \alpha(b'+K)$. Consider \(\varepsilon>0\) and  \(b\), so that one can evaluate the first inequality at $b'=b+2K+\varepsilon$ and the second inequality at $b'=b$. Then, $g(\overline\theta)\le \alpha(b+K)\le \alpha(b+K+\varepsilon)\le g(\underline\theta)$,
which contradicts the strict monotonicity of \(g\). Hence \(\Gamma_K(g)>0\).

\noindent\noindent\emph{Receiver loss.}
$a$ is nondecreasing and $|R(\theta,b)-b|\le K$, so $a(R(\theta,b))\in[a(b-K),a(b+K)]$. Uniform strong concavity of $U^R$ in the action implies
\[
U^R(a^R(\theta),\theta)-U^R(a(R(\theta,b)),\theta)
\ge
\frac{\eta}{2}\bigl(a(R(\theta,b))-a^R(\theta)\bigr)^2.
\]
The actual action lies in $[a(b-K),a(b+K)]$, so its distance from
$a^R(\theta)$ is at least the distance from $a^R(\theta)$ to this interval.
Thus
\[
L_A^R(d)\ge
\frac{\eta}{2}\,
\E\!\left[
\max\{a(b-K)-a^R(\theta),0\}^2
+
\max\{a^R(\theta)-a(b+K),0\}^2
\right].
\]
The equilibrium action rule $a$ is a nondecreasing rule in the
definition of $\Gamma_K(a^R)$. Hence $L_A^R(d)\ge\frac{\eta}{2}\Gamma_K(a^R)>0$.

\smallskip
\noindent\emph{Sender loss.}
The first-order approximation $a^S(\theta;d)=a^R(\theta)+d\hat\delta(\theta)+O(d^2)$ yields $\varepsilon_d=\sup_\theta|a^S(\theta;d)-a^R(\theta)|\to 0$. For a closed interval $[l,r]$, the square root of the integrand in the definition of $\Gamma_K$ equals $\max\{l-x,x-r,0\}$, which is $1$-Lipschitz in $x$. Hence $|a^S(\theta;d)-a^R(\theta)|\le\varepsilon_d$ together with the $L^2$ triangle inequality yields $\sqrt{\Gamma_K(a^S(\cdot;d))}\ge\sqrt{\Gamma_K(a^R)}-\varepsilon_d$. Consider $\bar d>0$ small enough that $\varepsilon_d\le\sqrt{\Gamma_K(a^R)}/2$ on $[0,\bar d]$; then $\Gamma_K(a^S(\cdot;d))\ge\Gamma_K(a^R)/4$. Assumption~\ref{ass:primitive-smallbias} provides uniform strong concavity of $U^S$ on $A\times\Theta\times[0,\bar d_0]$ with constant $\eta_S$, so
\[
L_A^S(d)\ge\frac{\eta_S}{2}\Gamma_K(a^S(\cdot;d))\ge\frac{\eta_S}{8}\Gamma_K(a^R)>0,
\]
uniformly in $d\in[0,\bar d]$. The conclusion follows readily.
\end{proof}

\begin{proof}[\textbf{Proof of Theorem~\ref{thm:hybrid-ranking}}]
Let $\kappa_i=\ell_i''(0)$ and $s(\theta)=g'(\theta)/g(\theta)$. Log-concavity of $g$ implies $s'\le 0$. Since $\ell_i$ is even and $C^5$, $\ell_i(e)=\kappa_ie^2/2+O(e^4)$ and $\ell_i'(e)=\kappa_ie+O(e^3)$ uniformly for small $e$.

\smallskip
\noindent\emph{Cheap talk only.}
For a cell of width $h$ centered at $m$, define $a_C(m,h)=m+\alpha$, as the receiver's optimal action, and $u=\theta-m$. The receiver's first-order condition is
\[
0=\int_{-h/2}^{h/2}\ell_R'(\alpha-u)g(m+u)\,du.
\]
Using $g(m+u)=g(m)\{1+s(m)u+O(u^2)\}$ and $\ell_R'(e)=\kappa_R e+O(e^3)$, this reduces to $\kappa_Rg(m)\{\alpha h-s(m)h^3/12\}+O(h^5)=0$, hence
\begin{equation}\label{eq:cell-action}
a_C(m,h)=m+\frac{s(m)}{12}h^2+O(h^4).
\end{equation}
At an interior boundary $\theta$ with left width $h$ and right width $q$, the boundary type's sender residuals, i.e., the difference between the receiver action and the sender preferred action, are
\begin{align}
e_+^0(q)&=\tfrac{q}{2}-d\hat\delta(\theta)+\tfrac{s(\theta)}{12}q^2+O(q^3+d^2),\label{eq:eplus}\\
e_-^0(h)&=-\tfrac{h}{2}-d\hat\delta(\theta)+\tfrac{s(\theta)}{12}h^2+O(h^3+d^2) ,\label{eq:eminus}
\end{align}
which follows from evaluating the receiver's action $a_C$ in (\ref{eq:cell-action}) for the left cell at $m=\theta-h/2$ and for the right cell at $m=\theta+q/2$.  

 Sender indifference $\ell_S(e_+^0(q))=\ell_S(e_-^0(h))$ implies
$e_+^0(q)+e_-^0(h)=0$. Using
\eqref{eq:eplus}--\eqref{eq:eminus} yields the boundary recursion
\begin{equation}
\label{eq:ct-boundary-recursion}
q-h
=
4d\hat\delta(\theta)
-\frac{s(\theta)}{6}(q^2+h^2)
+O(q^3+h^3+d^2).\end{equation}
Let
$
\underline\theta=\theta_0<\theta_1<\cdots<\theta_N=\overline\theta
$
be the most informative cheap-talk partition, and set $h_j=\theta_j-\theta_{j-1}$. It is well-known that the mesh must become fine as $d$ become small, so consider a boundary $\theta_j$ with $j\geq 1$. The left width is $h_j$ and the right width is
$h_{j+1}$, so the recursion becomes:\[
h_{j+1}-h_j
=
4d\hat\delta(\theta_j)
-\frac{s(\theta_j)}{6}(h_{j+1}^2+h_j^2)
+O(h_{j+1}^3+h_j^3+d^2).
\]
Since $\hat\delta$ and $s$ are bounded on $\Theta$, the preceding recursion implies that, for small widths,
$|h_{j+1}-h_j|
\le C(d+h_j^2+h_{j+1}^2)$. A straightforward manipulation of this inequality, with $y_j=\frac{h_j^2}{d}$, yields
$$|y_{j+1}-y_j|
\le
\underbrace{C(h_{j+1}+h_j)}_{\rho_j}(1+y_j+y_{j+1}).$$

Hence $
(1-\rho_j)y_{j+1}\le (1+\rho_j)y_j+\rho_j.
$
For small enough $d$,
$\rho_j\le 1/2$, so dividing by $1-\rho_j$ and absorbing the constant yields $y_{j+1} \le
(1+C(h_{j+1}+h_j))y_j+C(h_{j+1}+h_j)$. Using
$\sum_j h_j=\overline\theta-\underline\theta$ yields that $\sup_j y_j\le C$ and thus 
$h_j=O(\sqrt d)$. Repeating this argument at the first boundary with zero width for the lower cell further yields $h_1=O(d)$.

This implies $h_{j+1}-h_j=O(d)$. Hence $h_{j+1}^2+h_j^2=2h_j^2+O(h_j^3)$,
so that the recursion is
\begin{equation}\label{eq:cs-recursion}
h_C(\theta_{j+1})-h_C(\theta_j)
=
4d\hat\delta(\theta_j)-\frac{s(\theta_j)}{3}h_j^2+O(h_j^3),
\end{equation}
where $h_C(\theta_j)=h_j$ is the left cell width expressed in terms of $\theta_j$.

Using \eqref{eq:cs-recursion} and $h_{j+1}=h_j+O(d)=\theta_{j+1}-\theta_j$,

\begin{equation}\label{eq:ct-square-recursion}
\frac{h_{j+1}^2-h_j^2}{\theta_{j+1}-\theta_j}
=
8d\hat\delta(\theta_j)-\frac{2s(\theta_j)}{3}h_j^2+o(d).
\end{equation}

Given that $h_j^2=O(d)$, one can write $
h_C(\theta_j)^2=dY(\theta_j)+o(d)$, which substituted above and passing to the continuum limit implies $
Y'(\theta)=8\hat\delta(\theta)-\frac{2s(\theta)}{3}Y(\theta).
$\footnote{Here and for \eqref{eq:Dd-comparison} below, the passage from the width recursions to the limit ODEs is the standard Euler-scheme approximation: the squared-width recursion advances by steps $\theta_{j+1}-\theta_j=O(\sqrt d)$ with local truncation error $o(d)$ per unit length (resp.\ $o(d^2)$ for the difference recursion), so the piecewise-linear interpolant of the discrete profile converges uniformly to the solution of the limit ODE. Existence of the most informative partition at small $d$ and the bound $h_1=O(d)$ follow from the monotone forward-shooting structure of the boundary recursion, as in \citet{CrawfordSobel1982}: solutions are ordered by the first width $h_1$, the number of cells is nonincreasing in $h_1$, and successive maximal solutions differ in $h_1$ by $O(d)$.} The initial condition is $Y(\underline\theta)=0$,
because the left width at the left endpoint is zero. Since $\hat\delta>0$, the solution is
\begin{equation}\label{eq:Y-ode}
Y(\theta)
=
8\int_{\underline\theta}^{\theta}
\hat\delta(t)\exp\!\left(-\frac{2}{3}\int_t^\theta s(v)\,dv\right)dt.
\end{equation}
Aggregating the within-cell loss $\kappa_i h^2/24+O(h^4)$ implies\begin{equation}\label{eq:LC-leading}
L_C^i(d)=\frac{d}{24}\int_{\underline\theta}^{\overline\theta}g(\theta)\,\kappa_i\,Y(\theta)\,d\theta+o(d).
\end{equation}

\smallskip
\noindent\emph{Within-cell anchor information.}
Let $\lambda=\beta/\sigma$ and $s_\omega=\omega'/\omega$. Symmetry of $\omega$ kills the odd-moment terms of every likelihood expansion below, and the Fisher informativeness of the linear anchor is
\begin{equation}\label{eq:anchor-info}
\mathcal I=\lambda^2\int s_\omega(u)^2\omega(u)\,du>0.
\end{equation}
For the right cell $[\theta,\theta+q]$, let $u\in[0,q]$ be the distance from the boundary. A Taylor expansion at $u=0$ yields
\begin{equation}\label{eq:hr-likelihood}
\frac{\omega(z-\lambda u)}{\omega(z)}
=
1-\lambda s_\omega(z)u
+\frac{\lambda^2}{2}\bigl(s_\omega'(z)+s_\omega(z)^2\bigr)u^2
+O_z(u^3).
\end{equation}

Before using the numerical anchor, define $\E_+$ as expectation over $u\in[0,q]$ with density proportional to $g(\theta+u)$. Using the same first-order Taylor expansion of $g$ as above,
$
g(\theta+u)=g(\theta)\{1+s(\theta)u+O(u^2)\}$,
so, for $k\ge1$,
\[
\E_+[u^k]
=
\frac{\int_0^q u^k\{1+s(\theta)u+O(u^2)\}\,du}
{\int_0^q \{1+s(\theta)u+O(u^2)\}\,du}.
\]
Applying this formula for $k=1,2,3$ implies\begin{align}
\mu_+(q)&=\E_+[u]
=\tfrac{q}{2}+\tfrac{s(\theta)}{12}q^2+O(q^3),\notag\\
V_+(q)&=\E_+[u^2]-\mu_+(q)^2
=\tfrac{q^2}{12}+O(q^4),\notag\\
M_+(q)&=\E_+[u^3]-\mu_+(q)\E_+[u^2]
=\tfrac{q^3}{12}+O(q^4).
\label{eq:hr-cell-moments}
\end{align}

Define the anchor-induced shift, after observing $z$, in the posterior mean of $u$ as
\[
\Delta_+(z,q)
=
\frac{\E_+\!\left[u\,\omega(z-\lambda u)/\omega(z)\right]}
{\E_+\!\left[\omega(z-\lambda u)/\omega(z)\right]}
-\mu_+(q).
\]
Using \eqref{eq:hr-likelihood} in this definition,
\[
\begin{aligned}
\Delta_+(z,q)
&=
\frac{
\E_+\!\left[
u\left\{
1-\lambda s_\omega(z)u
+\frac{\lambda^2}{2}\bigl(s_\omega'(z)+s_\omega(z)^2\bigr)u^2
+O_z(u^3)
\right\}
\right]
}{
\E_+\!\left[
1-\lambda s_\omega(z)u
+\frac{\lambda^2}{2}\bigl(s_\omega'(z)+s_\omega(z)^2\bigr)u^2
+O_z(u^3)
\right]
}
-\mu_+(q).
\end{aligned}
\]
Multiplying the numerator by the reciprocal of the denominator and simplifying:\[
\begin{aligned}
\Delta_+(z,q)
&=
-\lambda s_\omega(z)\E_+[u^2]
+\lambda s_\omega(z)\mu_+(q)^2\\
&\quad
+\frac{\lambda^2}{2}\bigl(s_\omega'(z)+s_\omega(z)^2\bigr)\E_+[u^3]
-\frac{\lambda^2}{2}\bigl(s_\omega'(z)+s_\omega(z)^2\bigr)\mu_+(q)\E_+[u^2]\\
&\quad
+\lambda^2s_\omega(z)^2\mu_+(q)^3
-\lambda^2s_\omega(z)^2\mu_+(q)\E_+[u^2]
+O_z(q^4).
\end{aligned}
\]
Using the definitions of $V_+(q)$ and $M_+(q)$ in \eqref{eq:hr-cell-moments},
\begin{equation}\label{eq:hr-shift-raw}
\Delta_+(z,q)
=
-\lambda s_\omega(z)V_+(q)
+\frac{\lambda^2}{2}\bigl(s_\omega'(z)+s_\omega(z)^2\bigr)M_+(q)
-\lambda^2s_\omega(z)^2\mu_+(q)V_+(q)
+O_z(q^4).
\end{equation}
The score identities, together with \eqref{eq:anchor-info}, are
$\E_z[s_\omega(z)]=0$,
$\E_z[s_\omega'(z)+s_\omega(z)^2]=0$, and
$\lambda^2\E_z[s_\omega(z)^2]=\mathcal I$. Taking expectations over $z$ in \eqref{eq:hr-shift-raw} and using \eqref{eq:hr-cell-moments},
\begin{equation}\label{eq:shift-moments-plus}
\E_z[\Delta_+(z,q)]
=
-\frac{\mathcal I}{24}q^3-\frac{s(\theta)\mathcal I}{144}q^4+o(q^4),
\qquad
\E_z[\Delta_+(z,q)^2]
=
\frac{\mathcal I}{144}q^4+o(q^5).
\end{equation}

For the left cell $[\theta-h,\theta]$, one can write the state as $\theta+u$ with $u\in[-h,0]$. Define $\E_-$ analogously, using the density proportional to $g(\theta+u)$ on $[-h,0]$, and let $\mu_-(h)=\E_-[u]$. After observing $z$, define
\[
\Delta_-(z,h)
=
\frac{\E_-\!\left[u\,\omega(z-\lambda u)/\omega(z)\right]}
{\E_-\!\left[\omega(z-\lambda u)/\omega(z)\right]}
-\mu_-(h).
\]
The same calculation as for the right cell gives
\begin{equation}\label{eq:shift-moments-minus}
\E_z[\Delta_-(z,h)]
=
\frac{\mathcal I}{24}h^3-\frac{s(\theta)\mathcal I}{144}h^4+o(h^4),
\qquad
\E_z[\Delta_-(z,h)^2]
=
\frac{\mathcal I}{144}h^4+o(h^5).
\end{equation}

\smallskip
\noindent\noindent\emph{Hybrid cell scale.}
At a hybrid boundary, sender indifference includes the numerical report. Relative to the label-only condition that gives \eqref{eq:ct-boundary-recursion}, the anchor changes the receiver action by $\Delta_+(z,q)$ in the right cell and by $\Delta_-(z,h)$ in the left cell.

By \eqref{eq:eplus}--\eqref{eq:eminus}, the boundary type's sender residual is $O(q+d)$ in the right cell and $O(h+d)$ in the left cell. Since $\ell_S$ is locally quadratic, a receiver-action shift $\Delta$ changes sender loss by $O(e\Delta+\Delta^2)$. Hence \eqref{eq:shift-moments-plus}--\eqref{eq:shift-moments-minus} imply that the numerical report changes the boundary type's expected sender loss by
$
O((q+d)q^3+q^4)$
on the right, and 
$O((h+d)h^3+h^4)$
on the left.

To convert these payoff perturbations into a perturbation of the width recursion, use the local form of the label-only payoff gap. Since $\ell_S(e)=\kappa_Se^2/2+O(e^4)$,
\[
\ell_S(e_+^0(q))-\ell_S(e_-^0(h))
=
\frac{\kappa_S}{2}
\bigl(e_+^0(q)-e_-^0(h)\bigr)
\bigl(e_+^0(q)+e_-^0(h)\bigr)
+O(q^4+h^4+d^4).
\]
By \eqref{eq:eplus}--\eqref{eq:eminus},
\[
e_+^0(q)-e_-^0(h)
=
\frac{q+h}{2}+O(q^2+h^2+d^2).
\]
Thus a payoff perturbation of size
$
O((q+d)q^3+q^4)+O((h+d)h^3+h^4)
$
changes the condition for $e_+^0(q)+e_-^0(h)$, and hence the recursion for $q-h$, by at most
\[
C\frac{(q+d)q^3+(h+d)h^3+q^4+h^4}{h+q}.
\]
For small widths this is bounded by $C(d+h^2+q^2)$. Hence the hybrid widths satisfy the same inequality as the cheap-talk widths:
$
|q-h|\le C(d+h^2+q^2).
$
Using the same $y_j=h_j^2/d$ argument as in the cheap-talk part, together with the same endpoint condition for the most informative branch, implies $h_{H,j}=O(\sqrt d)$ uniformly over interior hybrid cells. The endpoint shooting argument also gives $h_{H,1}=O(d)$.

The scale argument above establishes that interior hybrid widths are $O(\sqrt d)$. At this scale, the leading right and left anchor payoff terms cancel in the boundary indifference condition. The remaining payoff imbalance is order $d^{5/2}$, and dividing by the $O(h+q)=O(\sqrt d)$ slope above yields an order-$d^2$ correction to the recursion for $q-h$.

Before computing that correction, note that reporting costs do not enter at the $d^2$ order. Let $a(r)$ be the receiver's action in a hybrid cell of width $h_{H,j}$ after numerical report $r$. The sender's first-order condition for $r$ is
\[
c\phi'(r-b)
=
-\ell_S'\bigl(a(r)-\theta-d\hat\delta(\theta)\bigr)a_r(r).
\]
The posterior after any report is still over states in the same cell. Hence the receiver action is within $O(h_{H,j})$ of any type in that cell, and since $\hat\delta$ is bounded,
\[
a(r)-\theta-d\hat\delta(\theta)=O(h_{H,j}+d).
\]
Moreover, \eqref{eq:hr-shift-raw} and \eqref{eq:hr-cell-moments} show that the report-dependent part of the posterior mean starts at order $h_{H,j}^2$. Differentiating the same expansion with respect to the standardized residual $z$, using the bounded derivatives of $s_\omega$, gives the same order for the sensitivity. Since $z$ is affine in $r$,
$
a_r(r)=O(h_{H,j}^2).
$
Therefore
$
r-b=O((h_{H,j}+d)h_{H,j}^2).
$
Using $h_{H,j}=O(\sqrt d)$, this returns $r-b=O(d^{3/2})$. Since $\phi'(0)=0$ and $\phi''(0)>0$, the reporting cost per type is
$
O((r-b)^2)=O(d^3)=o(d^2).
$
The induced extra change in the receiver action is
\[
a_r(r)(r-b)=O((h_{H,j}+d)h_{H,j}^4)=O(d^{5/2}),
\]
which changes expected quadratic losses only at order $d^3$. Thus reporting costs and the associated strategic report displacement do not affect the $d^2$ comparison.

\smallskip
\noindent\emph{Cutoff shift induced by the anchor.}
At an interior boundary $\theta$, let $G_+(\theta,q,d)$ and $G_-(\theta,h,d)$ be the boundary type's expected sender-utility gain from the numerical report in the right and left cells. Using \eqref{eq:eplus}--\eqref{eq:eminus}, \eqref{eq:shift-moments-plus}--\eqref{eq:shift-moments-minus}, and $\ell_S(e)=\kappa_Se^2/2+O(e^4)$,
\begin{align}
G_+(\theta,q,d)
&=
-\kappa_Se_+^0(q)\,\E_z[\Delta_+(z,q)]
-\frac{\kappa_S}{2}\E_z[\Delta_+(z,q)^2]
+o(q^5)\notag\\
&=
-\kappa_S
\left(\frac q2-d\hat\delta(\theta)+\frac{s(\theta)}{12}q^2\right)
\left(-\frac{\mathcal I}{24}q^3-\frac{s(\theta)\mathcal I}{144}q^4\right)
-\frac{\kappa_S\mathcal I}{288}q^4
+o(q^5)\notag\\
&=
\kappa_S\mathcal I
\left[
\frac{5}{288}q^4
-\frac{d\hat\delta(\theta)}{24}q^3
+\frac{s(\theta)}{144}q^5
\right]
+o(q^5),\notag\\
G_-(\theta,h,d)
&=
-\kappa_Se_-^0(h)\,\E_z[\Delta_-(z,h)]
-\frac{\kappa_S}{2}\E_z[\Delta_-(z,h)^2]
+o(h^5)\notag\\
&=
-\kappa_S
\left(-\frac h2-d\hat\delta(\theta)+\frac{s(\theta)}{12}h^2\right)
\left(\frac{\mathcal I}{24}h^3-\frac{s(\theta)\mathcal I}{144}h^4\right)
-\frac{\kappa_S\mathcal I}{288}h^4
+o(h^5)\notag\\
&=
\kappa_S\mathcal I
\left[
\frac{5}{288}h^4
+\frac{d\hat\delta(\theta)}{24}h^3
-\frac{s(\theta)}{144}h^5
\right]
+o(h^5).
\label{eq:Gpm}
\end{align}

\noindent

From \eqref{eq:Gpm},
\begin{align}
G_+(\theta,q,d)-G_-(\theta,h,d)
&=
\kappa_S\mathcal I
\left[
\frac{5}{288}(q^4-h^4)
-\frac{d\hat\delta(\theta)}{24}(q^3+h^3)
+\frac{s(\theta)}{144}(q^5+h^5)
\right]
+o(h^5).
\label{eq:G-difference-raw}
\end{align}

We now evaluate \eqref{eq:G-difference-raw} at the cheap-talk baseline boundary. In this paragraph, $h$ and $q$ denote the adjacent cheap-talk widths at that boundary; the hybrid correction is computed from this baseline. The squared-width calculation following \eqref{eq:cs-recursion} gives
\[
h^2=dY(\theta)+o(d),
\]
where $Y$ solves \eqref{eq:Y-ode}. For an interior boundary, \eqref{eq:Y-ode} implies $Y(\theta)>0$. Hence $h$ is of order $\sqrt d$, so $d=o(h)$.

Both $h$ and $q$ are of order $\sqrt d$. First use \eqref{eq:cs-recursion} only to get the crude order:
\begin{align}
q-h
&=
4d\hat\delta(\theta)-\frac{s(\theta)}{6}(q^2+h^2)
+O(q^3+h^3+d^2)=
O(d),
\end{align}
which also implies $ q=h+o(h)$. Returning to \eqref{eq:cs-recursion},
$q^2+h^2=2h^2+O(h(q-h))=2h^2+O(h^3)$,  so that
\begin{align}
q-h
&=
4d\hat\delta(\theta)-\frac{s(\theta)}{3}h^2+O(h^3).
\label{eq:q-close-to-h2}
\end{align}
Therefore
\begin{align}
q^3+h^3
=
2h^3+o(h^3), \quad
q^5+h^5
=
2h^5+o(h^5).
\label{eq:sum-power-substitutions}
\end{align}
For the difference $q^4-h^4$, one needs to use the sharper expression for $q-h$ in \eqref{eq:q-close-to-h2}:
\begin{align}
q^4-h^4
&=
(q-h)(q^3+q^2h+qh^2+h^3)\notag\\
&=
\bigl(4d\hat\delta(\theta)-\frac{s(\theta)}{3}h^2+O(h^3)\bigr)
\bigl(4h^3+o(h^3)\bigr)\notag\\
&=
16d\hat\delta(\theta)h^3-\frac{4s(\theta)}{3}h^5+o(h^5).
\label{eq:fourth-power-substitution}
\end{align}
Substituting \eqref{eq:sum-power-substitutions}--\eqref{eq:fourth-power-substitution} into \eqref{eq:G-difference-raw},\begin{equation}\label{eq:gap-difference2}
G_+(\theta,q,d)-G_-(\theta,h,d)
=
\kappa_S\mathcal I
\left[
\frac{7}{36}d\hat\delta(\theta)h^3
-\frac{s(\theta)}{108}h^5
\right]
+o(h^5).
\end{equation}

The gain difference in \eqref{eq:gap-difference2} is measured in utility units. To convert it into a width correction, consider the local response of the label-only sender loss gap to the right width $q$, holding the left width $h$ fixed. From \eqref{eq:eplus}, $\frac{\partial e_+^0(q)}{\partial q}=\frac12+o(1)$.
Also, since $q=h+o(h)$ and $d=o(h)$, \eqref{eq:eplus} implies $
e_+^0(q)=\frac h2+o(h).
$
Therefore, using $\ell_S'(e)=\kappa_Se+O(e^3)$,
\[
\begin{aligned}
\frac{\partial}{\partial q}
\left\{
\ell_S(e_+^0(q))-\ell_S(e_-^0(h))
\right\}
&=
\ell_S'(e_+^0(q))\frac{\partial e_+^0(q)}{\partial q}\\
&=
\left(\frac{\kappa_Sh}{2}+o(h)\right)
\left(\frac12+o(1)\right)=
\frac{\kappa_Sh}{4}+o(h).
\end{aligned}\]

Let $\Delta_A(\theta)$ denote the additional right-cell width needed to restore boundary indifference after introducing the anchor. The hybrid right-minus-left loss gap is the label-only loss gap minus the anchor gain difference. Since the label-only loss gap is zero at the boundary, linearizing around that boundary yields
\[
\left(\frac{\kappa_Sh}{4}+o(h)\right)\Delta_A(\theta)
-
\bigl\{G_+(\theta,q,d)-G_-(\theta,h,d)\bigr\}
=
0,
\]
Using \eqref{eq:gap-difference2},
\begin{align}
\Delta_A(\theta)
&=
\frac{G_+(\theta,q,d)-G_-(\theta,h,d)}{\kappa_Sh/4+o(h)}\notag\\
&=
\mathcal I
\left[
\frac{7}{9}d\hat\delta(\theta)h^2
-\frac{s(\theta)}{27}h^4
\right]
+o(dh^2+h^4).
\label{eq:width-correction-anchor}
\end{align}

Hence the hybrid width recursion is the cheap-talk width recursion \eqref{eq:cs-recursion} plus this additional term:
\[
h_H(\theta_{j+1})-h_H(\theta_j)
=
4d\hat\delta(\theta_j)
-\frac{s(\theta_j)}{3}h_H(\theta_j)^2
+\Delta_A(\theta_j)
+O(h_H(\theta_j)^3).
\]
Multiplying this recursion by $h_H(\theta_{j+1})+h_H(\theta_j)$
 returns the hybrid analogue of \eqref{eq:ct-square-recursion}:
\begin{equation}\label{eq:hybrid-square-recursion}
\frac{h_H(\theta_{j+1})^2-h_H(\theta_j)^2}{\theta_{j+1}-\theta_j}
=
8d\hat\delta(\theta_j)
-\frac{2s(\theta_j)}{3}h_H(\theta_j)^2
+2\Delta_A(\theta_j)
+o(d^2).
\end{equation}
The two equations  \eqref{eq:ct-square-recursion} and \eqref{eq:hybrid-square-recursion} have common terms
$8d\hat\delta(\theta_j)$ and $-\frac{2s(\theta_j)}{3}h(\theta_j)^2$; the anchor adds the term $2\Delta_A(\theta_j)$.
Interpolating the cheap-talk and hybrid squared-width profiles, define
$
D_d(\theta)=h_H(\theta)^2-h_C(\theta)^2.
$
The squared-width comparison can be written as
\begin{equation}\label{eq:Dd-comparison}
D_d'(\theta)
=
-\frac{2s(\theta)}{3}D_d(\theta)
+
2\Delta_A(\theta)
+
o(d^2).
\end{equation}
Since $h_H(\theta)=O(\sqrt d)$, \eqref{eq:width-correction-anchor} implies $\Delta_A(\theta)=O(d^2)$. Hence \eqref{eq:Dd-comparison} can be stated as
\[
D_d'(\theta)=-\frac{2s(\theta)}{3}D_d(\theta)+O(d^2).
\]
Since the squared-width difference starts from zero at the left endpoint, $D_d(\underline\theta)=0$, solving this equation returns
\[
D_d(\theta)
=
e^{-\frac{2}{3}\int_{\underline\theta}^{\theta}s(v)\,dv}
\int_{\underline\theta}^{\theta}
e^{\frac{2}{3}\int_{\underline\theta}^{t}s(v)\,dv}
O(d^2)\,dt.
\]
Because $s$ is bounded and the state interval is compact, $D_d(\theta)=O(d^2)$. Its leading coefficient can be written $D_d(\theta)=d^2R(\theta)+o(d^2)$.
Since $h_C(\theta)^2=dY(\theta)+o(d)$ and $D_d(\theta)=O(d^2)$, it must hold that $h_H(\theta)^2=dY(\theta)+o(d)$.
Using this in \eqref{eq:width-correction-anchor} yields
\begin{equation}\label{eq:DeltaA-leading}
\Delta_A(\theta)
=
d^2\mathcal I
\left[
\frac{7}{9}\hat\delta(\theta)Y(\theta)
-\frac{s(\theta)}{27}Y(\theta)^2
\right]
+o(d^2).
\end{equation}
Substituting $D_d(\theta)=d^2R(\theta)+o(d^2)$ and \eqref{eq:DeltaA-leading} into \eqref{eq:Dd-comparison}:
\begin{equation}\label{eq:R-ode}
R'(\theta)
=
-\frac{2s(\theta)}{3}R(\theta)
+
\mathcal I
\left[
\frac{14}{9}\hat\delta(\theta)Y(\theta)
-\frac{2s(\theta)}{27}Y(\theta)^2
\right],
\qquad
R(\underline\theta)=0.
\end{equation}

The widened partition raises player $i$'s loss by
\begin{equation}\label{eq:cutoff-cost}
G_{\mathrm{cut}}^i(d)
=
\frac{d^2}{24}
\int_{\underline\theta}^{\overline\theta}
g(\theta)\,\kappa_i\,R(\theta)\,d\theta
+o(d^2).
\end{equation}

\smallskip
\emph{Informational gain inside cells.}
Consider first a right cell of width $q$. Before using the anchor, the receiver's posterior mean of $u$ is $\mu_+(q)$. Since $\mu_+(q)+\Delta_+(z,q)$ is the posterior mean of $u$ after observing $z$,
$
\E[u\mid z]=\mu_+(q)+\Delta_+(z,q).
$ and
\[\begin{aligned}
\E[(\mu_+(q)-u)\Delta_+(z,q)]
&=
\E_z\!\left[\Delta_+(z,q)\{\mu_+(q)-\E[u\mid z]\}\right]\\
&=
-\E_z[\Delta_+(z,q)^2].
\end{aligned}
\]

Thus, using the quadratic expansion of receiver loss, the informational gain in a right cell is
\[
\frac{\kappa_R}{2}\E_z[\Delta_+(z,q)^2]+o(q^4)
=
\frac{\kappa_R\mathcal I}{288}q^4+o(q^4),
\]
where the last equality uses \eqref{eq:shift-moments-plus}. The left-cell calculation is identical, with $h$ and \eqref{eq:shift-moments-minus}.
Multiplying by cell mass $g(m)w$ and using
$h_H(\theta)^2=dY(\theta)+O(d^2)$ yields, for $i=R$,
\begin{equation}\label{eq:anchor-gain2}
G_{\mathrm{info}}^i(d)
=
\frac{d^2}{288}
\int_{\underline\theta}^{\overline\theta}
g(\theta)\,\kappa_i\,\mathcal I\,Y(\theta)^2\,d\theta
+o(d^2).
\end{equation}

For the sender, the additional cross-term $\kappa_Sd\hat\delta(\theta)\E[\Delta_\pm]$ has opposite signs on the two sides of each interior boundary. The leading part cancels in the aggregate; the remaining variation of $\hat\delta$ inside a cell is $O(dh_H(\theta)^4)$ per type and therefore $O(d^3)$ after aggregation. Thus Equation \eqref{eq:anchor-gain2} also applies to $i=S$.

\smallskip
\noindent\emph{Boundary cells.}
The first cell satisfies $h_1=O(d)$ from the endpoint argument used above, so its mass is $O(d)$, its per-type loss is $O(d^2)$, and its contribution is $O(d^3)$ to $L_C^i$ and $o(d^2)$ to $L_H^i-L_C^i$. The last cell has width $O(\sqrt d)$ and therefore contributes $o(d)$ to $L_C^i$ and $o(d^2)$ to $L_H^i-L_C^i$, implying that it doesn't affect the leading terms.

\smallskip
\noindent\emph{Comparison.}
Combining \eqref{eq:cutoff-cost} and \eqref{eq:anchor-gain2},
\begin{equation}\label{eq:loss-difference}
L_H^i(d)-L_C^i(d)
=
\frac{d^2}{24}
\int_{\underline\theta}^{\overline\theta}
g(\theta)\,\kappa_i\,D(\theta)\,d\theta
+o(d^2),
\end{equation}
where
\[
D(\theta)=R(\theta)-\frac{\mathcal I}{12}Y(\theta)^2,
\qquad
D(\underline\theta)=0.
\]
Using \eqref{eq:Y-ode} and \eqref{eq:R-ode}, direct differentiation yields\[
D'(\theta)
=
-\frac{2s(\theta)}{3}D(\theta)
+
\frac{\mathcal I}{18}
Y(\theta)\bigl(4\hat\delta(\theta)-\frac{s(\theta)Y(\theta)}{3}\bigr).
\]
It remains to show that the bracketed term is positive. By \eqref{eq:Y-ode},
$
Y'(\theta)=2\bigl(4\hat\delta(\theta)-\frac{s(\theta)Y(\theta)}{3}\bigr).
$
Hence
\[
\begin{aligned}
\frac{d}{d\theta}\bigl(4\hat\delta(\theta)-\frac{s(\theta)Y(\theta)}{3}\bigr)
&=
4\hat\delta'(\theta)-\frac{s'(\theta)Y(\theta)}{3}-\frac{s(\theta)Y'(\theta)}{3}\\
&=
4\hat\delta'(\theta)-\frac{s'(\theta)Y(\theta)}{3}
-\frac{2s(\theta)}{3}\bigl(4\hat\delta(\theta)-\frac{s(\theta)Y(\theta)}{3}\bigr).
\end{aligned}
\]
Since $Y(\underline\theta)=0$,
\[
4\hat\delta(\underline\theta)-\frac{s(\underline\theta)Y(\underline\theta)}{3}
=
4\hat\delta(\underline\theta)>0.
\]
Solving the above equation:
\[
4\hat\delta(\theta)-\frac{s(\theta)Y(\theta)}{3}
=
e^{-\frac{2}{3}\int_{\underline\theta}^{\theta}s(v)\,dv}
\left[
4\hat\delta(\underline\theta)
+
\int_{\underline\theta}^{\theta}
e^{\frac{2}{3}\int_{\underline\theta}^{t}s(v)\,dv}
\bigl(4\hat\delta'(t)-\frac{s'(t)Y(t)}{3}\bigr)\,dt
\right].
\]
Since $\hat\delta'(\theta)\ge 0$, $s'(\theta)\le 0$, and $Y(\theta)>0$ for interior $\theta$, the bracketed term is positive in the interior. Therefore solving the equation for $D$ yields\[
D(\theta)
=
\frac{\mathcal I}{18}
\int_{\underline\theta}^{\theta}
Y(t)\bigl(4\hat\delta(t)-\frac{s(t)Y(t)}{3}\bigr)
\exp\!\left(-\frac{2}{3}\int_t^\theta s(v)\,dv\right)dt,
\]
so $D(\theta)>0$ for every interior $\theta$.
Since $g(\theta)>0$ and $\kappa_i>0$, the integral in \eqref{eq:loss-difference} is strictly positive. Hence $L_H^i(d)>L_C^i(d)$ for small $d$, or equivalently $W_C^i(d)>W_H^i(d)$ for $i\in\{R,S\}$.
\end{proof}

\begin{small}

\end{small}
\clearpage
\section{Supplementary Appendix}\label{app:supp-existence}

\subsection{Existence with compact messages}

This subsection records the compact-message boundary equations and sufficient conditions for existence.  Let $M=[\underline m,\overline m]$ and let $H$ be the CDF of the anchor noise.  For $u,v\in\R$ and $p\ge0$, define
\[
b_-(u,p,\theta)=\underline m-\psi\!\left(\frac{U^S_1(u,\theta)p}{c}\right),
\qquad
b_+(v,p,\theta)=\overline m-\psi\!\left(\frac{U^S_1(v,\theta)p}{c}\right),
\]
and
\begin{align*}
F_-(u,p)&=\int_\Theta U^R_1(u,\theta)f(\theta)H\!\left(\frac{b_-(u,p,\theta)-b_0(\theta)}{\sigma}\right)d\theta,\\
F_+(v,p)&=\int_\Theta U^R_1(v,\theta)f(\theta)\left[1-H\!\left(\frac{b_+(v,p,\theta)-b_0(\theta)}{\sigma}\right)\right]d\theta.
\end{align*}

\begin{assumption}\label{ass:globalODE}
There exist a compact action interval $A^\dagger=[a_{\min},a_{\max}]\subset\Int(A)$ and constants $p_{\max},K,\underline{\mathcal A}>0$ such that, on $\mathcal S=M\times A^\dagger\times[0,p_{\max}]$, the boundary equations $F_-(u,p)=0$ and $F_+(v,p)=0$ have unique continuous solutions $u^\star(p)$ and $v^\star(p)$ in $A^\dagger$, $\mathcal A\ge\underline{\mathcal A}$, $|G(r,a,p)|\le K(1+p^2)$, and sender strict concavity holds after substituting this bound for $a''$. In addition, the message interval and the endpoint separation are such that the boundary-value operator maps the set $\mathcal D=\{a\in C^1(M):a(M)\subseteq A^\dagger,\ \underline p^\star\le a'\le p^\star\}$ into itself, where $p^\star$ is the positive solution of $x=\Delta a^\dagger/\Delta m+(\Delta m/2)K(1+x^2)$ and $\underline p^\star>0$ is the corresponding lower slope bound. Write $\delta_E$ for the endpoint-separation parameter $\inf\{v^\star(p_2)-u^\star(p_1):p_1,p_2\in[\underline p^\star,p^\star]\}$, which satisfies $\underline p^\star\Delta m+(\Delta m)^2K(1+(p^\star)^2)/2\le\delta_E\le\Delta a^\dagger$.
\end{assumption}

\begin{lemma}\label{lem:verification-compact}
Under Assumptions~\ref{ass:PB-main} and~\ref{ass:globalODE}, let $a\in C^3(M)$ satisfy $a(M)\subseteq A^\dagger$, $0<a'(r)\le p^\star$ on $M$, the interior ODE~\eqref{eq:ode-main2} on $\Int(M)$, and the boundary equations
\[
F_-(a(\underline m),a'(\underline m))=0,
\qquad
F_+(a(\overline m),a'(\overline m))=0.
\]
Let $R_a(\theta,b)$ denote the unique maximizer of
$U^S(a(r),\theta)-c\phi(r-b)$
over $r\in M$. Then $(R_a,a)$ is a regular equilibrium.
\end{lemma}

\begin{proof}
\emph{Step 1 (Unique sender best replies).}
For a given $(\theta,b)$, define
\[
\Pi(r;\theta,b)=U^S(a(r),\theta)-c\phi(r-b).
\]
Because $a(M)\subseteq A^\dagger$, $0<a'\le p^\star$, and $a$ solves the ODE, the growth bound in Assumption~\ref{ass:globalODE} provides
\[
|a''(r)|=|G(r,a(r),a'(r))|\le K(1+(p^\star)^2)
\qquad\text{for all }r\in M.
\]
Hence
\[
\Pi_{rr}(r;\theta,b)
=
U^S_{11}(a(r),\theta)\,(a'(r))^2
+
U^S_1(a(r),\theta)\,a''(r)
-
c\phi''(r-b)
<0
\]
by the strict-concavity requirement in Assumption~\ref{ass:globalODE}. So $\Pi(\cdot;\theta,b)$ is strictly concave on the compact set $M$, and the maximizer $R_a(\theta,b)$ is unique.

\emph{Step 2 (Continuity, monotonicity, and full use).}
The objective $\Pi(r;\theta,b)$ is jointly continuous in $(r,b)$ and has strictly increasing differences in $(r,b)$ because
\[
\frac{\partial^2 \Pi}{\partial r\,\partial b}(r;\theta,b)=c\phi''(r-b)>0.
\]
Berge's theorem and Topkis's theorem therefore imply that, for each fixed $\theta$, the map $b\mapsto R_a(\theta,b)$ is continuous and nondecreasing.

Because $a\in C^3(M)$ and $M$ is compact, the map $r\mapsto U^S(a(r),\theta)$ is Lipschitz for each fixed $\theta$; let
\[
L_\theta=\sup_{r\in M}|U^S_1(a(r),\theta)a'(r)|<\infty.
\]
By Assumption~\ref{ass:invert-main}, $D_1(\overline m,b)=c\phi'(\overline m-b)\to-\infty$ as $b\to\infty$. Choose $b^+$ so large that $D_1(r,b)<-L_\theta$ for every $r\in M$ and every $b\ge b^+$. Then
\[
\Pi_r(r;\theta,b)=U^S_1(a(r),\theta)a'(r)-D_1(r,b)>0
\]
for all $r\in M$ and all $b\ge b^+$, so the unique maximizer is $R_a(\theta,b)=\overline m$. A symmetric argument shows $R_a(\theta,b)=\underline m$ for all sufficiently negative $b$. Since $b\mapsto R_a(\theta,b)$ is continuous and $\mathbb R$ is connected, its image is an interval containing both endpoints of $M$, hence
$R_a(\theta,\mathbb R)=M
\hspace{.1cm}\text{for every }\theta\in\Theta.$

\emph{Step 3 (Interior report density).}
Let $r\in\Int(M)$ and $\theta\in\Theta$. Since $R_a(\theta,\mathbb R)=M$, some anchor value induces report $r$. Because the optimum is interior, the first-order condition must hold:
\[
c\phi'(r-b)=U^S_1(a(r),\theta)a'(r).
\]
Since $\phi'$ is onto and strictly increasing, there is a unique solution,
\[
b_a(r,\theta)
=
r-\psi\!\left(\frac{U^S_1(a(r),\theta)a'(r)}{c}\right).
\]
Differentiating,
\[
\partial_r b_a(r,\theta)
=
1-\frac{\psi'(s)}{c}
\left[
U^S_{11}(a(r),\theta)(a'(r))^2
+
U^S_1(a(r),\theta)a''(r)
\right],
\]
where $s=U^S_1(a(r),\theta)a'(r)/c$. At $b=b_a(r,\theta)$, $\phi''(r-b)=1/\psi'(s)$. Combining this identity with the strict concavity inequality from Step 1 produces
\[
U^S_{11}(a(r),\theta)(a'(r))^2
+
U^S_1(a(r),\theta)a''(r)
<
\frac{c}{\psi'(s)},
\]
hence $\partial_r b_a(r,\theta)>0$. Thus $b_a(\cdot,\theta)$ is strictly increasing on $\Int(M)$, and its inverse is the on-path report rule $R_a(\theta,\cdot)$ on the interior. The change-of-variables formula therefore produces the conditional report density
\[
g_a(r|\theta)
=
q_\sigma(b_a(r,\theta)|\theta)\,\partial_r b_a(r,\theta)
\qquad\text{for }r\in\Int(M).
\]

\emph{Step 4 (Boundary report masses).}
At the lower endpoint,
\[
\Pi_r(\underline m;\theta,b)
=
U^S_1(a(\underline m),\theta)a'(\underline m)-c\phi'(\underline m-b).
\]
Because $\Pi(\cdot;\theta,b)$ is strictly concave, the report $\underline m$ is optimal if and only if this derivative is weakly negative, that is,
\[
b\le
\underline m-\psi\!\left(\frac{U^S_1(a(\underline m),\theta)a'(\underline m)}{c}\right)
=
b_-(a(\underline m),a'(\underline m),\theta).
\]
Hence
\[
\Pr(R_a(\theta,b)=\underline m|\theta)
=
H\!\left(\frac{b_-(a(\underline m),a'(\underline m),\theta)-b_0(\theta)}{\sigma}\right).
\]
The same argument at the upper endpoint produces
\[
\Pr(R_a(\theta,b)=\overline m|\theta)
=
1-
H\!\left(\frac{b_+(a(\overline m),a'(\overline m),\theta)-b_0(\theta)}{\sigma}\right).
\]

\emph{Step 5 (Receiver optimality and Bayes consistency).}
By Theorem~\ref{thm:ode-main}, the interior ODE~\eqref{eq:ode-main2} is equivalent to
\[
\int_\Theta U^R_1(a(r),\theta)\,f(\theta)\,g_a(r|\theta)\,d\theta=0
\qquad\text{for every }r\in\Int(M).
\]
By Steps 3 and 4, the interior density and the boundary masses used in this equation are the actual on-path objects generated by $R_a$. The boundary equations
\[
F_-(a(\underline m),a'(\underline m))=0,
\qquad
F_+(a(\overline m),a'(\overline m))=0
\]
are exactly the receiver first-order conditions at the two endpoints. Since $U^R(\cdot,\theta)$ is strictly concave in $a$, these first-order conditions characterize the receiver's unique best reply at every report in $M$. Thus $(R_a,a)$ is a pure-strategy equilibrium.

Finally, $a$ is continuous and strictly increasing, with $a\in C^3(M)$, and Step 1 provides a unique sender best reply for every $(\theta,b)$. Step 1 also yields strict second-order conditions at every interior optimum. Thus $(R_a,a)$ is a regular equilibrium.
\end{proof}

\begin{theorem}\label{thm:globalexist-main}
Under Assumptions~\ref{ass:PB-main} and~\ref{ass:globalODE}, there is a regular equilibrium on $M$.
\end{theorem}

\begin{proof}
The admissible set $\mathcal D$ is nonempty, closed, bounded, and convex in $C^1(M)$. To see non-emptiness, let
\[
a_0(r)=a_{\min}+\underline p^\star(r-\underline m).
\]
Because $\underline p^\star>0$, $\underline p^\star\le p^\star$, and Assumption~\ref{ass:globalODE} gives
\[
\underline p^\star\Delta m\le\delta_E\le\Delta a^\dagger,
\]
it follows that $a_0(M)\subseteq A^\dagger$ and $\underline p^\star\le a_0'(r)\le p^\star$ on $M$, so $a_0\in\mathcal D$. For $a\in\mathcal D$, define the operator $\mathcal T:\mathcal D\to C^1(M)$ by
\begin{align*}
(\mathcal T a)(r)
&=
u^\star(a'(\underline m))
+\frac{r-\underline m}{\Delta m}
\Bigl(
v^\star(a'(\overline m))-u^\star(a'(\underline m))
\Bigr) \\
&\quad
-\int_{\underline m}^{\overline m}
\mathcal G(r,s)\,
G(s,a(s),a'(s))\,ds,
\end{align*}
where $\mathcal G$ is the Green's function for the Dirichlet problem on $[\underline m,\overline m]$.

Because $(r,a(r),a'(r))\in\mathcal S$ for every $a\in\mathcal D$, the integrand is well defined. The standard Green's function bounds
\[
\sup_r\int_{\underline m}^{\overline m}|\partial_r\mathcal G(r,s)|\,ds=\frac{\Delta m}{2},
\qquad
\sup_r\int_{\underline m}^{\overline m}|\mathcal G(r,s)|\,ds=\frac{(\Delta m)^2}{8},
\]
control slopes and levels separately.

For the upper slope bound,
\[
\bigl\|(\mathcal T a)'\bigr\|_\infty
\le
\frac{\Delta a^\dagger}{\Delta m}
+\frac{\Delta m}{2}K\bigl(1+(p^\star)^2\bigr)
=
p^\star
\]
by the defining equation for $p^\star$ in Assumption~\ref{ass:globalODE}. For the lower slope bound, the affine boundary term has slope
\[
\frac{v^\star(a'(\overline m))-u^\star(a'(\underline m))}{\Delta m}
\ge
\frac{\delta_E}{\Delta m},
\]
while the Green's term contributes at most
$\frac{\Delta m}{2}K\bigl(1+(p^\star)^2\bigr)$
in absolute value. Hence
\[
(\mathcal T a)'(r)
\ge
\frac{\delta_E}{\Delta m}
-\frac{\Delta m}{2}K\bigl(1+(p^\star)^2\bigr)
\ge\underline p^\star>0
\]
by the lower-slope requirement in Assumption~\ref{ass:globalODE}. For the level bound,
\[
\sup_r\left|
\int_{\underline m}^{\overline m}
\mathcal G(r,s)\,G(s,a(s),a'(s))\,ds
\right|
\le
\frac{(\Delta m)^2}{8}K\bigl(1+(p^\star)^2\bigr),
\]
and the affine boundary term always lies between
$u^\star(a'(\underline m))$
 and
$v^\star(a'(\overline m))$, both of which belong to $A^\dagger$. Assumption~\ref{ass:globalODE} therefore yields $(\mathcal T a)(M)\subseteq A^\dagger$. Thus $\mathcal T(\mathcal D)\subseteq\mathcal D$, and in fact every image satisfies $\underline p^\star\le (\mathcal T a)' \le p^\star$.

Differentiating the Green representation twice,
$(\mathcal T a)''(r)=G(r,a(r),a'(r))$, so every image satisfies $\|(\mathcal T a)''\|_\infty\le K\bigl(1+(p^\star)^2\bigr)$. Hence $\mathcal T(\mathcal D)$ is bounded in $C^2(M)$ and therefore precompact in $C^1(M)$ by Arzel\`a-Ascoli. Continuity of $\mathcal T$ follows from continuity of $G$, $u^\star$, and $v^\star$ on compact sets. Schauder's fixed-point theorem therefore delivers $a^\star\in\mathcal D$ such that $\mathcal T a^\star=a^\star$. Since $a^\star\in\mathcal T(\mathcal D)$, it also satisfies $(a^\star)'(r)\ge \underline p^\star>0$ for all $r$, so it is strictly increasing.

Because $a^\star=\mathcal T a^\star$, the Green representation yields that $a^\star$ solves the interior ODE. The Green function vanishes at the two endpoints, so $a^\star(\underline m)=u^\star((a^\star)'(\underline m))$
 and 
$a^\star(\overline m)=v^\star((a^\star)'(\overline m))$. By the definition of $u^\star$ and $v^\star$, this is equivalent to the boundary equations
\[
F_-(a^\star(\underline m),(a^\star)'(\underline m))=F_+(a^\star(\overline m),(a^\star)'(\overline m))=0.
\]
Since $G$ is $C^1$ on the compact strip $\mathcal S$ and $(a^\star)''=G(r,a^\star,(a^\star)')$, a standard bootstrap upgrades $a^\star$ to $C^3(M)$. The image bound $a^\star(M)\subseteq A^\dagger$ and the slope bound $0<(a^\star)'(r)\le p^\star$ were established above. Let $R^\star(\theta,b)$ denote the unique maximizer of $U^S(a^\star(r),\theta)-c\phi(r-b)$ over $r\in M$. Lemma~\ref{lem:verification-compact} therefore applies to $a^\star$, and the pair $(R^\star,a^\star)$ is a regular equilibrium on $M$.
\end{proof}

\subsection{Unbounded message spaces}\label{app:proofs-unbounded}

The unbounded-message case is a continuation result from truncated boundary-value solutions. A common compact action interval on expanding message spaces precludes a uniform lower slope bound, since $a_n'(r)\ge\underline p>0$ on $[-n,n]$ with $a_n(M_n)\subseteq A^\dagger$ would force image width at least $2n\underline p\to\infty$, contradicting compactness of $A^\dagger$. The theorem therefore imposes positive lower slope bounds only on fixed compact report intervals.

\begin{theorem}\label{thm:exist-unbounded}
Under Assumptions~\ref{ass:receiver-main}--\ref{ass:invert-main} and~\ref{ass:PB-main}, suppose that for each $n\ge 1$ there exists a function $a_n\in C^3(M_n)$ on $M_n=[-n,n]$ such that:
\begin{enumerate}[label=(\roman*)]
\item $a_n$ satisfies the compact-message ODE~\eqref{eq:ode-main2} on $(-n,n)$ and the corresponding boundary equations at $\pm n$;
\item there exist a compact interval $A^\dagger\subset\Int(A)$, a constant $\overline p<\infty$, and constants $\underline{\mathcal A},K<\infty$ such that for every $n$, $a_n(M_n)\subseteq A^\dagger$,
$0<a_n'(r)\le \overline p$
 for all $r\in M_n$,
 and on the strip $\mathcal S_\infty=\R\times A^\dagger\times[0,\overline p]$, the ODE objects satisfy
\[
\mathcal A(r,u,p)\ge \underline{\mathcal A}>0,
\qquad
|G(r,u,p)|\le K,
\qquad
G\in C^1(\mathcal S_\infty);
\]
\item for every $j\ge 1$ there exists a constant $\underline p_j>0$ such that
\[
a_n'(r)\ge \underline p_j
\qquad\text{for all }r\in[-j,j]\text{ and all }n\ge j;
\]
\item a uniform sender-concavity bound
holds:\[
c\underline\kappa>
\sup_{(u,\theta)\in A^\dagger\times\Theta}
\left[
|U^S_{11}(u,\theta)|\,\overline p^{\,2}
+
|U^S_1(u,\theta)|\,K
\right].
\]
\end{enumerate}
Then there exists a regular equilibrium on $M=\R$.
\end{theorem}

\begin{proof}
\emph{Step 1: local compactness and diagonal extraction.}
By assumption, $\|a_n\|_\infty$ and $\|a_n'\|_\infty$ are uniformly bounded on every compact interval. Since $a_n''(r)=G(r,a_n(r),a_n'(r))$ and $|G|\le K$ on $\mathcal S_\infty$, $|a_n''(r)|\le K$
for all $r\in M_n$.

For $n\ge j$, the restrictions of $a_n$ to $[-j,j]$ are uniformly bounded in $C^2([-j,j])$, hence precompact in $C^1([-j,j])$ by Arzel\`a-Ascoli. A standard diagonal argument shows a subsequence, still indexed by $n$, and a limit $a^\star\in C^1(\R)$ such that $a_n\to a^\star
\hspace{.1cm}\text{in }C^1_{\mathrm{loc}}(\R).$
The bounds pass to the limit, so
\[
a^\star(\R)\subseteq A^\dagger,
\qquad
0<(a^\star)'(r)\le \overline p
\quad\text{for every }r\in\R.
\]
More precisely, on each fixed interval $[-j,j]$, $(a^\star)'(r)\ge \underline p_j>0.$

\emph{Step 2: the limit solves the ODE.}
On a compact interval $[-j,j]$ with $n\ge j$,
\[
a_n'(r)-a_n'(0)=\int_0^r G(s,a_n(s),a_n'(s))\,ds.
\]
Because $a_n\to a^\star$ in $C^1([-j,j])$ and $G$ is continuous on $\mathcal S_\infty$, dominated convergence yields
\[
(a^\star)'(r)-(a^\star)'(0)=\int_0^r G(s,a^\star(s),(a^\star)'(s))\,ds
\qquad\text{for }r\in[-j,j].
\]
Hence $a^\star\in C^2([-j,j])$ and
$(a^\star)''(r)=G(r,a^\star(r),(a^\star)'(r))$
 on $[-j,j]$. Since $j$ was arbitrary, $a^\star\in C^2(\R)$ and solves the ODE on all of $\R$. The condition $G\in C^1(\mathcal S_\infty)$ then upgrades $a^\star$ to $C^3(\R)$.

\emph{Step 3: unique sender best replies under the limit action rule.}
Consider
\[
\Pi^\star(r;\theta,b)=U^S(a^\star(r),\theta)-c\phi(r-b).
\]
Because $a^\star(\R)\subseteq A^\dagger$, the strategic term is bounded on $\R\times\Theta$. Assumption~\ref{ass:cost-main}(iv) therefore yields coercivity of $\Pi^\star$ in $r$, so an optimal report exists. Moreover,
\[
\frac{\partial^2 \Pi^\star}{\partial r^2}(r;\theta,b)
=
U^S_{11}(a^\star(r),\theta)\,((a^\star)'(r))^2
+
U^S_1(a^\star(r),\theta)\,(a^\star)''(r)
-
c\phi''(r-b)
<0
\]
by part~(iv) of the theorem and the bounds from Steps 1 and 2. Hence $\Pi^\star(\cdot;\theta,b)$ is strictly concave, so the sender's maximizer $R^\star(\theta,b)$ is unique.

\emph{Step 4: continuity, monotonicity, and full use of reports.}
The objective $\Pi^\star(r;\theta,b)$ has strictly increasing differences in $(r,b)$ because $c\phi''(r-b)>0$, so $b\mapsto R^\star(\theta,b)$ is nondecreasing by Topkis's theorem.

Let
\[
L_0=\sup_{(u,\theta)\in A^\dagger\times\Theta}|U^S(u,\theta)|<\infty.
\]
By Assumption~\ref{ass:cost-main}(iv), there exists $K_D<\infty$ such that
\[
|r-b|>K_D
\quad\Longrightarrow\quad
c\phi(r-b)-c\phi(0)>2L_0.
\]
So any report with $|r-b|>K_D$ is strictly dominated by reporting the anchor $b$, and therefore $|R^\star(\theta,b)-b|\le K_D$
for all $(\theta,b)$.

Let  $b_0\in\R$. For $b$ near $b_0$, every optimal report lies in the compact interval $[b_0-K_D-1,\; b_0+K_D+1]$. Berge's theorem applied on that compact interval therefore yields continuity of $b\mapsto R^\star(\theta,b)$ at $b_0$. Since $b_0$ was arbitrary, the map is continuous on $\R$.

The distortion bound also forces $R^\star(\theta,b)\to\pm\infty
\hspace{.1cm}\text{as }b\to\pm\infty$. The image of the connected set $\R$ under the continuous map $b\mapsto R^\star(\theta,b)$ is therefore a connected unbounded subset of $\R$, hence $R^\star(\theta,\R)=\R$
for every $\theta\in\Theta$.

\emph{Step 5: Bayes consistency on $\R$.}
Since every report is used on path, the inducing anchor must satisfy the interior first-order condition $c\phi'(r-b)=U^S_1(a^\star(r),\theta)(a^\star)'(r)$.
Its unique solution is
\[
b_{a^\star}(r,\theta)
=
r-\psi\!\left(\frac{U^S_1(a^\star(r),\theta)(a^\star)'(r)}{c}\right).
\]
Differentiating and using the strict concavity inequality from Step 3 produces
$\partial_r b_{a^\star}(r,\theta)>0$. So $b_{a^\star}(\cdot,\theta)$ is strictly increasing and inverts the on-path report rule. The conditional density of reports given $\theta$ is therefore
\[
g^\star(r|\theta)
=
q_\sigma(b_{a^\star}(r,\theta)|\theta)\,
\partial_r b_{a^\star}(r,\theta).
\]
Because $a^\star$ solves the ODE on $\R$, Theorem~\ref{thm:ode-main} returns
\[
\int_\Theta U^R_1(a^\star(r),\theta)\,f(\theta)\,g^\star(r|\theta)\,d\theta=0
\qquad\text{for every }r\in\R.
\]
Strict concavity of $U^R(\cdot,\theta)$ implies that $a^\star(r)$ is the receiver's unique best reply at every on-path report.

The pair $(R^\star,a^\star)$ is therefore a pure-strategy equilibrium. Since $a^\star$ is continuous, strictly increasing, and $C^3$ on $\R$, sender best replies are unique, and the interior second-order condition is strict by Step 3, this equilibrium is regular.
\end{proof}

\subsection{Least-cost selection}\label{app:supp-leastcost}

Second-order differential equations may admit multiple boundary-compatible regular equilibria.  A natural selection criterion is the expected report-anchor cost.  Under the compact-message sufficient conditions in Assumption~\ref{ass:globalODE}, let
\[
\mathcal E_\sigma=
\left\{
\begin{array}{l}
a\in C^3(M):\ a\text{ solves \eqref{eq:ode-main2} and the compact-message boundary equations},\\
a(M)\subseteq A^\dagger,\quad \underline p^\star\le a'\le p^\star
\end{array}
\right\}.
\]
For $a\in\mathcal E_\sigma$, let $R_a(\theta,b)$ denote the unique sender-optimal report induced by $a$, and define
\begin{equation}\label{eq:commcost-main}
\mathcal C_\sigma(a)=
\int_\Theta\int_\R D(R_a(\theta,b),b)q_\sigma(b|\theta)db\,f(\theta)d\theta.
\end{equation}

\begin{assumption}\label{ass:costint}
The integrability condition
\[
\int_\Theta\int_\R\sup_{r\in M}D(r,b)q_\sigma(b|\theta)db\,f(\theta)d\theta<\infty
\]
holds.
\end{assumption}

\begin{lemma}\label{lem:compactness-main}
Under Assumptions~\ref{ass:PB-main} and~\ref{ass:globalODE}, the set $\mathcal E_\sigma$ is compact in the $C^1(M)$ topology.
\end{lemma}

\begin{proof}[\textbf{Proof of Lemma~\ref{lem:compactness-main}}]
Every $a\in\mathcal E_\sigma$ satisfies $\|a'\|_\infty\le p^\star$ and $\|a''\|_\infty\le K(1+(p^\star)^2)$. Since $a''=G(r,a,a')$ and $G$ is $C^1$ on the compact strip $\mathcal S$, the Lipschitz constant of $a''$ is uniformly bounded. Arzel\`a-Ascoli therefore delivers precompactness in $C^2(M)$. The ODE, the boundary equations, and the constraints pass to limits, and the identity $a''=G(r,a,a')$ with $G\in C^1$ upgrades any $C^2$ limit to $C^3$. Thus $\mathcal E_\sigma$ is closed and hence compact.
\end{proof}

\begin{proposition}\label{prop:leastcost-main}
Under Assumptions~\ref{ass:PB-main},~\ref{ass:globalODE}, and~\ref{ass:costint}, the map $\mathcal C_\sigma$ is continuous on $\mathcal E_\sigma$ in the $C^1(M)$ topology.  The set
\[
\mathcal E_\sigma^\star=\arg\min_{a\in\mathcal E_\sigma}\mathcal C_\sigma(a)
\]
is nonempty and compact.
\end{proposition}

\begin{proof}[\textbf{Proof of Proposition~\ref{prop:leastcost-main}}]
The strict-concavity requirement in Assumption~\ref{ass:globalODE} makes the sender's report rule single valued. Berge's maximum theorem yields continuity of $R_a(\theta,b)$ in $a$. Dominated convergence, using Assumption~\ref{ass:costint}, then provides continuity of $\mathcal C_\sigma$. Compactness of $\mathcal E_\sigma$ and Weierstrass's theorem complete the proof.
\end{proof}

\subsection{General private shifters}\label{app:supp-shifters}

The no-holes result in Section~\ref{sec:noholes} was stated for an additive private anchor $b$ with translation-invariant cost. The conclusion holds for any private variable whose marginal effect on the sender's report problem satisfies single-crossing, regardless of whether that variable enters through the cost, through the strategic payoff, or both. Let $z$ denote a private shifter with support $Z(\theta)\subseteq\R$. Write the sender's strategic payoff as $\widetilde U^S(a,\theta,z)$ and the reporting cost as $\widetilde D(r,\theta,z)$; a regular equilibrium is defined as in Definition~\ref{def:regular-main} with $z$ replacing $b$.

\begin{assumption}\label{ass:shifter-supp}
For every $\theta$, the support $Z(\theta)$ is an interval. The map
\[
V(r,\theta,z)=\widetilde U^S(a(r),\theta,z)-\widetilde D(r,\theta,z)
\]
is jointly continuous and $C^2$ in $(r,z)$ on $\Int(M)\times\Int(Z(\theta))$, and satisfies the increasing-differences condition
\begin{equation}\label{eq:gen-singlecross}
\widetilde U^S_{13}(a(r),\theta,z)\,a'(r)-\widetilde D_{13}(r,\theta,z)>0
\qquad\text{for all }(r,\theta,z).
\end{equation}
For each $(\theta,z_0)$ and each regular equilibrium under consideration, optimal reports for shifter values in a neighborhood of $z_0$ lie in a compact subset of $M$.
\end{assumption}

Condition \eqref{eq:gen-singlecross} is the substantive requirement, and it is general enough to admit two distinct economic structures. In the additive-anchor environment of Assumption~\ref{ass:cost-main}, $\widetilde U^S$ does not depend on the private variable $z=b$, so $\widetilde U^S_{13}=0$, while $\widetilde D=c\phi(r-b)$ gives $\widetilde D_{13}=-c\phi''(r-b)<0$; condition \eqref{eq:gen-singlecross} reduces to $c\phi''(r-b)>0$, automatic under cost convexity. In the model of \citet{FischerVerrecchia2000}, the private variable $z$ is a manager's taste for the receiver's action, the strategic payoff is $\widetilde U^S(a,\theta,z)=z\cdot a$, and the cost $\widetilde D(r,\theta)=(r-\theta)^2$ is independent of $z$; here $\widetilde U^S_{13}=1$ and $\widetilde D_{13}=0$, so condition \eqref{eq:gen-singlecross} reduces to $a'(r)>0$, which holds in any informative regular equilibrium. Other applications include private information about the marginal cost of inflating a public claim, as in \citet{FrankelKartik2019}, and any private taste or technology shifter that enters either the strategic motive or the cost with the same single-crossing structure.

\begin{theorem}\label{thm:general-shifter}
Under Assumption~\ref{ass:shifter-supp}, in any regular equilibrium and for each fixed $\theta$, the map $z\mapsto R(\theta,z)$ is nondecreasing and continuous on $Z(\theta)$. Hence $R(\theta,Z(\theta))$ is an interval.
\end{theorem}

\begin{proof}
By \eqref{eq:gen-singlecross}, $V(\cdot,\theta,\cdot)$ has strictly increasing differences in $(r,z)$. Uniqueness of the sender's maximizer in a regular equilibrium and Topkis's theorem imply that $z\mapsto R(\theta,z)$ is nondecreasing. The compactness clause provides a neighborhood $U$ of $z_0$ and a compact $C\subseteq M$ containing all optimal reports for $z\in U$; Berge's maximum theorem applied on $C$ yields continuity at $z_0$. The image of the connected set $Z(\theta)$ under a continuous map is connected, hence an interval.
\end{proof}

\subsection{Deterministic anchor}\label{app:supp-det}

The limit case of a deterministic anchor can be pinned down from a simplified version of the ODE\ in the baseline model. Suppose in this subsection that $b=b_0(\theta)$ almost surely and $D(r,b)=c\phi(r-b)$ with $\phi'$ onto and $\phi''>0$.

\begin{proposition}\label{prop:deterministic-main}
  In any separating equilibrium, on any interval of states where the report rule $r(\theta)$ is $C^1$ with $r'(\theta)\neq0$,
\begin{equation}\label{eq:ode-main}
c\phi'(r(\theta)-b_0(\theta))=U^S_1(a^R(\theta),\theta)\frac{(a^R)'(\theta)}{r'(\theta)}.
\end{equation}
\end{proposition}

\begin{proof}[\textbf{Proof of Proposition~\ref{prop:deterministic-main}}]
If the anchor is deterministic and the equilibrium is separating, the receiver infers $\theta$ from the report and chooses $a^R(\theta)$. The sender solves
\[
\max_r\{U^S(a^R(\hat\theta(r)),\theta)-c\phi(r-b_0(\theta))\},
\]
where $\hat\theta(r)$ is the inverse report rule. The first-order condition at $r=r(\theta)$ is
\[
U^S_1(a^R(\theta),\theta)(a^R)'(\theta)\hat\theta'(r(\theta))
=
c\phi'(r(\theta)-b_0(\theta)).
\]
Since $\hat\theta'(r(\theta))=1/r'(\theta)$, the claim follows.
\end{proof}

\begin{corollary}\label{cor:det-least}
Let $T\in C^1(\Theta)$ be bounded below by a positive constant, with $b_0,\phi'\in C^1$, $\phi''>0$, and $\phi'$ onto. The boundary problem
\[
c\phi'(\rho(\theta)-b_0(\theta))\rho'(\theta)=T(\theta),
\qquad
\rho(\underline\theta)=b_0(\underline\theta)
\]
has a unique strictly increasing solution on $\Theta$, continuous at $\underline\theta$ and $C^1$ on $(\underline\theta,\overline\theta]$.
\end{corollary}

\begin{proof}
The ODE is singular at $\theta=\underline\theta$ because $\phi'(0)=0$. Working with the inverse $\vartheta=\rho^{-1}$ removes the singularity: $\vartheta$ satisfies the non-singular
\[
\vartheta'(r)=\frac{c\phi'(r-b_0(\vartheta(r)))}{T(\vartheta(r))},
\qquad
\vartheta(r_0)=\underline\theta,\quad r_0=b_0(\underline\theta).
\]
The right-hand side is locally Lipschitz in $\vartheta$ on $\Theta$ (using $T\in C^1$, $T>0$ on $\Theta$, and $b_0,\phi'\in C^1$), so Picard-Lindel\"of gives local existence and uniqueness. Setting $g(r)=r-b_0(\vartheta(r))$, $g(r_0)=0$ and $g'(r_0)=1$ since $\phi'(0)=0$; the same identity at any later return to zero would contradict $g'\le0$, so $g>0$ and $\vartheta'>0$ for $r>r_0$. The solution reaches $\vartheta=\overline\theta$ at finite $r_H$: if $\vartheta$ stayed below $\overline\theta-\varepsilon$ along $r_n\to\infty$, then $g(r_n)\to\infty$ would force $\vartheta'(r_n)\to\infty$, a contradiction. Inverting produces the unique strictly increasing $\rho$ on $\Theta$, $C^1$ on $(\underline\theta,\overline\theta]$.
\end{proof}

\subsection{Post-anchor labels}\label{app:supp-postanchor}

This subsection explains why the hybrid construction in Section~\ref{sec:hybrid} uses a pre-anchor, or anchor-independent, qualitative channel. Suppose instead that the sender observes $(\theta,b)$ before choosing both the cheap-talk label and the numerical report. The label set remains finite and discrete, i.e., no continuity restriction across cheap-talk labels is imposed.

For a label $j$, let $(R_j,a_j)$ denote a regular anchored continuation and define the sender's realized value from choosing label $j$ by
\[
V_j(\theta,b)=\sup_{r\in\R}\{U^S(a_j(r),\theta)-D(r,b)\}.
\]
Regularity yields a unique maximizer $R_j(\theta,b)$ and strict monotonicity of $a_j$. At regular optima,
\[
V_{j,b}(\theta,b)=-D_2(R_j(\theta,b),b),
\qquad
V_{j,\theta}(\theta,b)=U^S_2(a_j(R_j(\theta,b)),\theta).
\]

\begin{proposition}[No one-dimensional post-anchor partition]
\label{prop:no-postanchor-onedim}
Suppose Assumptions~\ref{ass:receiver-main}--\ref{ass:cost-main} hold, $M=\R$, $B(\theta)=\R$ for every $\theta$, and the anchor density is strictly positive on $\R$. In the post-anchor protocol in which the sender observes $(\theta,b)$ before choosing both a cheap-talk label and a numerical report, the following two one-dimensional partitions cannot be sustained by pointwise label incentive compatibility.

\begin{enumerate}[label=(\roman*)]
\item There is no nontrivial fixed state-cell partition with adjacent ordered state cells and regular within-cell anchored continuations.

\item There is no nontrivial ordered anchor-cell partition with adjacent anchor cells and adjacent label action rules satisfying $a_k(r)<a_{k+1}(r)$ for every $r\in\R$.
\end{enumerate}
\end{proposition}

\begin{proof}
For part (i), consider an adjacent state cutoff $t$ between labels $k$ and $k+1$. Let $\Delta(\theta,b)=V_{k+1}(\theta,b)-V_k(\theta,b)$. A fixed state-cell rule requires $\Delta(\theta,b)\le0$ for $\theta<t$ and $\Delta(\theta,b)\ge0$ for $\theta>t$, for every anchor $b$. Hence $\Delta(t,b)=0$ for every $b$. Differentiating this identity in $b$, $D_2(R_k(t,b),b)=D_2(R_{k+1}(t,b),b)$.
Since $D_{12}<0$, the map $r\mapsto D_2(r,b)$ is strictly decreasing. Thus $R_k(t,b)=R_{k+1}(t,b)$ for every $b$. Value equality at the cutoff then yields
\[
U^S(a_k(R_k(t,b)),t)=U^S(a_{k+1}(R_k(t,b)),t)
\qquad\text{for every }b.
\]
By Corollary~\ref{cor:fulluse-unbounded}, applied inside label $k$, $R_k(t,\R)=\R$. Therefore
\[
U^S(a_k(r),t)=U^S(a_{k+1}(r),t)
\qquad\text{for every }r\in\R.
\]
The same local-language argument as in Proposition~\ref{prop:hybrid-local-language} implies $a_k(r)<a_{k+1}(r)$ for every $r$. Since $U^S(\cdot,t)$ is strictly concave, ordered equal-payoff action pairs, when they exist, must move in opposite directions: the higher action on a given indifference curve is a strictly decreasing function of the lower action. This is incompatible with both $a_k$ and $a_{k+1}$ being strictly increasing in $r$. Thus a fixed state-cell partition cannot satisfy pointwise post-anchor IC.

For part (ii), consider an adjacent anchor cutoff $\beta$ between labels $k$ and $k+1$, where label $k$ is chosen below $\beta$ and label $k+1$ above $\beta$. Pointwise label IC requires $\Delta(\theta,b)\le0$ for $b<\beta$ and $\Delta(\theta,b)\ge0$ for $b>\beta$, for every $\theta$. Thus $\Delta(\theta,\beta)=0$ for every $\theta$, and the crossing condition implies $\Delta_b(\theta,\beta)\ge0$.

Differentiating $\Delta(\theta,\beta)=0$ in $\theta$ produces
\[
U^S_2(a_{k+1}(R_{k+1}(\theta,\beta)),\theta)
=
U^S_2(a_k(R_k(\theta,\beta)),\theta).
\]
Because $U^S_{12}>0$, $U^S_2(\cdot,\theta)$ is strictly increasing. Hence the two labels induce the same action at the anchor boundary: $a_{k+1}(R_{k+1}(\theta,\beta))=a_k(R_k(\theta,\beta))$.
The ordered-label condition $a_k(r)<a_{k+1}(r)$ for every $r$, together with strict monotonicity of the action rules, implies $R_{k+1}(\theta,\beta)<R_k(\theta,\beta)$.

Finally,
\[
\Delta_b(\theta,\beta)
=
-D_2(R_{k+1}(\theta,\beta),\beta)
+
D_2(R_k(\theta,\beta),\beta).
\]
Since $r\mapsto D_2(r,\beta)$ is strictly decreasing and $R_{k+1}(\theta,\beta)<R_k(\theta,\beta)$, the above equation is strictly negative. This contradicts $\Delta_b(\theta,\beta)\ge0$. Hence an ordered anchor-cell partition cannot satisfy pointwise post-anchor IC.
\end{proof}

Proposition~\ref{prop:no-postanchor-onedim} rules out the two natural one-dimensional analogues of the hybrid partition when the label is freely chosen after $(\theta,b)$ is observed. It does not rule out unordered menus or equilibria whose label cells are genuinely two-dimensional. In particular, post-anchor equilibria, if they exist, have boundaries of the form $\theta=t_k(b)$ rather than fixed state cutoffs or fixed anchor cutoffs.

\end{document}